\documentclass[3p,preprint]{elsarticle}
\usepackage[colorlinks,urlcolor=black]{hyperref}
\usepackage{amsthm}
\usepackage{palatino}

\usepackage{fancybox}
\newlength{\mylength}
\newenvironment{frameqn}%
{\setlength{\fboxsep}{6pt}
\setlength{\mylength}{\linewidth}%
\addtolength{\mylength}{-2\fboxsep}%
\addtolength{\mylength}{-2\fboxrule}%
\Sbox
\minipage{\mylength}%
\setlength{\abovedisplayskip}{0pt}%
\setlength{\belowdisplayskip}{0pt}%
$$}%
{$$\endminipage\endSbox
\[\fbox{\TheSbox}\]}

\newenvironment{frametxt}%
{\setlength{\fboxsep}{5pt}
\setlength{\mylength}{\linewidth}%
\addtolength{\mylength}{-2\fboxsep}%
\addtolength{\mylength}{-2\fboxrule}%
\Sbox
\minipage{\mylength}%
\setlength{\abovedisplayskip}{0pt}%
\setlength{\belowdisplayskip}{0pt}%
}%
{\endminipage\endSbox
\[\fbox{\TheSbox}\]}

\usepackage{soul}
\usepackage[usenames]{color}

\definecolor{apricot}{rgb}{1.0,.95,.8}
\definecolor{plum}{rgb}{1.0,.8,.95}
\newcommand\greymath[1]{\colorbox{plum}{\ensuremath{#1}}}
\long\def\greybox#1{%
    \newbox\contentbox%
    \newbox\bkgdbox%
    \setbox\contentbox\hbox to \hsize{%
        \vtop{
            \kern\columnsep
            \hbox to \hsize{%
                \kern\columnsep%
                \advance\hsize by -2\columnsep%
                \setlength{\textwidth}{\hsize}%
                \vbox{
                    \parskip=\baselineskip
                    \parindent=0bp
                    #1
                }%
                \kern\columnsep%
            }%
            \kern\columnsep%
        }%
    }%
    \setbox\bkgdbox\vbox{
        \pdfliteral{0.85 0.85 0.85 rg}
        \hrule width  \wd\contentbox %
               height \ht\contentbox %
               depth  \dp\contentbox
        \pdfliteral{0 0 0 rg}
    }%
    \wd\bkgdbox=0bp%
    \vbox{\hbox to \hsize{\box\bkgdbox\box\contentbox}}%
    \vskip\baselineskip%
}

\usepackage{float}
\floatstyle{boxed} \restylefloat{figure}

 \renewenvironment{thebibliography}[1]{%
   \begin{oldthebibliography}{#1}%
     \setlength{\parskip}{0ex}%
     \setlength{\itemsep}{0ex}%
     \fontsize{9.5}{11.0}
     \selectfont
 }%
 {%
   \end{oldthebibliography}%
 }

\usepackage[english]{babel}
\usepackage{graphicx}
\usepackage{wrapfig}
\usepackage{amsmath}
\usepackage{amsfonts}
\usepackage{amssymb}
\usepackage{prooftree}
\usepackage{calc}
\usepackage{xspace}
\usepackage{pstricks}
\usepackage{stmaryrd}

\definecolor{highlight}{rgb}{0.1, 0.0, 0.1}

\newcommand\shs{\mathrel{\Delta}}
\newcommand\qinv{{{}^{\text{-}E}}}
\newcommand\liff{\Leftrightarrow}
\newcommand{\abeq}{\mathrel{=_\tinys{\alpha\beta}}}

\newcommand{\interdelta}[1]{\llbracket #1 \rrbracket_\Delta}
\newcommand{\inter}[1]{\llbracket #1 \rrbracket}

\newcommand\simpto[1]{\stackrel{\scriptstyle #1}{\Longrightarrow}}
\newcommand\simptostar[1]{\stackrel{\scriptstyle #1}{{\Longrightarrow}^{\hspace{-.25em}{{{}_*}}}}}

\newcommand\mmath[1]{\ensuremath{#1}\xspace}

\newcommand\smallss[2]{{{}^{{\scalebox{.65}{\hspace{-.1em}$#1$}}}_{{\scalebox{.65}{\hspace{-.1em}$#2$}}}}}
\newcommand\tinys[1]{{{}_{{\scalebox{.5}{$#1$}}}}}
\newcommand\smalls[1]{{{}_{{\scalebox{.65}{\hspace{-.1em}$#1$}}}}}
\newcommand\rhoinc{{\rho\smallss{\mathcal V}{\f{Inc}}}}
\newcommand\Vinc{{{\mathcal V'}\smallss{\mathcal V}{\f{Inc}}}}

\newtheorem{thrm}{Theorem}[section]
\newtheorem{lemm}[thrm]{Lemma}
\newtheorem{prop}[thrm]{Proposition}
\newtheorem{corr}[thrm]{Corollary}
\newtheorem{_nttn}[thrm]{Notation}

\newtheorem{_defn}[thrm]{Definition}
\newenvironment{defn}{\begin{_defn}\normalfont}{\end{_defn}}

\newtheorem{_rmrk}[thrm]{Remark}
\newenvironment{rmrk}{\begin{_rmrk}\normalfont}{\end{_rmrk}}
\def\qedhere{}

\newenvironment{calcenv}
  {\begin{displaymath}\begin{array}{l@{\ }c@{\ }l@{\quad}l}}
  {\end{array}\end{displaymath}}

\newenvironment{squishitem}
{\begin{list}
  {$\bullet$\hspace{.2em}}
  {\setlength{\parsep}{0pt}
   \setlength{\itemsep}{.5pt}
   \setlength{\topsep}{.5pt}
   \setlength{\partopsep}{0pt}
   \setlength{\leftmargin}{2em}
   \setlength{\labelwidth}{3em}
   \setlength{\labelsep}{.5em}
  }
}
{\end{list}}

\newenvironment{squishlist}
  {\begin{list}
    {$-$\hspace{.2em}}
    {\setlength{\parsep}{0pt}
     \setlength{\itemsep}{0pt}
     \setlength{\topsep}{0pt}
     \setlength{\partopsep}{0pt}
     \setlength{\leftmargin}{0.5em}
     \setlength{\labelwidth}{0em}
     \setlength{\labelsep}{0em}
    }
  }
  {\end{list}}

\newcounter{edefinitioncount}

\newcommand{\f}[1]{\ensuremath{\mathit{#1}}}

\newcommand{\fcomp}{\ensuremath{%
  \hspace{0.08em}\mbox{\small\ensuremath{\circ}}\hspace{0.08em}}}

\newcommand\lam[1]{\lambda #1.}

\newlength{\insidewd}

\newlength{\stackht}
\newcommand{\deffont}[1]{\textbf{#1}}

\newcommand{\rulefont}[1]{\mmath{(\mathbf{#1})}}

\newcommand{\act}{\cdot}

\newcommand{\id}{\f{id}}
\newcommand{\Id}{\f{id}}

\newcommand{\cent}{\vdash}

\newcommand{\sm}{\ssm}
\newcommand{\ssm}{{{:}{=}}}

\newcommand{\aeq}{\mathrel{{=}\smalls{\alpha}}}
\newcommand{\ueq}{\stackrel{\hspace{-0.025em}\scalebox{.55}{?}}{=}}

\newcommand{\tf}[1]{\mathsf{#1}}

\newcommand{\model}[1]{\ensuremath{\llbracket #1 \rrbracket}}

\newcommand{\mone}{{\text{-}1}}

\newcommand\Exists[1]{\exists #1.}

\journal{Version of record published in the Logic Journal of the IGPL, Volume 18, Issue 6, December 2010, Pages 769–822, \href{https://doi.org/10.1093/jigpal/jzq006}{DOI: 10.1093/jigpal/jzq006}}

\begin{document}

\title{Permissive nominal terms and their unification: \\ an infinite, co-infinite approach to nominal techniques\tnoteref{t1}\tnoteref{vor}}
\tnotetext[t1]{Thanks to the anonymous referees.}
\tnotetext[vor]{This is the authors' accepted manuscript, shared as per the journal's \href{https://academic.oup.com/pages/self_archiving_policy_b}{\emph{self-archiving policy}} (\href{https://web.archive.org/web/20231125172128/https://academic.oup.com/pages/self_archiving_policy_b}{permalink}).  Details of the version of record are below.}

\author[]{\href{http://www.lix.polytechnique.fr/~dowek}{Gilles Dowek}}
\author[]{\href{http://www.gabbay.org.uk}{Murdoch J. Gabbay}}
\author[]{\href{http://www.macs.hw.ac.uk/~dpm8}{Dominic P. Mulligan}}


\begin{abstract}
Nominal terms extend first-order terms with binding.
They lack some properties of first- and higher-order terms: 
Terms must be reasoned about in a context of `freshness assumptions';
it is not always possible to `choose a fresh variable symbol' for a nominal term;
it is not always possible to `$\alpha$-convert a bound variable symbol' or to `quotient by $\alpha$-equivalence'; the notion of unifier is not based just on substitution. 

\emph{Permissive} nominal terms closely resemble nominal terms but they recover these properties, and in particular the `always fresh' and `always rename' properties.
In the permissive world, freshness contexts are elided, equality is fixed, and the notion of unifier is based on substitution alone rather than on nominal terms' notion of unification based on substitution plus extra freshness conditions. 

We prove that expressivity is not lost moving to the permissive case and provide an injection of nominal terms unification problems and their solutions into permissive nominal terms problems and their solutions.

We investigate the relation between permissive nominal unification and higher-order pattern unification.
We show how to translate permissive nominal unification problems and solutions in a sound, complete, and optimal manner, in suitable senses which we make formal.
\end{abstract}
\maketitle


\tableofcontents

\section{Introduction}

\subsection{About nominal terms and permissive nominal terms}

Many formal languages feature names and binding: examples include quantification, $\lambda$-abstraction, sets comprehension $\{x\mid \phi(x)\}$, and process-calculi name-hiding.
Binding is ubiquitous, because variables are there to be bound or substituted.

In contrast, variables cannot be bound in first-order terms: In first-order logic, variables are bound in propositions by quantifiers and not at all in terms;
first-order rewriting does not allow binding as it is based on first-order terms;
and, many programming languages and proof systems allow datatypes of terms, but only of first-order terms.

This motivates logics where variables can be bound by any function or predicate symbol \cite{dowek:binl}, extensions of rewriting on terms with binders \cite{KlopJW:comrsis,NipkowT:higors,gabbay:nomr-jv}, and programming languages and proof systems allowing datatypes with binders \cite{Pfenning:hoas,miller:logpll,gabbay:frepbm,harper:fradl,paulson:isanst} and more generally, definitions of a notion of term where variables may be bound.

In particular, this motivates \emph{nominal terms} \cite{gabbay:nomu-jv}.
They are designed to directly represent informal schematic specifications.
For example:
\\[1.5ex]
\noindent
\begin{tabular}{ll}
Informal equality:
& 
If $y \not\in \f{fv}(t)$ then $\lam{x}t$ is $\alpha$-equivalent to $\lam{y}[y/x]t$.
\\
Nominal terms: 
&
$b\#X\cent \lambda [a]X = \lambda [b](b\ a)\act X$
\\
Permissive:
&
$\lambda [a]X^S = \lambda [b](b\ a)\act X^S$, where $b\not\in S$
\\[1ex]
Informal unification:
&
Which $t$ and $u$ make $\lam{x}\lam{y}(y\,t)$ equal to $\lam{x}\lam{x}(x\,u)$?
\\
Nominal unification:
&
$\varnothing\cent\lambda [a]\lambda [b] (b\,X) \ueq \lambda [a]\lambda [a] (a\,Y)$
\\
Permissive:
&
$\lambda [a]\lambda [b] (b\,X^S) \ueq \lambda [a]\lambda [a] (a\,Y^S)$ (here, $a\in S$, $b\in S$) 
\\
\end{tabular}
\\[1.5ex]
The first example is part of the usual specification of $\alpha$-equivalence.
The second example is from \cite{gabbay:nomu-jv}.
Definitions of nominal and permissive nominal terms will follow.

Nominal terms have been explored in logic-programming \cite{cheney:alppl}, rewriting \cite{gabbay:nomr-jv}, logic \cite{gabbay:oneaah-jv,gabbay:nomuae}, and elsewhere.

An intuition of the translation from informal specification to (permissive) nominal terms is as follows:
\begin{itemize}
\item
Object-level variable symbols $x$ and $y$ correspond to \emph{atoms} $a$ and $b$ (Definition~\ref{defn.atoms}).
\item
Meta-level variables $t$ and $u$ correspond to \emph{unknowns} $X$.

Substitution of unknowns is capturing, which models the effect of writing ``set $t$ to be $x$ in $\lam{x}t$; we get $\lam{x}x$''. 
Similarly, $(\lambda [a]X)[X\ssm a]\equiv [a]a$.

Atoms and unknowns are two distinct \emph{levels} of variable, just as object- and meta-level are two levels.
\item
Conditions like $x\not\in\f{fv}(t)$ correspond to \emph{freshness} conditions $b\#X$ in nominal terms.

$b\#X$ is a restriction, which is enforced on the notion of substitution, on what $X$ may be instantiated to.
$b\#X$ is a promise to instantiate $X$ only to terms for which $b$ is fresh (see the conditions `$\nabla'\cent \theta(\nabla)$' in Lemma 2.14, and `$\nabla\cent a\#\theta(t)$' in Definition 3.1 of \cite{gabbay:nomu-jv}).

\item
Object-level renaming $[y/x]$ is modelled by \emph{swapping} $(b\ a)$ (which maps $b$ to $a$, $a$ to $b$, and all other atoms $c$ to themselves; Definition~\ref{def.permutation}).

Swappings are bijective on atoms and invertible, unlike renamings, and for this reason they are \emph{naturally capture-avoiding}.\footnote{For example, $(b\ a)\act [a](b\,a) \equiv [b]a$ (Definition~\ref{defn.nominal.permutation}) whereas $(\lam{x}(y\,x))[x/y]$ is equal to $\lam{x'}(x\,x')$ where $x'$ is chosen fresh; with the permutation, the capture-avoidance is automatic, with the renaming we have to deliberately rename.  This nice turn of phrase is due to Cheney, as far as we know.}
This gives them attractive mathematical properties which can be exploited in the mathematical theory.
Using them to manage renaming in $\alpha$-equivalence is a conspicuous feature of nominal techniques and follows \cite{gabbay:newaas-jv,gabbay:nomu-jv}.

The reader can safely think of $[y/x]$ as corresponding to $(b\ a)$ so long as $b$ ($y$) is sufficiently fresh.
It is notable that freshness side-conditions in informal practice seem to appear exactly in those places in which they would be necessary for this to work. 
\item
$\lambda$ and application are term-formers (function-symbols).
In the examples above we write application infix, as is standard.
\end{itemize}

Permissive nominal terms differ in two ways from nominal terms:
\begin{itemize}
\item
\emph{Difference 1.}\quad
Freshness information is fixed once and for all, and associated to unknowns in the manner of a sorting or typing annotation.

For example: $b\#X\cent X$ is a nominal term-in-freshness context $X$ whose freshness context insists that $b$ not occur free in $X$.  
A corresponding permissive nominal term is $X^S$ where $S$ is a set of atoms we call the \emph{permission set} of $X$, and $b\not\in S$.
\item
\emph{Difference 2.}\quad
Permission sets are sets of atoms that are both infinite and co-infinite (see Remark~\ref{rmrk.co-infinite}; $S\subseteq \mathbb A$ is co-infinite  when $\mathbb A\setminus S$ is infinite).
\end{itemize}
These combine with consequences which we now discuss. 

\subsection{Difference 1: about fixing freshness contexts}

Difference~1 on its own is mostly a matter of presentation.
It is mentioned already in \cite[Remark~2.6]{gabbay:nomu-jv} (there, we would obtain `permission sets' $A$ such that $A$ is infinite and $\mathbb A\setminus A$ is finite).
It is still worth reflecting briefly on what that difference in presentation means, for some example results.
Part~2 of Lemma~2.14 from \cite{gabbay:nomu-jv} reads as follows:
\begin{quote}
If $\nabla'\cent \sigma(\nabla)$ and $\nabla\cent a\#t$, then $\nabla\cent a\#\sigma(t)$.
\end{quote}
The corresponding result in this paper (Lemma~\ref{lemm.substitution}) reads as follows:
\begin{quote}
$\f{fa}(r\theta)\subseteq\f{fa}(r)$.
\end{quote}
Parts~1 to 3 of Lemma~2.7 from \cite{gabbay:nomu-jv} read as follows:
\begin{quote}
\begin{enumerate}
\item
If $\nabla\cent a\#\pi\act t$ then $\nabla\cent \pi^\mone\act a\#t$.
\item
If $\nabla\cent \pi\act a\#t$ then $\nabla\cent a\#\pi^\mone\act t$.
\item
If $\nabla\cent a\#t$ then $\nabla\cent \pi\act a\#\pi\act t$.
\end{enumerate}
\end{quote}
These three statements correspond to one result in this paper (Lemma~\ref{lemm.pi.ftma}) which reads as follows:
\begin{quote}
$\f{fa}(\pi\act r)=\pi\act \f{fa}(r)$.
\end{quote}
To us, the `permissive' versions seem shorter and clearer.

In fact there is a little more to Difference~1 than presentation.
If freshness contexts are fixed then some `consistency' properties of the same term across different freshness contexts become irrelevant. 
For instance consider the \emph{weakening} and \emph{strengthening} results for freshness contexts like those in Subsection~3.2 of \cite{gabbay:curhif-jv} or Subsection~3.3.2 of \cite{gabbay:nomuae};
if we had used permissive nominal terms, then these results would have been irrelevant and could have been omitted. 

\subsection{Difference 2: about choosing fresh atoms}

The reader used to nominal terms can think of Difference~2 as allowing and requiring infinite freshness contexts.\footnote{An similar idea is behind the notion of freshness contexts with \emph{sufficient freshnesses} in \cite[Subsection~3.1]{gabbay:newcc}.}
Together with Difference~1, this gives us two key properties which are absent in \cite{gabbay:nomu-jv}: 
\begin{itemize}
\item
An infinite supply of fresh atoms is guaranteed for every term.  There is no need to change a freshness context to obtain fresh atoms, because they are already there. 
\item
Terms may be directly quotiented by $\alpha$-equivalence. 
They also become susceptible to the nominal inductive reasoning principles of \cite{gabbay:newaas-jv}.
\end{itemize}
This matters.
The nominal terms as presented in \cite{gabbay:nomu-jv} cannot be quotiented by $\alpha$-equivalence, and cannot be subjected to the nominal inductive principles of nominal abstract syntax \cite{gabbay:newaas-jv}. 
In this respect, they are less tractable than first-order or higher-order terms, which can be.

True enough, it is always possible to `add a fresh atom to the freshness context' in nominal terms and rename, but this causes similar difficulties to those encountered in $\alpha$-renaming on name-carrying syntax. 

This is potentially a problem for the viability of any reasoning on, programming on, or extensions of nominal terms.
We have experienced this ourselves in \cite{gabbay:curhif-jv,gabbay:twollc}.

It becomes possible to argue against nominal terms because they appear to be harder to manipulate than the language they describe, and in particular, because they reintroduce the problem of $\alpha$-equivalence which nominal abstract syntax was originally designed to solve.

The existence of permissive nominal terms demonstrates that in fact, the apparent problem is just an artefact of the way matters were set up in \cite{gabbay:nomu-jv}.
We \emph{can} quotient permissive nominal terms by $\alpha$-equivalence.
In fact we can also use nominal inductive reasoning principles to reason on permissive nominal terms \cite{gabbay:semnt-ea,gabbay:pernas}.

Turning to the issue of fresh atoms; in the world of finite freshness contexts based on \cite{gabbay:nomu-jv}, we may always extend freshness contexts if we want more fresh atoms. 
This is made formal for example by \rulefont{fr} in \cite[Figure~2]{gabbay:nomuae} and \rulefont{Tfr} in \cite[Figure~1]{gabbay:curhif-jv}.

We can observe the effects of this solution to the problem of generating fresh atoms, by considering the mathematics in \cite{gabbay:curhif-jv}.
\rulefont{Tfr} is not a syntax-directed rule and this complicates case-analysis on derivations.  Subsection~5.1 of \cite{gabbay:curhif-jv} is devoted to controlling this issue.
Further, since freshness contexts may change during a proof, it is necessary to account for these changes in statements and proofs of results; the results in Subsection~5.3 of \cite{gabbay:curhif-jv} therefore contain an existential quantification over `sufficiently freshened' freshness contexts.

In fact we should recognise what happens in \cite{gabbay:curhif-jv} because it is typical of something we have seen before.
A `state' of `generated atoms' is being carried by the freshness context and threaded through all the proofs.
This is a situation familiar from implementing languages with binding; we may need to keep track of the `generated fresh atoms' so that when we generate a new one we can easily pick an `even fresher atom'.
We thread through the program a `state' of `generated atoms' (see for instance the `fresh monad' in Cheney's FreshLib~\cite{cheney:scrynf}).
This is undesirable if it can be avoided, because it interferes with the state-free style of functional programming, forces us to program more sequentially, and it complicates code.\footnote{Another layer of complexity arises. Some natural definitions rely on a sufficiently large supply of fresh atoms being available.  Even if we can always extend the freshness context, for a \emph{fixed} freshness context there may be inputs for which the function cannot be defined, and so is partial.  This happens for example to the \emph{canonical form} function in Subsection~5.3 of \cite{gabbay:curhif-jv}.}
This is what motivated FreshML \cite{gabbay:metpbn,gabbay:frepbm}.

The paper \cite{gabbay:curhif-jv} is not an isolated case.
Similar issues have arisen in proofs on the two-level lambda-calculus \cite{gabbay:twollc} and in other work \cite{gabbay:oneaah-jv,gabbay:capasn-jv,gabbay:nomalc}.
Furthermore, the `fresh atom' issue is not restricted to proofs:

A prototype implementation of permissive nominal terms and their unification by the third author~\cite{mulligan:imppnt} has a direct mechanism for generating a fresh name for any permissive term.
Even in a pure functional language like Haskell, no monadic programming is needed just to generate fresh names.

Calv\`{e}s recently wrote the Haskell Nominal Toolkit (HNT)~\cite{calves:hasnt}, which provides efficient implementations of several algorithms on nominal terms.
The HNT provides an elegant programming interface for manipulating nominal terms, and any efficient implementation of permissive nominal terms would also likely expose a similar API to the programmer.

However, Calv\`{e}s' underlying model is the nominal term and this is reflected in the types of programs written with the HNT.
Whereas core functions in the implementation of permissive nominal terms (e.g. unification and alpha-equivalence checks) are pure, the types of their HNT counterparts are heavily monadic, and explicit freshness contexts must also be passed around.

As an example, we compare the types (1) of the alpha-equivalence check function from the HNT and (2) of the implementation of permissive nominal terms:\footnote{See the documentation for module \texttt{Nominal.Matching} at \url{http://www.dcs.kcl.ac.uk/pg/calves/hnt/doc/Nominal-Matching.html} and module \texttt{Terms.Terms} in the permissive nominal terms implementation source code, available at \url{http://www.macs.hw.ac.uk/~dpm8/permissive/}}

\begin{displaymath}
\begin{array}{c@{\qquad}lcl}
(1) & \mathtt{alpha'check} & :: & (\mathtt{Show}\ t,\ \mathtt{Eq}\ t,\ \mathtt{Ord}\ a,\ \mathtt{Ord} v) \Rightarrow \mathtt{FrsCtxt}\ v\ a \rightarrow \\
&             &    & \mathtt{Term}\ a\ t\ v \rightarrow \mathtt{Term}\ a\ t\ v \rightarrow \\
&             &    & \mathtt{CS}\ r\ (\mathtt{ExtB}\ l\ e\ (\mathtt{ErrorT}\ [\mathtt{Char}]))\ m\ () \\[1.5ex]
(2) & \mathtt{aeq} & :: & (\mathtt{Eq}\ a, \mathtt{Permissive}\ b) \Rightarrow \mathtt{Term}\ a\ b \rightarrow \mathtt{Term}\ a\ b\ \rightarrow \mathtt{Bool}
\end{array}
\end{displaymath}

To us, the `permissive' type seems shorter and clearer.
An implementation of the HNT, based on permissive nominal terms, would likely present less daunting types to the programmer, and removing the burden of handling freshness contexts would likely make code shorter and neater.

\subsection{Permissive nominal terms in this paper}

Permissive nominal terms are designed to deal with these issues.  
Their theory of $\alpha$-equivalence and `fresh atom of' restore the good features of nominal abstract syntax without losing any expressivity or computational properties.

In fact, permissive nominal terms are `best possible' in a certain sense:
in \cite{gabbay:semnt-ea,gabbay:pernas} we demonstrate how the atoms-abstraction construction from \cite{gabbay:newaas-jv} can be applied directly to permissive nominal terms syntax.
That is, in \cite{gabbay:semnt-ea,gabbay:pernas} we show that it is possible to take $[a]r$ to be \emph{literally} the Gabbay-Pitts atoms-abstraction of $a$ in the set that is the abstract syntax tree $r$, so that $\alpha$-equivalence is \emph{literal} identity, and the definition of permissive nominal terms syntax (Definition~\ref{defn.terms}) becomes a nominal abstract syntax style inductive datatype of syntax-with-binding. 

In this paper we introduce permissive nominal terms.
We study their theory of unification.
Given any new syntax, it is important to connect it to existing denotations.
Therefore, we relate permissive nominal terms and their unification to nominal unification, and to unification of higher-order patterns.
Denotations in nominal sets will be the topic of a separate manuscript \cite{gabbay:pernas} (some elements are also in \cite{gabbay:semnt-ea}).

The contributions are as follows:
\begin{squishitem}
\item
We introduce permissive nominal terms Definitions~\ref{defn.terms}, along with theories of $\alpha$-equivalence and freshness with the `always fresh' and `always rename' properties discussed above (Corollaries~\ref{corr.always.fresh} and~\ref{corr.always.rename}).
\item
It may look like permissive nominal terms are infinite, because the $S$ in $X^S$ in Definition~\ref{defn.terms} is an infinite set.
In Remark~\ref{rmrk.harmless} we mention that they are just as `finite', and computable-upon, as nominal terms.
\item
We develop a notion of unification (Definition~\ref{defn.unification.algorithm}).
We use a simplified notion of unifier (Definition~\ref{defn.unif.sol}) which is more like the notion of unifier from first- and higher-order unification in that it is based just on a substitution, as compared to the notion of unifier used in \cite{gabbay:nomu-jv}, which is not.
\item
We make precise the connection between nominal terms and permissive nominal terms by translating nominal terms, nominal unification problems, and their solutions, to the permissive context (Definition~\ref{defn.interpretation.nominal}). 
We verify that no expressivity is lost (Theorem~\ref{thrm.no.missing}).
\item
We connect permissive nominal terms and higher order patterns \cite{miller:uniump,miller:logpll} (Definition~\ref{defn.nominal.trans.g}) by translating permissive nominal terms, unification problems, and their solutions, to a very general notion of untyped pattern unification problems and their solutions (Definition~\ref{defn.thetaC}).
We also prove that this translation is `best possible' and `complete' in senses which we make formal (Theorems~\ref{thrm.at.least} and Theorems~\ref{thrm.sound} and~\ref{thrm.strong}).

\end{squishitem}

\subsection{Map of the paper}

The paper is organised as follows:

We introduce permissive nominal terms in Section~\ref{sect.pernt}.
Notable results are the `always fresh' and `always rename' properties for terms (Corollaries~\ref{corr.always.fresh} and~\ref{corr.always.rename} respectively).

In Section~\ref{sec.substitutions} we introduce technical definitions and results which will prove useful in the development of the unification algorithm, including permissive nominal terms substitution (Definition~\ref{defn.subst}), unification problems and their solutions (Definition~\ref{defn.problems}).

We clarify the relationship between nominal and permissive nominal terms in Section~\ref{sec.relation.to.nominal.terms}, by injecting nominal terms into permissive nominal terms (Definition~\ref{defn.interpretation.nominal} onwards).
We also elucidate the relationship between solutions of nominal unification problems, based on subsititution+freshness, and solutions of permissive nominal unification problems, based on substitution alone (Definition~\ref{defn.interpretation.solutions}).

In Sections~\ref{sect.Inc} and~\ref{sect.pernu} we present an algorithm to compute most general solutions to permissive nominal unification problems (Definition~\ref{defn.unification.algorithm}).
We prove it correct in Theorem~\ref{thrm.algorithm.correctness}.

We define $\lambda$-terms syntax (Definition~\ref{defn.terms.g}) in Section~\ref{sect.lambda.calculus}, and higher-order patterns (Definition~\ref{defn.phi.patterns}). 
Section~\ref{sect.nominal.terms.to.lambda} translates from permissive nominal term syntax to $\lambda$-term syntax (Definition~\ref{defn.nominal.trans.g} onwards). 
We prove that this translation is `best possible', in a suitable sense which we make formal (Theorem~\ref{thrm.at.least}).

Section~\ref{sect.translating.solutions} relates the solutions of higher-order pattern unification with solutions of permissive nominal unification (Definition~\ref{defn.thetaC}).
We show that the instantiation ordering, hence the property of solutions being `more general', is preserved by the translation (Corollary~\ref{corr.instant}).
We prove a soundness and weak completeness result (Theorem~\ref{thrm.sound}). 
Finally, we refine this to a more complex but more powerful completeness result (Theorem~\ref{thrm.strong}). 

In Section~\ref{sect.conclusion}, we conclude and suggest ideas for future work.
 
\section{Permissive nominal terms}
\label{sect.pernt}

We set up the syntax of permissive nominal terms (Definition~\ref{defn.terms}) and their notion of $\alpha$-equivalence $\aeq$, which does not require a freshness context (Definition~\ref{defn.aeq}).
We prove that permissive nominal terms have the `always fresh' and `always rename' properties (Corollaries~\ref{corr.always.fresh} and~\ref{corr.always.rename}).
Finally, we verify that $\aeq$ is an equivalence relation; this mirrors the result for nominal terms \cite[Theorem 2.11]{gabbay:nomu-jv}, though the proof-method is based on \cite[Subsection 3.2]{gabbay:nomr-jv}.  Where we omit proofs they are routine (or see a technical report \cite{gabbay:perntu-tr}, or \cite{gabbay:nomr-jv}). 

Note the \emph{two} notions of `free variables of'; $\f{fa}(r)$ is the free atoms in $r$, and $\f{fV}(r)$ is the free unknowns in $r$ (Definitions~\ref{defn.fa} and~\ref{defn.fV}).  This reflects the two-level structure of nominal terms familiar from previous work \cite{gabbay:nomu-jv}. 

\begin{defn}
\label{defn.atoms}
Fix two disjoint countably infinite sets $\mathbb A^<$ and $\mathbb A^>$ of \deffont{atoms} and write 
$$
\mathbb A=\mathbb A^<\uplus \mathbb A^> .
$$
(Here $\uplus$ denotes disjoint set union)
$a,b,c,\ldots$ will range over \emph{distinct} elements of $\mathbb A$ (we call this the \deffont{permutative convention}). 
\end{defn}

\begin{defn}
\label{defn.the.comb}
Define $\mathcal P$ by 
$$
\mathcal P = \{(\mathbb A^< \setminus A)\cup B\mid A\subseteq\mathbb A^<,\ B\subseteq\mathbb A^>,\ A,B\text{ finite}\} .
$$
Call elements of $\mathcal P$ \deffont{permission sets}.
$S,S',T$ will range over permission sets.
\end{defn}

Call $S\subseteq\mathbb A$ \deffont{co-infinite} when $\mathbb A\setminus S$ is infinite.
$\mathcal P$ is a set of infinite, co-infinite sets of atoms.\footnote{Other versions of $\mathcal P$ are possible.

We believe that a sufficient property for $\mathcal P$ is that it be: \emph{co-infinitely down-closed} (so if $S\in\mathcal P$ and $S'\subseteq S$ is co-infinite, then $S'\in\mathcal P$); and such that for any finite $\{S_1,\ldots,S_n\}\subseteq\mathcal P$, the union $S_1\cup\ldots\cup S_n$ is co-infinite (it does not have to be in $\mathcal P$).

The first property is sufficient to build the permission set in \eqref{eq.S'}.  The second property is sufficient to guarantee Corollary~\ref{corr.always.fresh}.
This is similar in spirit to the notion of \emph{support ideal} from \cite[Definition~4.1]{cheney:comhtn}.

}

\begin{defn}
\label{defn.unknowns}
For each permission set $S$ fix a disjoint countably infinite set of \deffont{unknowns} of sort $S$.
\ $X^S$, $Y^S$, $Z^S$, will range over distinct unknowns of sort $S$.
\ If $S\neq S'$ then there is no particular connection between $X^S$ and $X^{S'}$.
\ $\mathcal V$ will range over finite sets of unknowns (we use this from Section~\ref{sect.Inc} onwards).
\end{defn}

\begin{defn}
\label{defn.dom.bijection}
Suppose $f$ is a function from atoms to atoms.
Define $\f{nontriv}(f)$ by:
\begin{displaymath}
\f{nontriv}(f) = \{ a \mid f(a) \neq a\}
\end{displaymath}
\end{defn}
This has also been called the \emph{support} of $\pi$ \cite{gabbay:newaas-jv}.

\begin{defn}
\label{def.permutation}
A (finite) \deffont{permutation} is a bijection on atoms such that $\f{nontriv}(\pi)$ is finite.
$\pi$ and $\pi'$ will range over finite permutations.

Write $\pi \fcomp \pi'$ for the \deffont{composition} of $\pi$ and $\pi'$ (so $(\pi\circ\pi')(a)=\pi(\pi'(a))$).
Write $\Id$ for the \deffont{identity} permutation (so $\Id(a)=a$ always).
Write $(a\ b)$ for the \deffont{swapping} permutation that swaps $a$ and $b$. 
\end{defn}

\begin{defn} 
\label{defn.terms}
Fix a set of \deffont{term-formers}.
$\tf f,\tf g,\tf h$ will range over distinct term-formers.

Define \deffont{(permissive nominal) terms} by:
\begin{frameqn}
r, s, t, \ldots\ ::=\ a \mid \tf f(r,\ldots,r) \mid [a]r \mid \pi \act X^S 
\end{frameqn}

We write $\equiv$ for syntactic identity; $r\equiv s$ when $r$ and $s$ denote identical terms.
Note that $X^S$ (the unknown) is not a term, however $\pi\act X^S$ is a term and in particular $\Id\act X^S$ is a term, which we may write as $X^S$.
\end{defn}

\begin{rmrk}
\label{rmrk.harmless}
Permissive-nominal terms are finite. 
They are finitely branching finitely deep trees.

Equality of permissive-nominal terms may be calculated in finite time.
Permission sets trivially admit a finite representation as the pair of finite sets $A$ and $B$ in Definition~\ref{defn.the.comb}.

In an implementation of the algorithms in this paper by Mulligan \cite{mulligan:imppnt}, atoms are implemented concretely as numbers. 
$\mathbb A^< $ is identified with the even numbers.
Permission sets are represented finitely as their finite deviation from $\mathbb A^< $.
 
Another, quite elegant, presentation is possible using \emph{exclusive or}; see
Remark~\ref{rmrk.co-infinite}. 
\end{rmrk}

\begin{rmrk}
See Section~\ref{sec.relation.to.nominal.terms} for a comparison between permissive nominal terms of Definition~\ref{defn.terms} and `ordinary' nominal terms \cite{gabbay:nomu-jv}. 
Atoms represent \emph{variable symbols}; term-formers \emph{functions}; unknowns \emph{meta-variables}; abstraction $[a]r$ \emph{binding}; and $\pi\act X^S$ a meta-variable with a suspended substitution, like `$t[y/x]$'.
For example, suppose term-formers $\tf{app}$ and $\lambda$:
\begin{squishlist}
\item
$\tf{app}(a,b)$ can represent `$xy$' ($x$ applied to $y$). 
\item
$\tf{app}(\tf{lam}([a]a),b)$ can represent `$(\lambda x.x)y$' (identity applied to $y$).
\item
$\lambda([a]X^S)$ can represent `$\lambda x.t$' if $a\in S$, and `$\lambda x.t$, where $x\not\in\f{fv}(t)$' if $a\not\in S$.
\end{squishlist}
\end{rmrk}

\begin{defn}
\label{defn.perm}
Define a \deffont{permutation action} by: 
\begin{frameqn}
\begin{array}{r@{\ }l@{\qquad}r@{\ }l}
\pi \act a\equiv& \pi(a) 
&
\pi \act (\tf f(r_1,\ldots,r_n)) \equiv& \tf f(\pi \act r_1,\ldots,\pi\act r_n)
\\
\pi \act [a]r \equiv& [\pi(a)](\pi \act r)
&
\pi \act (\pi' \act X^S) \equiv& (\pi {\fcomp} \pi') \act X^S
\end{array}
\end{frameqn}
\end{defn}

\begin{defn}
\label{defn.pointwise.action}
If $S \subseteq \mathbb{A}$, define the \deffont{pointwise} action by:
\begin{frameqn}
\pi \act S = \{ \pi(a) \mid a \in S \}
\end{frameqn}
\end{defn}

\begin{defn}
\label{defn.fa}
Define \deffont{free atoms} $\f{fa}(r)$ by: 
\begin{frameqn}
\begin{array}{r@{\ }l@{\qquad}r@{\ }l}
\f{fa}(a) =& \{a\}
&
\f{fa}(\tf f(r_1,\ldots,r_n)) =& \bigcup_{1\leq i\leq n} \f{fa}(r_i) 
\\
\f{fa}([a]r) =& \f{fa}(r){\setminus}\{a\} 
&
\f{fa}(\pi {\act} X^S) =& \pi{\act} S
\end{array}
\end{frameqn}
\end{defn}
Note that $\f{fa}(\pi\act X^S)=\pi\act S$.
  Thus, an intuition for $\f{fa}(r)$ is `the free atoms we can have after instantiation'.

\begin{defn}
\label{defn.fV}
Define \deffont{free unknowns} $\f{fV}(r)$ by:
\begin{frameqn}
\begin{array}{r@{\ }l@{\qquad}r@{\ }l}
\f{fV}(a)  =& \varnothing
&
\f{fV}(\tf f(r_1,\ldots,r_n)) =& \f{fV}(r_1)\cup\cdots\cup\f{fV}(r_n) 
\\
\f{fV}([a]r) =&\f{fV}(r)
&
\f{fV}(\pi{\act} X^S) =& \{X^S\}
\end{array}
\end{frameqn}
\end{defn}

\begin{defn}
\label{defn.aeq}
If $A\subseteq\mathbb A$, define $\pi|_A$,\ \  \deffont{$\pi$ restricted to $A$}, to be:
\begin{frameqn}
\begin{array}{r@{\ }l@{\qquad}l}
\pi|_A(a)=&\pi(a) & \text{when}\ \ a\in A
\\
\pi|_A(a)\mathrel{\phantom{=}} &\text{undefined}&\text{when}\ \ a\in \mathbb A\setminus A
\end{array}
\end{frameqn}
Define \deffont{$\alpha$-equivalence} $\aeq$ inductively by the rules in Figure~\ref{fig.permissive.aeq}. 
\begin{figure*}
\begin{gather*}
\begin{prooftree}
\justifies
a\aeq a
\using\rulefont{{\aeq}aa}
\end{prooftree}
\qquad
\begin{prooftree}
  r_1\aeq s_1\quad\cdots\quad r_n\aeq s_n
  \justifies
  \tf f(r_1,\ldots,r_n) \aeq \tf f(s_1,\ldots,s_n)
  \using \rulefont{{\aeq}\tf f}
\end{prooftree}
\qquad
\begin{prooftree}
  r \aeq s
  \justifies
  [a]r \aeq [a]s
  \using \rulefont{{\aeq}[a]}
\end{prooftree}
\\[2ex]
\begin{prooftree}
  (b\ a)\act r \aeq s\qquad (b\not\in\f{fa}(r))
  \justifies
  [a]r \aeq [b]s
  \using \rulefont{{\aeq}[b]}
\end{prooftree}
\qquad
\begin{prooftree}
(\pi|_S  = \pi'|_S)
\justifies
\pi\act X^S \aeq \pi'\act X^S
\using\rulefont{{\aeq}X}
\end{prooftree}
\end{gather*}
\caption{Derivable $\alpha$-equivalence on permissive nominal terms.}
\label{fig.permissive.aeq}
\end{figure*}
\end{defn}

Corollaries~\ref{corr.always.fresh} and~\ref{corr.always.rename} are properties that `ordinary syntax' has, that nominal terms do not have, and that permissive nominal terms recover;
we can always choose a fresh variable,
and we can always $\alpha$-rename with it.

\begin{corr}
\label{corr.always.fresh}
For any $r_1$, \ldots , $r_n$  there exist infinitely many $b$ such that $b \not\in \bigcup\{\f{fa}(r_i)\mid 1\leq i\leq n\}$.   
\end{corr}
\begin{proof}
Define $\f{atoms}(r)$ inductively by:
$$
\begin{array}{r@{\ }l@{\qquad}r@{\ }l}
\f{atoms}(a)=&\{a\}
&
\f{atoms}(\tf f(r_1,\ldots,r_n))=&\f{atoms}(r_1)\cup\ldots\cup\f{atoms}(r_n)
\\
\f{atoms}([a]r)=&\f{atoms}(r)\cup\{a\}
&
\f{atoms}(\pi\act X^S)=&\f{nontriv}(\pi)
\end{array}
$$
It is not hard to prove 
by induction on term syntax that 
$\f{fa}(r_i)\subseteq\f{atoms}(r_i)\cup\bigcup\{S\mid X^S\in\f{fV}(r_i)\}$ for $1\leq i\leq n$.
The syntax of $r_i$ is finite so $\f{atoms}(r_i)$ is finite, and also $\f{fV}(r_i)$ is finite.
It follows 
that $\bigcup\{S\mid X^S\in\f{fV}(r_i)\text{ for some }i\}$ is co-infinite. The result follows.
\end{proof}

Later on, we will often need to say `choose an atom fresh' (see for example Definition~\ref{defn.Pr.translation}).
When we do this, we are using Corollary~\ref{corr.always.fresh}.

Corollary~\ref{corr.always.rename} expresses that we can always $\alpha$-rename:
\begin{corr}
\label{corr.always.rename}
For any $r$ and $a$ there exists infinitely many fresh $b$ (so $b\not\in\f{fa}(r)$) such that for some $s$,\ $[a]r\aeq [b]s$.
\end{corr} 
\begin{proof}
Immediate, by Corollary~\ref{corr.always.fresh} and \rulefont{{\aeq}[b]}.
\end{proof}

Our changes do not affect basic results about nominal terms \cite{gabbay:nomu-jv}; the proofs of the following lemmas are by routine inductions (see \cite{gabbay:perntu-tr} for details):
\begin{lemm}
\label{lemm.permutation.comp}
\begin{enumerate}
\item
$\id\act r\equiv r$
\item
$\pi' \act (\pi \act r) \equiv (\pi' \fcomp \pi) \act r$
\end{enumerate}
\end{lemm}

\begin{lemm}
\label{lemm.pi.ftma}
$\pi\act \f{fa}(r)=\f{fa}(\pi\act r)$.
\end{lemm}

\begin{lemm}
\label{lemm.equality.permutation}
If $r \aeq s$ then $\pi \act r \aeq \pi \act s$.
\end{lemm}

\begin{lemm}
\label{lemm.aeq.ftma.pres}
If $r \aeq s$ then $\f{fa}(r) = \f{fa}(s)$.
\end{lemm}

\begin{lemm}
\label{lemm.ds.pi}
If $\pi|_{\f{fa}(r)} = \pi'|_{\f{fa}(r)}$ then $\pi \act r \aeq \pi' \act r$.
\end{lemm}

\begin{prop}
\label{prop.aeq.transitive}
$\aeq$ is transitive, reflexive, and symmetric.
\end{prop}
\begin{proof}
See Appendix~\ref{sect.omitted.proofs}.
We use Lemmas~\ref{lemm.permutation.comp}, \ref{lemm.pi.ftma}, \ref{lemm.equality.permutation}, \ref{lemm.aeq.ftma.pres} and~\ref{lemm.ds.pi}.
\end{proof}

\subsection{Foundations of permission sets}

\begin{rmrk}
\label{rmrk.fm.not.enough}
The Fraenkel-Mostowski sets model used in the work which introduced nominal techniques \cite{gabbay:newaas-jv} famously does not admit sets like $S$ and $T$, because they do not have finite support.
It is not an issue in this paper because we are not concerned with representing permissive nominal syntax-up-to-binding.
Permissive nominal syntax-up-to-binding can be constructed though; see \cite{gabbay:semnt-ea,gabbay:pernas}. 
See also generalisations of nominal sets by the second author \cite{gabbay:fmhotn,gabbay:genmn} or by Cheney (\cite[Section 3]{cheney:thesis}, or \cite{cheney:comhtn}).
\end{rmrk} 

\begin{rmrk}
\label{rmrk.co-infinite}
Define $S\shs T$ (the \deffont{exclusive or}) by
$$
S\shs T = (S\setminus T) \cup (T\setminus S) .
$$
It is a fact that for every $S\in\mathcal P$,\ the set $S\shs\mathbb A^< $ is a finite set of atoms, and it is the unique finite set such that $\mathbb A^< \shs(S\shs\mathbb A^< )=S$.
This gives a nice finite representation of permission sets, alternative to the two mentioned in Remark~\ref{rmrk.harmless}. 
\end{rmrk}

\section{Substitutions, problems, and solutions}
\label{sec.substitutions}

\subsection{Substitutions}

The purpose of an unknown $X^S$ is to represent an `unknown term/unknown element'.
We therefore define a substitution action for unknowns.
Consistent with nominal terms, substitution for unknowns is capturing for abstraction by atoms.

\begin{defn}
\label{defn.subst}
A \deffont{substitution} $\theta$ is a function from unknowns to terms such that 
$\f{fa}(\theta(X^S))\subseteq S$ always (so $S$ in $X^S$ describes the `permission' we have to instantiate $X^S$, namely to terms with free atoms in $S$).
$\theta$, $\theta'$, $\theta_1$, $\theta_2$, will range over substitutions. 

Write $\id$ for the \deffont{identity} substitution mapping $X^S$ to $\id\act X^S$ always. 
It will always be clear whether $\Id$ means the identity substitution or permutation.
Suppose $\f{fa}(t)\subseteq S$.
Write $[X^S\sm t]$ for the substitution such that
\begin{frameqn}
\begin{aligned}
& [X^S \sm t](X^S) \equiv\ t
& \text{and} &&
[X^S\sm t](Y^T)\equiv\ & \id\act Y^T 
&\ &\text{for all other}\  Y^T.
\end{aligned}
\end{frameqn}
\end{defn}
`$\f{fa}(\theta(X^S))\subseteq S$' looks absent in nominal terms theory (\cite[Definition 2.13]{gabbay:nomu-jv}, \cite[Definition 4]{gabbay:nomr-jv}), yet it is there: see the conditions `$\nabla'\cent \theta(\nabla)$' in Lemma 2.14, and `$\nabla\cent a\#\theta(t)$' in Definition 3.1 of \cite{gabbay:nomu-jv}.
More on this in Section~\ref{sec.relation.to.nominal.terms}.

\begin{defn}
\label{defn.subst.action}
Define a \deffont{substitution action} on terms by:
\begin{frameqn}
\begin{array}{r@{\ }l@{\qquad}r@{\ }l}
a\theta \equiv& a
&
\tf f(r_1,\ldots,r_n)\theta \equiv& \tf f (r_1\theta,\ldots,r_n\theta)
\\
([a]r)\theta \equiv& [a](r\theta)
&
(\pi{\act} X^S)\theta \equiv& \pi{\act} \theta(X^S)
\end{array}
\end{frameqn}
\end{defn}
Note that $X^S\theta$ means `$\theta$ acting on $\id\act X^S$'; $\theta(X^S)$ means `the value of function $\theta$ at $X^S$'.

\begin{lemm}
\label{lemm.substitution}
$\f{fa}(r\theta)\subseteq\f{fa}(r)$. 
\end{lemm}
\begin{proof}
See Appendix~\ref{sect.omitted.proofs}.
We use Lemma~\ref{lemm.pi.ftma}.
\end{proof}

\begin{lemm}
\label{lemm.sub.perm}
$\pi\act (r\theta)\equiv (\pi\act r)\theta$.
\end{lemm}
\begin{proof}
By induction on $r$.
\end{proof}

\begin{lemm}
\label{lemm.eq.X.eq.always}
If $\theta_1(X^S)\aeq \theta_2(X^S)$ for all $X^S\in \f{fV}(r)$, then $r\theta_1\aeq r\theta_2$. 
\end{lemm}
\begin{proof}
By induction on $r$.
\begin{squishitem}
\item
The cases $a$ and $\tf{f}(r_1, \ldots, r_n)$ are straightforward.
\item
The case $[a]r$.\quad
Suppose $\theta_1(X^S) \aeq \theta_2(X^S)$ for every $X^S \in \f{fV}([a]r)$.
$\f{fV}([a]r)=\f{fV}(r)$ so by inductive hypothesis $r\theta_1 \aeq r\theta_2$. 
By \rulefont{{\aeq}[a]} also $[a](r\theta_1) \aeq [a](r\theta_2)$. 
The result follows by Definition~\ref{defn.subst.action}.
\item
The case $\pi \act X^S$.\quad
By assumption, $\theta_1(X^S) \aeq \theta_2(X^S)$.
Using Lemma~\ref{lemm.equality.permutation}, $\pi \act (\theta_1(X^S)) \aeq \pi \act (\theta_2(X^S))$.
By Definition~\ref{defn.subst.action} $(\pi \act X^S)\theta_1 \aeq (\pi \act X^S)\theta_2$, as required.
\end{squishitem}
\end{proof}

\begin{lemm}
\label{lemm.aeq.subst}
If $r \aeq s$ then $r\theta \aeq s\theta$.
\end{lemm}
\begin{proof}
By induction on the derivation of $r \aeq s$.
We consider one case:
\begin{squishitem}
\item
The case \rulefont{{\aeq}[b]}.\quad
Suppose $(b\ a) \act r \aeq s$ and $b \not\in \f{fa}(r)$.
Then $((b\ a) \act r)\theta \aeq s\theta$ by assumption.
By Lemma~\ref{lemm.sub.perm}, $(b\ a) \act (r\theta) \aeq s\theta$.
By Lemma~\ref{lemm.substitution}, $b \not\in \f{fa}(r\theta)$, therefore $[a](r\theta) \aeq [b](s\theta)$ by \rulefont{{\aeq}[b]}.
By Definition~\ref{defn.subst}, $[a](r\theta) \equiv ([a]r)\theta$, and the result follows.
\end{squishitem}
\end{proof}

\begin{defn}
\label{defn.sigma.com}
Define \deffont{composition} $\theta_1 \fcomp \theta_2$ by $(\theta_1{\fcomp}\theta_2)(X^S) \equiv (\theta_1(X^S))\theta_2$. 
\end{defn} 

\begin{lemm}
\label{lemm.sigma.sigma'.circ}
$(r\theta)\theta' \equiv r(\theta \fcomp \theta')$.
\end{lemm}
\begin{proof}
By induction on $r$.
We consider one case:
\begin{squishitem}
\item
The case $\pi \act X^S$
\begin{calcenv}
(\pi\act X^S)(\theta \fcomp \theta') & \equiv & \pi \act (\theta \fcomp \theta')(X^S) & \text{Definition~\ref{defn.subst.action}} \\
                                     & \equiv & \pi \act (\theta(X^S)\theta') & \text{Definition~\ref{defn.sigma.com}} \\
                                     & \equiv & (\pi \act \theta(X^S))\theta' & \text{Lemma~\ref{lemm.sub.perm}} \\
                                     & \equiv & ((\pi \act X^S)\theta)\theta' & \text{Lemma~\ref{lemm.sub.perm}}
\end{calcenv}
\qedhere
\end{squishitem}
\end{proof}

\subsection{Unification problems, and solutions}
\label{subsect.unif}

This is a brief subsection, but it is useful:
A solution to $Pr$ `makes the equalities valid', as for first- and higher-order unification.
This simplifies the nominal unification notion of solution (Definition~\ref{defn.substitution.interp} or \cite[Definition~3.1]{gabbay:nomu-jv}) based on `a substitution + a freshness context'.
We prove results about these definitions in Sections~\ref{sec.relation.to.nominal.terms} (connection with nominal unification) and~\ref{sect.pernu} (unification algorithm).

\begin{defn}
\label{defn.problems}
\label{defn.unif.sol}
An \deffont{equality} is a pair $r\ueq s$.
A \deffont{problem} $Pr$ is a finite multiset of equalities. 
Define $Pr\theta$ by:
\begin{displaymath}
Pr\theta = \{r\theta\ueq s\theta \mid r\ueq s\in Pr\}
\end{displaymath}
Say that $\theta$ \deffont{solves} $Pr$ when 
$r\ueq s\in Pr$ implies $r\theta\aeq s\theta$.
Write $\f{Sol}(Pr)$ for the set of solutions to $Pr$.
Call $Pr$ \deffont{solvable} when $\f{Sol}(Pr)$ is non-empty.
\end{defn}

\section{Relation to nominal terms}
\label{sec.relation.to.nominal.terms}

In Subsection~\ref{subsect.unif} we stated what a permissive nominal unification problem is, and what a solution to it is.
We now make precise a mathematical sense in which nominal terms and their unification can be considered a subsystem of permissive nominal terms and their unification. 

We recall the notions of nominal term and nominal unification problem (Definitions~\ref{defn.nominal.terms} and~\ref{defn.substitution.interp}).
We translate from the `nominal world' to the `permissive nominal world' (Definition~\ref{defn.interpretation.nominal}).
Theorem~\ref{thrm.interpretation.injective} expresses how this translation is sound and complete for respective notions of $\alpha$-equivalence.
Theorem~\ref{thrm.no.missing} then shows that furthermore, solutions to unification problems are preserved 1-1 across the translation.

Recall from the Introduction that in nominal terms we often to enrich the freshness context.
An interesting feature of Definition~\ref{defn.interpretation.nominal} is how it maps nominal terms to permissive nominal terms with free atoms in $\mathbb A^< $, which is an infinite, co-infinite set of atoms.
One way to view the interpretation of Definition~\ref{defn.interpretation.nominal} is therefore this: $\mathbb A^< $ is `the atoms we had available so far' (any other permission set would do as well) and $\mathbb A^> $ is `the atoms with which we will extend the freshness context, in the future'.
Both these sets are infinite, and syntax is finite, so it is not absolutely necessary to explicitly separate them: permissive nominal terms do this, for each fixed permission set $S$; nominal terms do not.

\subsection{Alpha-equivalence between nominal and permissive nominal terms}

\begin{defn}
\label{defn.translation}
Fix a countably infinite set of \deffont{nominal atoms}, $\dot{\mathbb A}$.
$\dot a,\dot b,\dot c,\ldots$ will range over distinct nominal atoms.

Fix a bijection $\iota$ between $\dot{\mathbb A}$ and any permission set.
For concreteness we will suppose it is $\mathbb A^< $ from Definition~\ref{defn.the.comb} but any permission set will do as well.

Fix a countably infinite set of \deffont{nominal unknowns}.
$\dot X, \dot Y, \dot Z, \ldots$ will range over distinct nominal unknowns.
A \deffont{nominal permutation} is a bijection $\dot\pi$ on $\dot{\mathbb A}$ such that $\f{nontriv}(\dot\pi)$ is finite.
$\dot\pi, \dot\pi', \dot\pi'', \ldots$ will range over permutations.

Write $\dot\pi^\mone$ for the inverse of $\dot\pi$, $\dot\id$ for the identity permutation, and $\dot\pi\fcomp\dot\pi'$ for function composition, as is standard.
For example, $(\dot\pi \fcomp \dot\pi')(\dot a) = \dot\pi(\dot\pi'(\dot a))$
\end{defn}

\begin{defn}
\label{defn.nominal.terms}
Define \deffont{nominal terms} by: 
\begin{frameqn}
\dot r,\dot s,\dot t ::= \dot a \mid \dot\pi \act \dot X \mid [\dot a]\dot r \mid \tf f(\dot r, \ldots, \dot r)
\end{frameqn}
\end{defn}

\begin{defn}
\label{defn.nominal.permutation}
Define a \deffont{permutation action} on nominal terms by: 
\begin{frameqn}
\begin{array}{r@{\ }l@{\qquad}r@{\ }l}
\dot\pi \act \dot a \equiv& \dot\pi(\dot a)
&
\dot \pi \act \tf f(\dot r_1, \ldots, r_n) \equiv& \tf f(\dot \pi \act \dot r_1, \ldots, \dot \pi \act r_n)
\\
\dot\pi \act {[\dot a] \dot r} \equiv& [\dot\pi(\dot a)](\dot\pi \act \dot r)
&
\dot\pi \act (\dot\pi' \act \dot X) \equiv& (\dot\pi \fcomp \dot\pi') \act \dot X
\end{array}
\end{frameqn}
\end{defn}
Write $\equiv$ for syntactic identity. 
$\tf f$ ranges over term-formers (Definition~\ref{defn.atoms}).

\begin{figure}
\begin{gather*}
\begin{prooftree}
  \justifies
  \Delta \cent \dot{a} \# \dot{b}
  \using\rulefont{\#\dot b}
\end{prooftree}
\qquad
\begin{prooftree}
  \Delta \cent \dot{a} \# \dot{r}_i \quad (1 \leq i \leq n)
  \justifies
  \Delta \cent \dot{a} \# \tf f(\dot r_1,\ldots,\dot r_n)
  \using\rulefont{\#\tf f}
\end{prooftree}
\qquad
\begin{prooftree}
  \justifies
  \Delta \cent \dot{a} \# [\dot{a}]\dot{r}
  \using\rulefont{\#[\dot a]}
\end{prooftree}
\\[1.25ex]
\begin{prooftree}
  \Delta \cent \dot{a} \# \dot{r}
  \justifies
  \Delta \cent \dot{a} \# [\dot{b}]\dot{r}
  \using\rulefont{\#[\dot b]}
\end{prooftree}
\qquad
\begin{prooftree}
  ({\dot{\pi}^\mone(\dot a)} \# \dot{X} \in \Delta)
  \justifies
  \Delta \cent \dot{a} \# \dot{\pi}\act\dot{X}
  \using\rulefont{\#\dot X}
\end{prooftree}
\end{gather*}
\caption{Derivable freshness on nominal terms}
\label{fig.nominal.freshness}
\end{figure}

\begin{defn}
\label{defn.nominal.freshness}
A \deffont{freshness} is a pair $\dot a\#\dot r$.
A \deffont{freshness context} is a finite set of freshnesses of the form $\dot a\#\dot X$.
Define \deffont{derivable freshness} on nominal terms by the rules in Figure~\ref{fig.nominal.freshness}.
\end{defn}

\begin{figure}
\begin{gather*}
\begin{prooftree}
  \justifies
  \Delta \cent \dot{a} = \dot{a}
  \using\rulefont{{=}\dot a}
\end{prooftree}
\qquad
\begin{prooftree}
  \Delta \cent \dot{r}_i = \dot{s}_i \quad (1 \leq i \leq n)
  \justifies
  \Delta \cent \tf f(\dot r_1,\ldots,\dot r_n) = \tf f(\dot s_1,\ldots,\dot s_n)
  \using\rulefont{{=}\tf f}
\end{prooftree}
\qquad
\begin{prooftree}
  \Delta \cent \dot{r} = \dot{s}
  \justifies
  \Delta \cent [\dot{a}]\dot{r} = [\dot{a}]\dot{s}
  \using\rulefont{{=}[\dot a]}
\end{prooftree}
\\[1.25ex]
\begin{prooftree}
  \Delta \cent (\dot{b}\ \dot{a})\act\dot{r} = \dot{s} \quad \Delta \cent \dot{b} \# \dot{r}
  \justifies
  \Delta \cent [\dot{a}]\dot{r} = [\dot{b}]\dot{s}
  \using\rulefont{{=}[\dot b]}
\end{prooftree}
\qquad
\begin{prooftree}
  (\dot{a} \# \dot{X}\in\Delta \text{ for every } \dot\pi(\dot{a}) \neq \dot\pi'(\dot{a}))
  \justifies
  \Delta \cent \dot{\pi}\act\dot{X} = \dot{\pi}'\act\dot{X}
  \using\rulefont{{=}\dot X}
\end{prooftree}
\end{gather*}
\caption{Derivable equality on nominal terms}
\label{fig.nominal.equality}
\end{figure}

Definition~\ref{defn.nominal.equality} repeats \cite[Figure 2]{gabbay:nomu-jv}, up to differences in presentation:
\begin{defn}
\label{defn.nominal.equality}
An \deffont{equality} is a pair $\dot r=\dot s$.
Define \deffont{derivable equality} on nominal terms by the rules in Figure~\ref{fig.nominal.equality}.
\end{defn}

\begin{defn}
\label{defn.interpretation.nominal}
Define a mapping $\inter{\dot \pi}$ from nominal permutations to permissive nominal permutations by: 
\begin{frameqn}
\begin{array}{r@{\ }l@{\quad}l}
\inter{\dot \pi}(\iota(\dot a))=&\iota(\dot \pi(\dot a))&
\\
\inter{\dot\pi}(c)=&c&\text{all }c\in \mathbb A^> 
\end{array}
\end{frameqn}
Define an \deffont{interpretation} $\interdelta{\dot r}$ by:
\begin{frameqn}
\begin{array}{r@{\ }l@{}l}
\interdelta{\dot{a}} \equiv& \iota(\dot{a}) 
\\
\interdelta{\tf f(\dot r_1, \ldots, \dot r_n)} \equiv& \tf{f}(\interdelta{\dot r_1}, \ldots, \interdelta{\dot r_n})
\\
\interdelta{[\dot{a}]\dot{r}} \equiv& [\iota(\dot a)]\interdelta{\dot r}
\\
\interdelta{\dot\pi \act \dot X} \equiv& \inter{\dot\pi} \act X^S & \text{where }S = \mathbb A^<  \setminus \{ \iota(\dot{a}) \mid \dot{a} \# \dot{X} \in \Delta \}
\end{array}
\end{frameqn}
Here, we make a fixed but arbitrary choice of $X^S$ for each $\dot X$, injectively so that $\interdelta{\dot X}$ and $\interdelta{\dot Y}$ are always distinct.
\end{defn}

$\interdelta{\dot r}$ commutes with permutation and it preserves and reflects freshness:

\begin{lemm}
\label{lem.nominal.pi.interpretation}
$\inter{\dot{\pi}} \act \interdelta{\dot{r}} \equiv \interdelta{\dot{\pi}\act\dot{r}}$
\end{lemm}
\begin{proof}
By induction on $\dot{r}$.
\end{proof}

\begin{lemm}
\label{lem.interpretation.fresh.preserve.iff}
$\iota(\dot{a}) \not\in \f{fa}(\interdelta{\dot{r}})$ if and only if $\Delta \cent \dot{a} \# \dot{r}$.
\end{lemm}
\begin{proof}
See Appendix~\ref{sect.omitted.proofs}.
\end{proof}

Theorem~\ref{thrm.interpretation.injective} states that $\alpha$-equivalent nominal-terms-in-context map precisely to $\alpha$-equivalent permissive nominal terms: 

\begin{thrm}
\label{thrm.interpretation.injective}
\label{thm.interpretation.equality.preserve}
$\interdelta{\dot{r}} \aeq \interdelta{\dot{s}}$ if and only if $\Delta \cent \dot{r} = \dot{s}$.
\end{thrm}
\begin{proof}
See Appendix~\ref{sect.omitted.proofs}.
We use Lemmas~\ref{lem.nominal.pi.interpretation} and~\ref{lem.interpretation.fresh.preserve.iff}.
\end{proof}

\subsection{Substitutions and solutions between nominal and permissive nominal unification problems}

\begin{defn}
\label{defn.nominal.substitution.action}
A \deffont{substitution} $\dot\theta$ is a function from nominal unknowns to nominal terms such that $\{\dot X\mid \dot\theta(\dot X)\not\equiv\dot\id\act\dot X\}$ is finite.
$\dot\theta, \dot\theta', \dot\theta'', \ldots$ will range over nominal substitutions.

Write $\dot\id$ for the \deffont{identity}, mapping $\dot X$ to $\dot\id\act \dot X$. 

Define a \deffont{substitution action} on nominal terms by:
\begin{frameqn}
\dot a\dot\theta \equiv \dot a
\quad
\tf f(\dot r_1, \ldots, \dot r_n)\dot\theta \equiv \tf f(\dot r_1\dot\theta, \ldots, \dot r_n\dot\theta)
\quad
([\dot a]\dot r)\dot\theta \equiv [\dot a](\dot r\dot\theta)
\quad
(\dot\pi {\act} \dot X)\dot \theta \equiv \dot\pi {\act} \dot\theta(\dot X)
\end{frameqn}
\end{defn}

\begin{defn}
\label{defn.substitution.interp}
A \deffont{unification problem} $\dot{Pr}$ is a finite multiset of freshnesses and equalities.
A \deffont{solution} to $\dot{Pr}$ is a pair $(\Delta,\dot{\theta})$ such that $\Delta\cent \dot a\#\dot r\dot\theta$ for every $\dot a\#\dot r\in\dot{Pr}$, and $\Delta\cent \dot r\theta =\dot s\theta$ for every $\dot r=\dot s\in\dot{Pr}$.  
This follows \cite[Definition 3.1]{gabbay:nomu-jv}.
\end{defn}

\begin{defn}
\label{defn.interpretation.solutions}
We extend the interpretation of Definition~\ref{defn.interpretation.nominal} to solutions of nominal unification problems by: 
\begin{frameqn}
\inter{(\Delta,\dot\theta)}(X^S)\equiv\interdelta{\dot\theta(X)}\text{ if }\id{\act} X^S{\equiv}\interdelta{X}
\qquad
\inter{(\Delta,\dot\theta)}(Y^T)\equiv\id{\act} Y^T\text{ otherwise}
\end{frameqn}
\end{defn}

\begin{lemm}
\label{lemm.substitution.commute.interp}
$\interdelta{\dot r}\inter{(\Delta,\dot\theta)} \equiv \interdelta{\dot r\dot\theta}$.
\end{lemm}
\begin{proof}
See Appendix~\ref{sect.omitted.proofs}.
\end{proof}

\begin{defn}
\label{defn.Pr.translation}
Define $\interdelta{\dot{Pr}}$ by mapping $\dot r=\dot s$ to $\interdelta{\dot r}\ueq\interdelta{\dot s}$ and mapping $\dot a\#\dot r$ to $(b\ \iota(\dot a))\act \interdelta{\dot r}\ueq \interdelta{\dot r}$, for some choice of fresh $b$ (so $b\not\in\f{fa}(\interdelta{\dot r})$; in fact, it suffices to choose some $b\not\in\mathbb A^< $).  
\end{defn}

\begin{lemm}
\label{lemm.basic}
Suppose $b\not\in\f{fa}(r)$.
Then $a\not\in\f{fa}(r)$ if and only if $(b\ a)\act r\aeq r$.
\end{lemm}
\begin{proof}
See Appendix~\ref{sect.omitted.proofs}.
We use Lemmas~\ref{lemm.permutation.comp} and~\ref{lemm.pi.ftma}, and Proposition~\ref{prop.aeq.transitive}.
\end{proof}

No solutions go missing, moving from the nominal to the permissive nominal world:
\begin{thrm}
\label{thrm.no.missing}
$(\Delta,\dot\theta)$ solves $\dot{Pr}$ if and only if $\inter{(\Delta,\dot\theta)}$ solves $\interdelta{\dot{Pr}}$.
\end{thrm}
\begin{proof}
\emph{If $(\Delta,\dot\theta)$ solves $\dot{Pr}$ then $\inter{(\Delta,\dot\theta)}$ solves $\interdelta{\dot{Pr}}$.}\quad
Suppose $\Delta\cent \dot r\dot\theta=\dot s\dot\theta$.
Using Lemma~\ref{lemm.substitution.commute.interp} and Theorem~\ref{thm.interpretation.equality.preserve}, $\interdelta{\dot r}\inter{(\Delta,\dot\theta)}\aeq\interdelta{\dot s}\inter{(\Delta,\dot\theta)}$.
 
Suppose $\Delta\cent a\#\dot r\dot\theta$.
Using Lemma~\ref{lem.interpretation.fresh.preserve.iff}, $\iota(\dot a) \not\in \f{fa}(\interdelta{\dot r\dot\theta})$.
By Lemma~\ref{lemm.substitution.commute.interp}, $\iota(\dot a) \not\in \f{fa}(\interdelta{\dot r}\inter{(\Delta, \dot\theta)})$.
By Lemma~\ref{lemm.basic}, $(b\ \iota(\dot a))\act \interdelta{\dot r}\inter{(\Delta,\dot\theta)} \aeq \interdelta{\dot r}\inter{(\Delta,\dot\theta)}$, where $b$ is fresh (see Definition~\ref{defn.Pr.translation}).
By Lemma~\ref{lemm.sub.perm}, $((b\ \iota(\dot a))\act \interdelta{\dot r})\inter{(\Delta,\dot\theta)} \aeq \interdelta{\dot r}\inter{(\Delta,\dot\theta)}$.
The result follows.

\emph{If $\inter{(\Delta,\dot\theta)}$ solves $\interdelta{\dot{Pr}}$ then $(\Delta,\dot\theta)$ solves $\dot{Pr}$.}\quad
Suppose that $\interdelta{\dot r}\inter{(\Delta,\dot\theta)} \aeq \interdelta{\dot s}\inter{(\Delta, \dot\theta)}$.
By Theorem~\ref{thrm.interpretation.injective}, $\Delta \cent r\theta = s\theta$.
 
Suppose $((b\ \iota(\dot a))\act \interdelta{\dot r})\inter{(\Delta,\dot\theta)} \aeq \interdelta{\dot r}\inter{(\Delta,\dot\theta)}$.
By Lemma~\ref{lemm.sub.perm}, $(b\ \iota(\dot a))\act \interdelta{\dot r}\inter{(\Delta,\dot\theta)} \aeq \interdelta{\dot r}\inter{(\Delta,\dot\theta)}$.
By Lemma~\ref{lemm.basic}, $\iota(\dot a) \not\in \f{fa}(\interdelta{\dot r}\inter{(\Delta, \dot\theta)})$.
Using Lemma~\ref{lemm.substitution.commute.interp}, $\iota(\dot a) \not\in \f{fa}(\interdelta{\dot r\dot\theta})$.
By Lemma~\ref{lem.interpretation.fresh.preserve.iff}, $\Delta\cent a\#\dot r\dot\theta$, and the result follows.
\end{proof}

\section{Support inclusion problems}
\label{sect.Inc}

The freshness symbol $a\#r$ used in \cite{gabbay:newaas-jv} and \cite{gabbay:nomu-jv} is ambiguous.
Do we mean 
\begin{squishlist}
\item
`$a$ is not free in the syntax of $r$' or 
\item
`$a$ is not in the support of the denotation of $r$'?
\end{squishlist}
The first option can be called \emph{intensional} or \emph{syntactic} freshness; the second option can be called \emph{extensional} or \emph{semantic} freshness.
 
In nominal terms this question is slightly obscured because a `free atoms of' function on terms is hard to define.
In permissive nominal terms we can easily define a `free atoms of' function; see Definition~\ref{defn.fa}, $\f{fa}(r)$.

We have seen in Section~\ref{sec.relation.to.nominal.terms} how nominal terms' notion of freshness `$a\#r$' translates to a syntactic freshness `$a\not\in \f{fa}(r)$'.
 
Nominal terms unification solves equality and freshness problems within a single rewrite system. 
When we designed the permissive nominal terms unification algorithm, we solve these separately --- one algorithm is described in this section, the other (for equalities) is described in Section~\ref{sect.pernu}.
This is simply a design choice, but it is informed by the translation of $\#$ to a \emph{syntactic} judgement, described above.
For, in future it may be useful to consider unification modulo equational theories (and note that semantic freshness can be easily captured using equations; see \cite[Theorem~5.5]{gabbay:forcie} or \cite[Theorem~4.52]{gabbay:nomuae}).
In the presence of equational theories, because we have factored out fragment of the computation that checks `$a\not\in\f{fa}(r)$', that fragment will remain modular and unaffected by the imposition of equality axioms.

Recall from Definition~\ref{defn.subst} that $\f{fa}(\theta(X^S)) \subseteq \f{fa}(X^S) = S$, and from Lemma~\ref{lemm.substitution} that instantiation must reduce the set of free atoms.
We will exhibit an algorithm which, intuitively, solves the problem ``please make $\f{fa}(r\theta)\subseteq T$ true'' (Definition~\ref{defn.rho} and Lemma~\ref{lemm.consistent.inc.form}).
In fact the algorithm calculates solutions that are most general, in a sense made formal in Theorem~\ref{thrm.sigma.rho.sigma}.  

Next, in Section~\ref{sect.pernu}, we construct an algorithm to solve equality problems.

\subsection{Simplification reduction and normal forms}

\begin{defn}
\label{defn.solves.Inc}
A \deffont{support inclusion} is a pair $r\sqsubseteq T$ of a term and a permissions set. 
A \deffont{support inclusion problem} is a finite multiset of support inclusions;
$\f{Inc}$ will range over support inclusion problems.
Call $\theta$ a \deffont{solution} to $\f{Inc}$ when $\f{fa}(r\theta)\subseteq T$ for every $r\sqsubseteq T\in\f{Inc}$.
Write $\f{Sol}(\f{Inc})$ for the solutions of $\f{Inc}$.
Call $\f{Inc}$ \deffont{solvable} when $\f{Sol}(\f{Inc})\neq\varnothing$, and \deffont{non-trivial} when $\f{nf}(\f{Inc})\neq\varnothing$.
\end{defn}

\begin{defn}
\label{defn.supp.inc}
Define a \deffont{simplification} rewrite relation by the rules in Figure~\ref{fig.simp}.
\begin{figure*}
\begin{displaymath}
\begin{array}{l@{}r@{\ }c@{\ }l@{\quad}l}
\rulefont{{\sqsubseteq}a} & a \sqsubseteq T, \f{Inc} & \simpto{} & \f{Inc} & (a\in T)  \\
\rulefont{{\sqsubseteq}\tf f} & \tf f(r_1, \ldots, r_n) \sqsubseteq T, \f{Inc} & \simpto{} & r_1 \sqsubseteq T, \ldots, r_n \sqsubseteq T, \f{Inc} \\
\rulefont{{\sqsubseteq}[]} & [a]r \sqsubseteq T, \f{Inc} & \simpto{} & r\sqsubseteq T\cup\{a\},\ \f{Inc} \\
\rulefont{{\sqsubseteq}X} & \pi\act X^S \sqsubseteq T, \f{Inc} & \simpto{} & X^S\sqsubseteq \pi^\mone\act T,\ \f{Inc} & (S\not\subseteq\pi^\mone\act T,\ \pi\neq\id) \\
\rulefont{{\sqsubseteq}X'} & \pi\act X^S \sqsubseteq T, \f{Inc} & \simpto{} & \f{Inc} & (S\subseteq\pi^\mone\act T) 
\end{array}
\end{displaymath}
\caption{Simplification of support inclusion problems}
\label{fig.simp}
\end{figure*}
\end{defn}

Definition~\ref{defn.supp.inc} can easily be expressed as a type I conditional term rewrite system, according to the classification scheme of \cite[Definition~7.1.1]{nipkow:terraa}.
This becomes evident if we bear in mind that we can represent atoms as numbers, $\mathbb A^< $ as the even numbers, $S$ as the finite set $S\shs\mathbb A^< $ (Remark~\ref{rmrk.co-infinite}), and $T$ as the finite set $T\shs\mathbb A^< $.

\begin{thrm}
\label{thrm.simpto.sol.inc}
If $\f{Inc}\simpto{}\f{Inc}'$ then $\f{Sol}(\f{Inc})=\f{Sol}(\f{Inc}')$.
\end{thrm}
\begin{proof}
First, we make the following claims:
\begin{squishitem}
\item[\textbf{Claim 1}:]
\label{claim.sol.inc.1}
If $a\in T$ then $\f{fa}(a\theta) \subseteq T$ always.\quad
Since $\f{fa}(a\theta) = \f{fa}(a) = \{ a \}$.
\item[\textbf{Claim 2}:]
\label{claim.sol.inc.2}
$\f{fa}(\tf f(r_1,\ldots,r_n)\theta)\subseteq T$ if and only if $\f{fa}(r_i\theta)\subseteq T$ for $1\leq i\leq n$.\quad 
Since $\f{fa}(\tf{f}(r_1, \ldots, r_n)) = \f{fa}(r_1) \cup \ldots \cup \f{fa}(r_n)$, and $\tf{f}(r_1, \ldots, r_n)\theta \equiv \tf{f}(r_1\theta, \ldots, r_n\theta)$.
\item[\textbf{Claim 3}:]
\label{claim.sol.inc.3}
$\f{fa}(([a]s)\theta)\subseteq T$ if and only if $\f{fa}(s\theta)\subseteq T\cup\{a\}$.\quad
Suppose $\f{fa}(([a]s)\theta) \subseteq T$, therefore $\f{fa}([a]s\theta) \subseteq T$.
Then $\f{fa}(s\theta) \setminus \{ a \} \subseteq T$, therefore $\f{fa}(s\theta) \subseteq T \cup \{ a \}$ and the result follows.
The reverse direction is similar.
\item[\textbf{Claim 4}:]
\label{claim.sol.inc.4}
$\f{fa}((\pi\act X^S)\theta)\subseteq T$ if and only if $\f{fa}(X^S\theta)\subseteq\pi^\mone\act T$.\quad
We consider only one case.
Suppose $\theta = [X^S\sm t]$ and $\f{fa}(t) \subseteq S$, therefore $\f{fa}((\pi\act X^S)[X^S\sm t]) = \f{fa}(\pi \act t)$ hence $\f{fa}(\pi \act t) \subseteq T$ by assumption.
By Lemma~\ref{lemm.pi.ftma}, $\pi \act \f{fa}(t) \subseteq T$, and by Lemma~\ref{lemm.permutation.comp} and Lemma~\ref{lemm.pi.ftma}, $(\pi^\mone \fcomp \pi) \act \f{fa}(t) \subseteq \pi^\mone \act T$.
As $\pi^\mone \fcomp \pi = \id$, we have $\f{fa}(X^S[X^S\sm t]) \subseteq \pi^\mone \act T$, and the result follows.

Alternatively, suppose $\f{fa}(t) \not\subseteq S$.
Then $\f{fa}((\pi\act X^S)[X^S\sm t]) = \f{fa}(\pi \act X^S)$ and $\pi \act \f{fa}(X^S) \subseteq T$ by Lemma~\ref{lemm.pi.ftma}.
By Lemmas~\ref{lemm.permutation.comp} and~\ref{lemm.pi.ftma}, $\f{fa}(X^S[X^S\sm t]) \subseteq \pi^\mone \act T$, and the result follows.
The reverse implication is no harder.
\item[\textbf{Claim 5}:]
\label{claim.sol.inc.5}
By Definition~\ref{defn.subst.action} we have $\f{fa}((\pi \act X^S)\theta) = \f{fa}(\pi \act (X^S\theta))$.
If $S\subseteq\pi^\mone\act T$ then $\f{fa}(\pi\act (X^S\theta))\subseteq T$ always.\quad
Note, $S \subseteq \pi^\mone \act T$ if and only if $\pi \act S \subseteq T$ and $\f{fa}(\pi \act X^S) = \pi \act S$.
Using Definition~\ref{defn.subst.action} and Lemma~\ref{lemm.pi.ftma}, $\f{fa}(\pi \act (X^S\theta)) = \pi \act \f{fa}(\theta(X^S)) \subseteq \pi \act S$.
Therefore, $\f{fa}((\pi \act X^S)\theta) \subseteq T$, and the result follows.
\end{squishitem}
We now proceed by case analysis on $\f{Inc} \simpto{} \f{Inc'}$ (Definition~\ref{defn.supp.inc}):
\begin{squishitem}
\item
The case \rulefont{{\sqsubseteq}a}.\quad
Suppose $a \in T$.
If $\theta \in \f{Sol}(a \sqsubseteq T, \f{Inc'})$ then $\theta \in \f{Sol}(\f{Inc'})$ and the result follows immediately.
Conversely, suppose $\theta \in \f{Sol}(\f{Inc'})$.
Using Claim 1, $\f{fa}(a\theta) \subseteq T$, and the result follows.
\item
The case \rulefont{{\sqsubseteq}\tf f}.\quad
From Claim 2.
\item
The case \rulefont{{\sqsubseteq}[]}.\quad
If $\theta \in \f{Sol}(r \sqsubseteq T \cup \{ a \}, \f{Inc'})$ then $\f{fa}(r\theta) \subseteq T \cup \{ a \}$.
By Claim 3, $\f{fa}([a](r\theta)) \subseteq T$.
As $\f{fa}([a](r\theta)) = \f{fa}(([a]r)\theta)$ and $\theta \in \f{Sol}(Inc')$, the result follows.
The reverse implication is similar.
\item
The case \rulefont{{\sqsubseteq}X}.\quad
Suppose $S \not\subseteq \pi^\mone \act T$, $\pi \not= \id$ and $\theta \in \f{Sol}(\pi \act X^S \sqsubseteq T, \f{Inc})$, so $\f{fa}((\pi \act X^S)\theta) \subseteq T$.
By Claim 4, $\f{fa}(X^S\theta) \subseteq \pi^\mone \act T$, and as $\theta \in \f{Sol}(\f{Inc'})$, the result follows.
The reverse implication is similar.
\item
The case \rulefont{{\sqsubseteq}X'}.\quad
Suppose $S \subseteq \pi^\mone \act T$.
If $\theta \in \f{Sol}(\pi \act X^S, \f{Inc'})$ then $\theta \in \f{Sol}(\f{Inc'})$ and the result follows.
Conversely, suppose $\theta \in \f{Sol}(\f{Inc'})$.
By Claim 5, $\f{fa}((\pi \act X^S)\theta) \subseteq T$.
The result follows.
\end{squishitem}
\end{proof}

\begin{prop}
\label{prop.supp.reduct.strong.normalisation}
Support inclusion problem simplification is strongly normalising.
\end{prop}
\begin{proof}
See Appendix~\ref{sect.omitted.proofs}.
\end{proof}

We conclude with a few useful observations:

\begin{defn}
\label{defn.nf}
For every $\f{Inc}$ make a fixed but arbitrary choice of normal form $\f{nf}(\f{Inc})$, guaranteed to exist by Proposition~\ref{prop.supp.reduct.strong.normalisation}.\footnote{In fact, support inclusion simplification is confluent so $\f{nf}(\f{Inc})$ is also unique.
A proof of confluence is in a technical report~\cite{gabbay:perntu-tr}. For our purposes in this paper, it suffices to know that a normal form exists.}
\end{defn}

\begin{defn}
\label{defn.nf.inc}
Call $\f{Inc}$ \deffont{consistent} when $a\sqsubseteq T\not\in \f{nf}(\f{Inc})$ for all atoms $a$ and permission sets $T$. 
\end{defn}

\begin{lemm}
\label{lemm.consistent.inc.form}
If $\f{Inc}$ is consistent then all $\f{inc}\in \f{nf}(\f{Inc})$ have the form $X^S\sqsubseteq T$ where $S\not\subseteq T$.
\end{lemm}

\subsection{Building solutions for support inclusion problems}
\label{subsect.building.solutions}

Our main results are Theorems~\ref{thrm.when.Inc.consistent} and~\ref{thrm.sigma.rho.sigma}.

\begin{defn}
\label{defn.inc.fv}
Define $\f{fV}(\f{Inc})$ by $\f{fV}(\f{Inc})=\bigcup\{\f{fV}(r)\mid \Exists{T}r\sqsubseteq T\in\f{Inc}\}$.

In words, $\f{fV}(\f{Inc})$ is ``the unknowns appearing in terms appearing in $\f{Inc}$''.
\end{defn}

Recall from Definition~\ref{defn.unknowns} that $\mathcal V$ ranges over finite sets of unknowns. 
\begin{defn}
\label{defn.rho}
Let $\mathcal V$ range over finite sets of unknowns.

Suppose $\f{Inc}$ is consistent.
For every $X^S\in\mathcal V$ make a fixed but arbitrary choice of ${X'}^{S'}$ such that 
${X'}^{S'}\not\in\mathcal V$ 
and 
\begin{equation}
\label{eq.S'}
S'=S\cap\bigcap \{T \mid X^S \sqsubseteq T\in\f{nf}(\f{Inc})\}.
\end{equation}
We make our choice injectively; for distinct $X^S\in\f{fV}(\f{Inc})$ and $Y^T\in\f{fV}(\f{Inc})$, we choose ${X'}^{S'}$ and ${Y'}^{T'}$ distinct.
It will be convenient to write $\Vinc$ for the set of our choices $\{{X'}^{S'}\mid X^S\in\mathcal V\}$.

Define a substitution $\rhoinc$ by:
\begin{frameqn}
\begin{array}{r@{\ }l@{\quad}l}
\rhoinc(X^S)\equiv&\id\act {X'}^{S'}&\text{if }X^S\in\mathcal V
\\
\rhoinc(Y^T)\equiv&\id\act Y^T&\text{otherwise}
\end{array}
\end{frameqn}
\end{defn}

\begin{rmrk}
For example, take $a, b, c \in \mathbb A^< $,
\ $S = \mathbb A^<  \setminus \{ c \}$, 
$T = \mathbb A^<  \setminus \{ a \}$, 
$U = \mathbb A^<  \setminus \{ b \}$,
and $\mathcal{V} = \{ X^S \}$.

It is easy to see that the support reduction problem $\{ X^S \sqsubseteq T,\ X^S \sqsubseteq U \}$ is in normal form.
Let ${X'}^{S'}$ where $S' = \mathbb A^<  \setminus \{ a, b, c \}$ be the fixed but arbitrary choice of fresh variable made in Definition~\ref{defn.rho}.
Then:
$$
\begin{array}{r@{\ }l@{\quad}l}
\rhoinc(X^S) =& \id \act {X'}^{S'} 
&\text{and}
\\
\rhoinc(Y^T) =& \id \act Y^T &\text{for all other $Y^T$}
\end{array}
$$
\end{rmrk}

\begin{rmrk}
$\rhoinc$ is a substitution that ``makes $\f{Inc}$ true on $\mathcal V$''.
Nominal unification does not have this notion because nominal terms unknowns are not permanently labelled with freshness information --- instead, nominal terms unification emits `fresh' freshness conditions.

It is easy to verify that 
$$\f{fa}(\rhoinc(X^S))\subseteq S\quad\text{for all}\quad X^S\in\mathcal V.
$$
In fact, $\rhoinc$ is the most general solution with property; intuitively, all other solutions must factor through $\rhoinc$ on $\mathcal V$.
This is made formal in Theorem~\ref{thrm.sigma.rho.sigma}.
\end{rmrk}

\begin{lemm}
\label{lemm.sigma.solves.Inc}
If $\f{Inc}$ is consistent then $\rhoinc\in\f{Sol}(\f{Inc})$.  (`$\rhoinc$ solves $\f{Inc}$.')
\end{lemm}
\begin{proof}
Suppose $\f{Inc}$ is a $\simpto{}$-normal form.
If $X^S \sqsubseteq T \in \f{Inc}$ then $\rhoinc(X) = \id \act {X'}^{S'}$ for an $S'$ which satisfies $S'\subseteq T$.
The result follows.

More generally, if $\f{Inc}$ is not a $\simpto{}$-normal form, by Theorem~\ref{thrm.simpto.sol.inc} $\f{Sol}(\f{Inc})=\f{Sol}(\f{nf}(\f{Inc}))$, and we use the previous paragraph.
\end{proof}

\begin{thrm}
\label{thrm.when.Inc.consistent}
$\f{Inc}$ is consistent (Definition~\ref{defn.nf.inc}) if and only if $\f{Inc}$ is solvable (Definition~\ref{defn.solves.Inc}).
\end{thrm}
\begin{proof}
By Theorem~\ref{thrm.simpto.sol.inc} $\f{Sol}(\f{Inc}) = \f{Sol}(\f{nf}(\f{Inc}))$, so it suffices to show the result for the case when $\f{Inc}$ is a $\simpto{}$-normal form.

Suppose $\f{Inc}$ is inconsistent, so $\f{nf}(Inc)$ contains a support inclusions of the form $a \sqsubseteq T$ where $a\not\in T$. 
Then $a\theta \equiv a$ always, so there is no substitution $\theta$ such that $a\theta \subseteq T$.
Conversely, if $\f{Inc}$ is consistent, the result follows by Lemma~\ref{lemm.sigma.solves.Inc}.
\end{proof}

\begin{defn}
\label{defn.sigma-rho}
Suppose that $\f{Inc}$ is consistent,\ $\f{fV}(\f{Inc})\subseteq\mathcal V$,\ and $\theta\in \f{Sol}(\f{Inc})$.
Define a substitution $\theta{-}\rhoinc$ by:
\begin{frametxt}
\begin{squishitem}
\item
$(\theta{-}\rhoinc)({X'}^{S'})\equiv \theta(X^S)$ if $X^S\in\mathcal V$ and $\rhoinc(X^S)\equiv\id\act {X'}^{S'}$.
\item
$(\theta{-}\rhoinc)(X^S)\equiv \theta(X^S)$ if $X^S\not\in\mathcal V$.
\end{squishitem}
\end{frametxt}
\end{defn}

We check that Definition~\ref{defn.sigma-rho} is well-defined:

\begin{lemm}
\label{lemm.theta.minus.rhoinc.well.defined}
If $\theta{-}\rhoinc$ exists then it is well-defined.
\end{lemm}
\begin{proof}
Suppose $\theta{-}\rhoinc$ exists.
Then:
\begin{squishitem}
\item
Suppose $X^S \not= Y^T$, $X^S \not\in \mathcal V$ and $Y^T \not\in \mathcal V$.\quad
By Definition~\ref{defn.sigma-rho}, $(\theta{-}\rhoinc)(X^S) \equiv \theta(X^S)$ and $(\theta{-}\rhoinc)(Y^T) \equiv \theta(Y^T)$.
The result follows, as substitutions are well-defined.
\item
Suppose $X'^{S'} \not= Y'^{T'}$, $\rhoinc(X^S) \equiv \id \act X'^{S'}$, $\rhoinc(Y^T) \equiv \id \act Y'^{T'}$ and $X^S, Y^T \not\in \mathcal V$.\quad
Then $(\theta{-}\rhoinc)(X'^{S'}) \equiv \theta(X^S)$ and $(\theta{-}\rhoinc)(Y'^{T'}) \equiv \theta(Y^T)$.
Since $X^S \not= Y^T$, the result follows as substitutions are well-defined.
\item
Suppose $X'^{S'} \not= Y^T$, $\rhoinc(X^S) \equiv \id \act X'^{S'}$, $X^S \not\in \mathcal V$ and $Y^T \in \mathcal V$.\quad
Then $(\theta{-}\rhoinc)(Y^T) \equiv \theta(Y^T)$ and $(\theta{-}\rhoinc)(X'^{S'}) \equiv \theta(X^S)$.
By Definition~\ref{defn.rho}, $\rhoinc(X^S) \not\equiv \id \act Y^T$ as $Y^T \in \mathcal V$.
The result follows as substitutions are well-defined.
\end{squishitem}
The case $Y'^{T'} \not= X^S$, $\rhoinc(Y^T) \equiv \id \act Y'^{T'}$, $Y^T \not\in \mathcal V$ and $X^S \in \mathcal V$ is similar to the case for $X'^{S'} \not= Y^T$, $\rhoinc(X^S) \equiv \id \act X'^{S'}$, $X^S \not\in \mathcal V$ and $Y^T \in \mathcal V$.
\end{proof}

\begin{lemm}
\label{lemm.rhoinc.exists}
If $\theta\in \f{Sol}(\f{Inc})$ then $\rhoinc$ exists.
\end{lemm}
\begin{proof}
By assumption, $\f{Inc}$ is solvable.
By Theorem~\ref{thrm.when.Inc.consistent}, $\f{Inc}$ is consistent.
Using Definition~\ref{defn.rho}, $\rhoinc$ exists.
\end{proof}

\begin{lemm}
\label{lemm.rhoinc.well.defined}
If $\rhoinc$ exists, then it is well-defined.
\end{lemm}
\begin{proof}
Suppose $\rhoinc$ exists and $X^S \not= Y^T$.
Then:
\begin{squishitem}
\item
$X^S \in \mathcal V$ and $Y^T \in \mathcal V$.\quad
Then $\rhoinc(X^S) = \id \act {X'}^{S'}$ and $\rhoinc(Y^T) = \id \act {Y'}^{T'}$.
By Definition~\ref{defn.rho}, ${X'}^{S'}$ and ${Y'}^{T'}$ are chosen so ${X'}^{S'} \not= {Y'}^{T'}$.
The result follows.
\item
$X^S \not\in \mathcal V$ and $Y^T \not\in \mathcal V$.\quad
Then $\rhoinc(X^S) = \id \act X^S$ and $\rhoinc(Y^T) = \id \act Y^T$.
The result follows.
\item
$X^S \in \mathcal V$ and $Y^T \not\in \mathcal V$.\quad
Then $\rhoinc(Y^T) = \id \act Y^T$ and $\rhoinc(X^S) = \id \act {X'}^{S'}$ with ${X'}^{S'} \not\in \mathcal V$.
The result follows.
\item
The case $X^S \not\in \mathcal V$ and $Y^T \in \mathcal V$ is similar to the case for $X^S \in \mathcal V$ and $Y^T \not\in \mathcal V$.
\end{squishitem}
\end{proof}

\begin{lemm}
\label{lemm.supp.reduct.fv.reduce}
If $\f{Inc} \simpto{} \f{Inc'}$ then $\f{fV}(Inc') \subseteq \f{fV}(\f{Inc})$.
\end{lemm}
\begin{proof}
By case analysis on the rules defining $\simpto{}$ (Definition~\ref{defn.supp.inc}).
\end{proof}

\begin{lemm}
\label{lemm.theta.minus.rhoinc.subst}
If $\theta\in \f{Sol}(\f{Inc})$ and $\f{fV}(\f{Inc})\subseteq\mathcal V$ then $\theta{-}\rhoinc$ is a substitution.
\end{lemm}
\begin{proof}
By Lemma~\ref{lemm.rhoinc.exists}, $\rhoinc$ exists.
We show $\f{fa}((\theta{-}\rhoinc)({X'}^{S'}))\subseteq S$ by cases:
\begin{squishitem}
\item
The case $\id\act {X'}^{S'}\equiv \rhoinc(X^S)$ for $X^S\in\mathcal V$.

Using Lemma~\ref{lemm.supp.reduct.fv.reduce}, $\f{fV}(\f{nf}(\f{Inc}))\subseteq\f{fV}(\f{Inc})$.
Then $\f{fV}(\f{nf}(\f{Inc}))\subseteq\mathcal V$, as $\f{fV}(\f{Inc}) \subseteq \mathcal V$ by assumption.
There are two sub-cases:
\begin{squishitem}
\item
The case $X^S\not\in\f{fV}(\f{nf}(\f{Inc}))$.\quad
Then $S=S'$ and $(\theta{-}\rhoinc)({X'}^S)=\theta(X^S)$ by Definition~\ref{defn.sigma-rho}.
By assumption $\f{fa}(\theta(X^S)) \subseteq S$.
The result follows.
\item
The case $X^S\in\f{fV}(\f{nf}(\f{Inc}))$.\quad
By assumption, $\theta\in\f{Sol}(\f{Inc})$ so $\theta\in\f{Sol}(\f{nf}(\f{Inc}))$ using Theorem~\ref{thrm.simpto.sol.inc}.
Then $\f{fa}(\theta(X^S)) \subseteq T$ for every $T$ such that $X^S\sqsubseteq T\in\f{nf}(\f{Inc})$, using Definition~\ref{defn.solves.Inc}.
By Definition~\ref{defn.sigma-rho}, $S' = \bigcap \{T \mid X^S \sqsubseteq T\in\f{nf}(\f{Inc})\}$.
The result follows.
\end{squishitem}
\item
Otherwise, $(\theta{-}\rhoinc)(X^S)\equiv \theta(X^S)$ and $\f{fa}(\theta(X^S))\subseteq S$ by assumption.
\end{squishitem}
\end{proof}

\begin{thrm}
\label{thrm.sigma.rho.sigma}
Suppose $\theta\in \f{Sol}(\f{Inc})$ and $\f{fV}(\f{Inc})\subseteq\mathcal V$.

Then $\theta(X^S) \equiv (\rhoinc \fcomp (\theta{-}\rhoinc))(X^S)$ for every $X^S\in\mathcal V$.
\end{thrm}
\begin{proof}
For some fresh ${X'}^S\not\in\mathcal V$, $\rho(X^S)\equiv \id\act {X'}^S$, and $(\theta{-}\rhoinc)({X'}^S)\equiv\theta(X^S)$.
The result follows by Lemma~\ref{lemm.permutation.comp}.
\end{proof}

\subsection{Support reduction example}
\label{subsect.examples.support.reduction}

Suppose $a, c \in \mathbb A^< $ and $b, d \not\in \mathbb A^< $.
Take $S = \mathbb A^<  = T$, $U = \mathbb A^<  \cup \{ b \}$ and $V = \mathbb A^<  \cup \{ d \}$.
We first run the support reduction algorithm on $\f{Inc} = \{ a \sqsubseteq S,\ \tf{f}([a]a, \id \act Y^T) \sqsubseteq T,\ (c\ a) \act Z^U \sqsubseteq S,\ (b\ a) \act W^V \sqsubseteq T \}$:

\begin{displaymath}
\begin{array}{l@{\hspace{1.25em}}c@{\qquad}l@{\;}l}
\greymath{a \sqsubseteq S},\ \tf{f}([a]a, \id \act Y^T) \sqsubseteq T,\ (c\ a) \act Z^U \sqsubseteq S,\ (b\ a) \act W^V \sqsubseteq T
& \simpto{} &
\rulefont{{\sqsubseteq}a}
\\[1.25ex]
\greymath{\tf{f}([a]a, \id \act Y^T) \sqsubseteq T},\  (c\ a) \act Z^U \sqsubseteq S,\ (b\ a) \act W^V \sqsubseteq T
& \simpto{} &
\rulefont{{\sqsubseteq}\tf f}
\\[1.25ex]
\greymath{[a]a \sqsubseteq T},\ \id \act Y^T \sqsubseteq T,\ (c\ a) \act Z^U \sqsubseteq S,\ (b\ a) \act W^V \sqsubseteq T
& \simpto{} &
\rulefont{{\sqsubseteq}[]}
\\[1.25ex]
\greymath{\id \act Y^T \sqsubseteq T},\ (c\ a) \act Z^U \sqsubseteq S,\ (b\ a) \act W^V \sqsubseteq T
& \simpto{} &
\rulefont{{\sqsubseteq}X'}
\\[1.25ex]
\greymath{(c\ a) \act Z^U \sqsubseteq S},\ (b\ a) \act W^V \sqsubseteq T
& \simpto{} &
\rulefont{{\sqsubseteq}X}
\\[1.25ex]
Z^U \sqsubseteq (c\ a) \act S,\ \greymath{(b\ a) \act W^V \sqsubseteq T}
& \simpto{} &
\rulefont{{\sqsubseteq}X}
\\[1.25ex]
Z^U \sqsubseteq (c\ a) \act S,\ W^V \sqsubseteq (b\ a) \act T
\end{array}
\end{displaymath}

We take $\f{nf}(\f{Inc}) = \{ Z^U \sqsubseteq (c\ a) \act S,\ W^V \sqsubseteq (b\ a) \act T \}$.
By Definition~\ref{defn.nf.inc} we have $\f{nf}(\f{Inc})$ is consistent, therefore by Theorem~\ref{thrm.when.Inc.consistent} a solution for $\f{Inc}$ exists.

We now construct $\rho^{\mathcal{V}}_\f{Inc}$.
Take $W'$ and $Z'$ as our injective choices of fresh unknowns.
Take $U' = U \cap ((c\ a) \act S \cap (b\ a) \act T)$ and $V' = V \cap ((c\ a) \act S \cap (b\ a) \act T)$.
An easy calculation shows that $U' = \mathbb A^<  \setminus \{ a \} = V'$.
We define $\rho^{\mathcal{V}}_\f{Inc}$ piecewise by:
\begin{gather*}
\rho^{\mathcal{V}}_\f{Inc}(Z^U) = \id \act {Z'}^{U'} \;\text{and}\;
\rho^{\mathcal{V}}_\f{Inc}(W^V) = \id \act {W'}^{V'} \;\text{and}\;
\rho^{\mathcal{V}}_\f{Inc}(X^S) = \id \act X^S \;\text{for all other $X^S$}
\end{gather*}
It is easy to verify that this is in fact a substitution (i.e. $\f{fa}(\rho^{\mathcal{V}}_\f{Inc}(X^S)) \subseteq S$ always).

\section{Permissive nominal unification problems}
\label{sect.pernu}

We can now give an algorithm to compute solutions to permissive nominal unification problems (Definition~\ref{defn.problems}) and we prove that our algorithm computes most general solutions.
Note, as we have observed before, that the notion of problem and solution is based just on equalities and substitutions. 

\subsection{The unification algorithm}

\begin{lemm}
\label{lemm.Sol.circ}
$\theta \fcomp \theta' \in \f{Sol}(Pr)$ if and only if $\theta' \in \f{Sol}(Pr\theta)$.
\end{lemm}
\begin{proof}
By unpacking Definition~\ref{defn.unif.sol} and using Lemma~\ref{lemm.sigma.sigma'.circ}.
\end{proof}

\begin{defn}
\label{defn.Prsqsubseteq}
If $Pr$ is a problem, define a support inclusion problem $Pr\tinys\sqsubseteq$ by:
\begin{displaymath}
Pr\tinys\sqsubseteq = \{r\sqsubseteq\f{fa}(s),\ s\sqsubseteq\f{fa}(r) \mid r\ueq s\in Pr\}
\end{displaymath}
Call a support inclusion problem $\f{Inc}$ \deffont{non-trivial} when $\f{nf}(\f{Inc})\neq\varnothing$.
\end{defn}

\begin{figure}
\begin{displaymath}
\begin{array}{l@{\quad}l@{\hspace{-1em}}c@{\ }l}
\rulefont{{\ueq}a} & \mathcal V; a\ueq a,\, Pr & \simpto{}& \mathcal V; Pr \\
\rulefont{{\ueq}\tf f} & \mathcal V; \tf f(r_1,\, \ldots) \ueq \tf f(s_1,\, \ldots),\, Pr & \simpto{} & \mathcal V; r_1\ueq s_1,\, \ldots,\, Pr \\
\rulefont{{\ueq}[a]} & \mathcal V; [a]r\ueq [a]s,\, Pr & \simpto{} & \mathcal V; r\ueq s,\, Pr \\
\rulefont{{\ueq}[b]} & \mathcal V; [a]r\ueq [b]s,\, Pr & \simpto{} & \mathcal V; (b\ a) \act r \ueq s,\, Pr \\
& & & \qquad (b\not\in\f{fa}(r)) \\
\rulefont{{\ueq}X} & \mathcal V; \pi\act X^S\ueq \pi\act X^S,\, Pr &\simpto{} & \mathcal V; Pr \\
\rulefont{I1} & \mathcal V; \pi\act X^S\ueq s,\, Pr  & \simpto{[X^S\sm\pi^\mone\act s]} & \mathcal V; Pr[X^S\sm \pi^\mone\act s] \\
& & & \qquad(X^S\not\in\f{fV}(s),\ \f{fa}(s)\subseteq\pi\act S) \\
\rulefont{I2} & \mathcal V; r\ueq \pi\act X^S,\, Pr  & \simpto{[X^S\sm\pi^\mone\act r]} & \mathcal V; Pr[X^S\sm \pi^\mone\act r] \\
& & & \qquad(X^S\not\in\f{fV}(r),\ \f{fa}(r)\subseteq\pi\act S) \\
\rulefont{I3} & \mathcal V; Pr  & \simpto{\rho\smallss{\mathcal V}{Pr\tinys\sqsubseteq}} & \mathcal V\cup{\mathcal V'}\smallss{\mathcal V}{Pr\tinys{\sqsubseteq}}; Pr\ \rho\smallss{\mathcal V}{Pr\tinys\sqsubseteq} 
\\
& & & \qquad(Pr\tinys{\sqsubseteq}\text{ consistent and non-trivial})
\end{array}
\end{displaymath}
\caption{Simplification rules for problems}
\label{fig.simr}
\end{figure}

\begin{defn}
\label{defn.unif}
Define a \deffont{simplification} rewrite relation $\mathcal V;Pr\simpto{} \mathcal V';Pr'$ on unification problems by the rules in Figure~\ref{fig.simr}.

Call \rulefont{{\ueq}a}, \rulefont{{\ueq}\tf f}, \rulefont{{\ueq}[a]}, \rulefont{{\ueq}[b]}, and \rulefont{{\ueq}X} \deffont{non-instantiating rules}.
Call \rulefont{I1}, \rulefont{I2}, and \rulefont{I3} \deffont{instantiating rules}.
Write $\simptostar{}$ for the transitive and reflexive closure of $\simpto{}$.
\end{defn}

In \rulefont{I3} we insist that $Pr\tinys\sqsubseteq$ is non-trivial to avoid indefinite rewrites.
We insist $Pr\tinys\sqsubseteq$ is consistent so that $\rho\smallss{\mathcal V}{Pr\tinys\sqsubseteq}$ exists.
$\rho\smallss{\mathcal V}{Pr\tinys\sqsubseteq}$ and  
(${\mathcal V'}\smallss{\mathcal V}{Pr\tinys{\sqsubseteq}}$ are defined in Definition~\ref{defn.rho}.)

\begin{defn}
\label{defn.fV.Pr}
Define
$\f{fV}(Pr)=\bigcup\{\f{fV}(r)\cup\f{fV}(s)\mid r\ueq s\in Pr\}$.
\end{defn}

\begin{defn}
\label{defn.restricted}
Suppose $\mathcal V$ is a set of unknowns.
Define $\theta|_{\mathcal V}$ by:
\begin{frameqn}
\begin{array}{r@{\ }l@{\quad}l}
\theta|_{\mathcal V}(X)\equiv&\theta(X) &\text{if }X\in \mathcal V
\\
\theta|_{\mathcal V}(X)\equiv&\id\act X& \text{otherwise}
\end{array}
\end{frameqn}
\end{defn}
\noindent(We overload $|$, for technical convenience: $\pi|_S$ is partial and $\theta|_{\mathcal V}$ is total.)

\begin{defn}
\label{defn.unification.algorithm}
If $Pr$ is a problem, define a \deffont{unification algorithm} by: 
\begin{frametxt}
\begin{enumerate}
\item
Rewrite $\f{fV}(Pr);Pr$ using the rules of Definition~\ref{defn.unif} as much as possible, with top-down precedence (so we apply \rulefont{{\ueq}a} before \rulefont{{\ueq}\tf f}, and so on down to \rulefont{I3}).
\item
If we reduce to $\mathcal V';\varnothing$, we succeed and return $\theta|_{\mathcal V}$ where $\theta$ is the functional composition of all the substitutions labelling rewrites (we take $\theta=\id$ if there are none).
Otherwise, we fail.
\end{enumerate}
\end{frametxt}
\end{defn}

\begin{prop}
\label{prop.unification.algorithm.strong.normalisation}
The algorithm of Definition~\ref{defn.unification.algorithm} always terminates.
\end{prop}
\begin{proof}
See Appendix~\ref{sect.omitted.proofs}.
\end{proof}

\subsection{Examples of the algorithm}

\subsubsection*{Example one.}

Suppose $a, c \in \mathbb A^< $ and $d \not\in \mathbb A^< $.
Take $S = \mathbb A^<$.
Take $\mathcal{V} = \{ X^S \}$.
Suppose a term-former $\tf{g}$.
We apply the algorithm to  $\{ \tf{g}([a](\id \act X^S), [a]a) \ueq \tf{g}([d]c, [d]d) \}$:

\begin{displaymath}
\begin{array}{l@{\hspace{1.25em}}c@{\qquad}l@{\;}l}
\mathcal{V};\ \greymath{\tf{g}([a](\id \act X^S), [a]a) \ueq \tf{g}([d]c, [d]d)}
& \simpto{} &
\rulefont{{\ueq}\tf f}
\\[1.25ex]
\mathcal{V};\ \greymath{[a](\id \act X^S) \ueq [d]c}, [a]a \ueq [d]d
& \simpto{} &
\rulefont{{\ueq}[b]}
\\[1.25ex]
\mathcal{V};\ \greymath{(d\ a) \act X^S \ueq c}, [a]a \ueq [d]d
& \simpto{[X^S \ssm c]} &
\rulefont{I1}
\\[1.25ex]
\mathcal{V};\ \greymath{[a]a \ueq [d]d}
& \simpto{} &
\rulefont{{\ueq}[b]}
\\[1.25ex]
\mathcal{V};\ \greymath{(d\ a) \act a \ueq d}
& \simpto{} &
\rulefont{{\ueq}a}
\\[1.25ex]
\mathcal{V};\ \varnothing \qquad\text{(Success!)}
\end{array}
\end{displaymath}

The algorithm succeeds and returns the substitution $[X^S \ssm c]$.

\subsubsection*{Example two.}

Suppose $a,c\in\mathbb A^< $ and $b,d\not\in\mathbb A^< $.
Take $S = \mathbb A^<  \cup \{ b, d \}$ and $T = \mathbb A^<  \cup \{ f \}$ and $U = \mathbb A^< $.
Take $\mathcal{V} = \{ X^S, Y^T \}$. 
Suppose a term-former $\tf f$.

We apply the algorithm to $\{\tf{f}([a]b, \id \act Z^U, \id \act X^S)\ueq \tf{f}([d]b, [a]a, \id \act Y^T)\}$:
\begin{displaymath}
\begin{array}{l@{\hspace{-2.5em}}c@{\qquad}l@{\;}l}
\mathcal{V};\ \greymath{\tf{f}([a]b,  \id \act Z^U,  \id \act X^S) \ueq \tf{f}([d]b,  [a]a,  \id \act Y^T)}
&
\simpto{}
&
\rulefont{{\ueq}\tf{f}} 
\\[1.25ex]
\mathcal{V};\ \greymath{[a]b \ueq [d]b},\  \id \act Z^U \ueq [a]a,\  \id \act X^S \ueq \id \act Y^T 
&
\simpto{}
&
\rulefont{{\ueq}[b]} 
\\[1.25ex]
\mathcal{V};\ \greymath{(d\ a) \act b \ueq b},\  \id \act Z^U \ueq [a]a,\  \id \act X^S \ueq \id \act Y^T 
&
\simpto{} 
&
\rulefont{{\ueq}a} 
\\[1.25ex]
\mathcal{V}; \greymath{\id \act Z^U \ueq [a]a},\  \id \act X^S \ueq \id \act Y^T 
&
\simpto{[Z^U \ssm [a]a]} 
&
\rulefont{{\ueq}I1} 
\\[1.25ex]
\mathcal{V};\ \greymath{\id \act X^S \ueq \id \act Y^T} 
&
\simpto{[X^S \ssm X'^{S'}][Y^T \ssm Y'^{T'}]} 
&
\rulefont{I3} 
\\[1.25ex]
\mathcal{V} \cup \{ X'^{S'},  Y'^{T'} \};\ \greymath{\id \act X'^{S'} \ueq \id \act Y'^{T'}} 
&
\simpto{[X'^{S'} \ssm Y'^{T'}]} 
&
\rulefont{I1} 
\\[1.25ex]
\mathcal{V} \cup \{ X'^{S'},  Y'^{T'} \};\ \greymath{\id \act Y'^{T'} \ueq \id \act Y'^{T'}} 
&
\simpto{} 
&
\rulefont{{\ueq}X} 
\\[1.25ex]
\mathcal{V} \cup \{ X'^{S'},  Y'^{T'} \};\ \varnothing 
\qquad (\text{Success!})
\end{array}
\end{displaymath}
Here $X'$ and $Y'$ are the choice of unknown made in Definition~\ref{defn.rho},
and $S' = \mathbb A^<  = T'$.

The algorithm succeeds and returns the substitution 
$$
[X'^{S'} \ssm Y'^{T'}] \fcomp [X^S \ssm X'^{S'}] \fcomp [Y^T \ssm Y'^{T'}] \fcomp [Z^U \ssm [a]a] .
$$

\subsubsection*{Example three.}

An example that fails to unify.
Take $S = \mathbb A^< $.
Take $\mathcal{V} = \{ X^S \}$.
We run the algorithm on $\{ [a]([b](\id \act X^S)) \ueq[a](\id \act X^S) \}$:

\begin{displaymath}
\begin{array}{l@{\hspace{1.25em}}c@{\qquad}l@{\;}l}
\mathcal{V};\ \greymath{[a]([b](\id \act X^S)) \ueq [a](\id \act X^S)}
& \simpto{} &
\rulefont{{\ueq}[a]}
\\[1.25ex]
\mathcal{V};\ \greymath{[b](id \act X^S) \ueq \id \act X^S} \qquad\text{(Failure!)}
\end{array}
\end{displaymath}

The algorithm fails as the precondition of rule \rulefont{I2}, $X^S \not\in \f{fV}([b](id \act X^S))$, the `occurs check', fails to hold.
By Theorem~\ref{thrm.algorithm.correctness} there is no solution to the unification problem.

\subsection{Preservation of solutions}

\begin{lemm} 
\label{lemm.preservation.of.solutions}
If $\mathcal V; Pr \simpto{} \mathcal V; Pr'$ by a non-instantiating rule (Definition~\ref{defn.unif})
then $\f{Sol}(Pr)=\f{Sol}(Pr')$.
\end{lemm}
\begin{proof}
See Appendix~\ref{sect.omitted.proofs}.
We use Lemmas~\ref{lemm.substitution} and~\ref{lemm.sub.perm}, and Proposition~\ref{prop.aeq.transitive}.
\end{proof}

\begin{lemm}
\label{lemm.sigma.seq.Pr}
Suppose $\theta(X^S)\aeq \theta'(X^S)$ for all $X^S\in\f{fV}(Pr)$. 
Then $\theta\in\f{Sol}(Pr)$ if and only if $\theta'\in\f{Sol}(Pr)$.
\end{lemm}
\begin{proof}
Unpacking Definition~\ref{defn.unif.sol} it suffices to show that $r\theta\aeq s\theta$ if and only if $r\theta'\aeq s\theta'$, for every $r\ueq s\in Pr$.
This is easy using Lemma~\ref{lemm.eq.X.eq.always} and the fact by construction (Definition~\ref{defn.fV.Pr}) that $\f{fV}(r)\subseteq\f{fV}(Pr)$ and $\f{fV}(s)\subseteq\f{fV}(Pr)$.
\end{proof}

\begin{defn}
\label{defn.sigma-X}
Write $\theta{-}X^S$ for the substitution such that
\begin{frameqn}
\begin{array}{r@{\ }l@{\quad}l}
(\theta{-}X^S)(X^S)\equiv&\id\act X^S
&\text{and}\quad
\\
(\theta{-}X^S)(Y^T)\equiv&\theta(Y^T)&\text{for all other }Y^T.
\end{array}
\end{frameqn}
\end{defn}

In the right circumstances, a substitution $\theta$ can be factored as `a part of $\theta$ that does not touch $X^S$' and `a single substitution for $X^S$':
\begin{thrm}
\label{thrm.single.out.X}
Suppose $X^S\theta\aeq s\theta$ and $X^S\not\in\f{fV}(s)$.
Then
$$
\begin{array}{r@{\ }l}
X^S\theta \aeq& X^S([X^S\sm s] \fcomp (\theta{-}X^S))
\quad\text{and}
\\
Y^T\theta \aeq& Y^T([X^S\sm s] \fcomp (\theta{-}X^S)) .
\end{array}
$$
\end{thrm}
\begin{proof}
We reason as follows:
\begin{calcenv}
X^S([X^S\sm s] \fcomp (\theta{-}X^S)) & \equiv & s(\theta{-}X^S) & \text{Definition~\ref{defn.subst.action}} \\
                                      &\equiv & s\theta & X^S\not\in\f{fV}(s),\ \text{Lemma~\ref{lemm.eq.X.eq.always}} \\
                                      &\aeq & X^S\theta & \text{Assumption} \\\\
Y^T([X^S\sm s] \fcomp (\theta{-}X^S)) & \equiv & Y^T(\theta{-}X^S) & \text{Definition~\ref{defn.sigma.com}} \\
                                      &\equiv & Y^T\theta & \text{Definition~\ref{defn.sigma-X}}
\end{calcenv}
\end{proof}

\subsection{Simplification rewrites calculate principal solutions} 

\begin{defn}
\label{defn.instantiation.ordering}
Write $\theta_1\leq\theta_2$ when there exists some $\theta'$ such that $X^S\theta_2\aeq X^S(\theta_1\circ\theta')$ always.
Call $\leq$ the \deffont{instantiation ordering}. 
\end{defn}

\begin{defn}
\label{def.principal.solution}
A \deffont{principal} (or \deffont{most general}) solution to a problem $Pr$ is a solution $\theta\in\f{Sol}(Pr)$ such that $\theta\leq\theta'$ for all other $\theta'\in\f{Sol}(Pr)$.
\end{defn}

Our main results are Theorems~\ref{thrm.sol.is.solution} --- the unification algorithm from Definition~\ref{defn.unification.algorithm} calculates a solution --- and~\ref{thrm.sol.is.principal} --- the solution it calculates, is principal.

\begin{lemm}
\label{lemm.simp.fv.pres}
If $\f{fV}(Pr) \subseteq \mathcal V$ and $\mathcal V; Pr \simpto{} \mathcal V'; Pr'$ using a non-instantiating rule, then $\f{fV}(Pr') \subseteq \mathcal V$.
\end{lemm}
\begin{proof}
See Appendix~\ref{sect.omitted.proofs}.
\end{proof}

\begin{lemm}
\label{lemm.subst.fV}
$\f{fV}(r[X^S \ssm s]) \subseteq \f{fV}(r) \cup \f{fV}(s)$.
\end{lemm}
\begin{proof}
By induction on $r$.
\end{proof}

\begin{lemm}
\label{lemm.inst.rule.fv.pres}
If $\f{fV}(Pr) \subseteq \mathcal V$ and $\mathcal V; Pr \simpto{\theta} \mathcal V'; Pr'\theta$ using an instantiating rule, then $\f{fV}(Pr'\theta) \subseteq \mathcal V'$.
\end{lemm}
\begin{proof}
There are two cases:
\begin{squishitem}
\item
The case \rulefont{I1}.\quad
Suppose $\f{fV}(\pi \act X^S \ueq s, Pr') \subseteq \mathcal V$ and $\mathcal V; \pi \act X^S \ueq s, Pr' \simpto{[X^S \ssm \pi^\mone \act s]} \mathcal V'; Pr'[X^S \ssm \pi^\mone \act s]$ using \rulefont{I1}.
Using Definition~\ref{defn.fV.Pr} and Lemma~\ref{lemm.subst.fV}, we have $\f{fV}(Pr'[X^S \ssm \pi^\mone \act s]) \subseteq \f{fV}(Pr') \cup \f{fV}(\pi^\mone \act s)$.
The result follows.
\item
The case \rulefont{I3}.\quad
A corollary of Lemma~\ref{lemm.supp.reduct.fv.reduce}.
\end{squishitem}
\end{proof}

\begin{lemm}
\label{lemm.rest.fcomp.move}
If $X^S \in \mathcal V$ then $([X^S \ssm s] \fcomp \theta)|_{\mathcal V} = [X^S \ssm s] \fcomp (\theta|_{\mathcal V})$
\end{lemm}
\begin{proof}
We consider cases:
\begin{squishitem}
\item
The case $X^S$ with $X^S \in \mathcal V$.\quad
We reason as follows:
\begin{calcenv}
([X^S \ssm s] \fcomp \theta)|_\mathcal V(X^S) & \equiv & ([X^S \ssm s] \fcomp \theta)(X^S) & \text{Definition~\ref{defn.restricted}, $X^S \in \mathcal V$} \\
                                              & \equiv & s\theta & \text{Definition~\ref{defn.sigma.com}} \\
                                              & \equiv & s\theta|_\mathcal V & \text{Definition~\ref{defn.restricted}}
\end{calcenv}
\item
The case $Y^T$ with $Y^T \in \mathcal V$.\quad
We reason as follows:
\begin{calcenv}
([X^S \ssm s] \fcomp \theta)|_\mathcal V(Y^T) & \equiv & ([X^S \ssm s] \fcomp \theta)(Y^T) & \text{Definition~\ref{defn.restricted}, $Y^T \in \mathcal V$} \\
                                              & \equiv & \theta(Y^T) & \text{Definition~\ref{defn.sigma.com}} \\
                                              & \equiv & \theta|_{\mathcal V}(Y^T) & \text{Definition~\ref{defn.restricted}}
\end{calcenv}
\item
The case $Y^T$ with $Y^T \not\in \mathcal V$.\quad
Since $([X^S \ssm s] \fcomp \theta)|_{\mathcal V}(Y^T) \equiv \id \act Y^T$ and $\theta|_{\mathcal V} \equiv \id \act Y^T$.
\end{squishitem}
\end{proof}

Recall that $\theta|_{\mathcal V}$ is defined in Definition~\ref{defn.restricted}:
\begin{thrm}
\label{thrm.sol.is.solution}
If $\f{fV}(Pr)\subseteq\mathcal V$ then $\mathcal V;Pr \simptostar{\theta} \mathcal V';\varnothing$ implies $\theta|_{\mathcal V} \in \f{Sol}(Pr)$.
\end{thrm}
\begin{proof}
By induction on the length of the path in $\simptostar{\theta}$.
\begin{squishitem}
\item
\emph{Length $0$.}\quad
Then $Pr=\varnothing$ and $\theta \equiv \id$. 
The result follows.
\item
\emph{Length $k+1$.}\quad
There are three cases: 
\begin{squishitem}
\item
\emph{The non-instantiating case.}\quad
Suppose $\mathcal V;Pr \simpto{} \mathcal V;Pr'' \simptostar{\theta} \mathcal V';\varnothing$.
Using Lemma~\ref{lemm.simp.fv.pres}, $\f{fV}(Pr'')\subseteq\mathcal V$ and $\theta\in\f{Sol}(Pr'')$ by inductive hypothesis.
Using Lemma~\ref{lemm.preservation.of.solutions}, $\theta\in\f{Sol}(Pr)$, and the result follows.
\item
\emph{The case of \rulefont{I1} or \rulefont{I2}.}\quad
Suppose $\mathcal V;Pr \simpto{\chi} \mathcal V; Pr\chi \simptostar{\theta'} \mathcal V';\varnothing$.
Using Lemma~\ref{lemm.inst.rule.fv.pres}, $\f{fV}(Pr\chi) \subseteq \mathcal V$.
By inductive hypothesis $\theta'|_{\mathcal V} \in \f{Sol}(Pr\chi)$.
By Lemma~\ref{lemm.rest.fcomp.move}, $(\chi \fcomp \theta')|_{\mathcal V} = \chi \fcomp (\theta'|_{\mathcal V})$.
By Lemma~\ref{lemm.Sol.circ}, $(\chi \fcomp \theta')|_{\mathcal V}\in\f{Sol}(Pr)$.
\item
\emph{The case of \rulefont{I3}.}\quad
Suppose $\mathcal V;Pr \simpto{\rho} \mathcal V';Pr\rho \simptostar{\theta'} \mathcal V'';\varnothing$. Using Lemma~\ref{lemm.inst.rule.fv.pres}, $\f{fV}(Pr\rho)\subseteq\mathcal V'$.
By inductive hypothesis $\theta'|_{\mathcal V'} \in \f{Sol}(Pr\rho)$.
By Lemma~\ref{lemm.Sol.circ}, $\rho \fcomp (\theta'|_{\mathcal V'}) \in \f{Sol}(Pr)$.
By Lemma~\ref{lemm.rest.fcomp.move}, $\rho \fcomp (\theta'|_{\mathcal V'})=(\rho \fcomp \theta')|_{\mathcal V'}$.
By Lemma~\ref{lemm.sigma.seq.Pr},\ $(\rho \fcomp \theta')|_{\mathcal V}\in\f{Sol}(Pr)$. 
\end{squishitem}
\qedhere
\end{squishitem}
\end{proof}

\begin{lemm}
\label{lemm.add.substitution.on.left}
If $\theta_1\leq \theta_2$ then $\theta \fcomp \theta_1 \leq \theta \fcomp \theta_2$.
\end{lemm}
\begin{proof}
By Definition~\ref{defn.instantiation.ordering}, $\theta'$ exists such that $X^S\theta_2 \aeq X^S(\theta_1 \fcomp \theta')$ always.
Then:
\begin{calcenv}
X^S(\theta\fcomp\theta_2) & \equiv & (X^S\theta)\theta_2 & \text{Lemma~\ref{lemm.sigma.sigma'.circ}} \\
                          & \aeq & (X^S\theta)(\theta_1\fcomp\theta') & \text{Lemma~\ref{lemm.eq.X.eq.always}} \\
                          & \equiv & X^S((\theta\fcomp\theta_1)\fcomp\theta') &\text{Lemma~\ref{lemm.sigma.sigma'.circ}}
\qedhere
\end{calcenv}
\end{proof}
 
\begin{lemm}
\label{lemm.leq.equiv}
Suppose $X^S\theta_2\aeq X^S\theta_2'$ always. 
Then $\theta_1\leq\theta_2$ implies $\theta_1\leq\theta_2'$.
\end{lemm}
\begin{proof}
By a routine calculation using Definition~\ref{defn.instantiation.ordering} and using Proposition~\ref{prop.aeq.transitive}.
\end{proof}

\begin{lemm}
\label{lemm.Sol.Pr.Sol.Inc}
If $\theta\in\f{Sol}(Pr)$ (Definition~\ref{defn.unif.sol}) then $\theta\in\f{Sol}(Pr\tinys{\sqsubseteq})$ (Definition~\ref{defn.solves.Inc}).
\end{lemm}
\begin{proof}
By a routine calculation, using Definitions~\ref{defn.unif.sol} and~\ref{defn.Prsqsubseteq}, and Lemma~\ref{lemm.aeq.ftma.pres}.
\end{proof}

\begin{lemm}
\label{lemm.theta.transfer.v}
If $X^S \in \mathcal V$ then $(\theta|_{\mathcal V} - X^S) = (\theta{-}X^S)|_{\mathcal V}$.
\end{lemm}
\begin{proof}
We consider cases:
\begin{squishitem}
\item
The case $X^S$.\quad
Then $(\theta|_{\mathcal V} - X^S)(X^S) = \id \act X^S$ and $(\theta{-}X^S)|_{\mathcal V}(X^S) = \id \act X^S$.
The result follows.
\item
The case $Y^T$ with $Y^T \not\in \mathcal V$.\quad
Then $(\theta|_{\mathcal V} - X^S)(Y^T) = \id \act Y^T$ and $(\theta{-}X^S)|_{\mathcal V}(Y^T) = \id \act Y^T$.
The result follows.
\item
The case $Y^T$ with $Y^T \in \mathcal V$.\quad
Then $(\theta|_{\mathcal V} - X^S)(Y^T) = \theta|_\mathcal{V}(Y^T)$ and $(\theta{-}X^S)|_{\mathcal V}(Y^T) = \theta(Y^T)$.
As $\theta|_\mathcal V(Y^T) = \theta(Y^T)$ when $Y^T \in \mathcal V$, the result follows.
\end{squishitem}
\end{proof}

\begin{thrm}
\label{thrm.sol.is.principal}
Suppose $\f{fV}(Pr) \subseteq \mathcal V$.

If $\mathcal V;Pr \simptostar{\theta} \mathcal V';\varnothing$ then $\theta|_{\mathcal V}$ is a principal solution to $Pr$ (Definition~\ref{def.principal.solution}).
\end{thrm}
\begin{proof}
By Theorem~\ref{thrm.sol.is.solution}, $\theta|_{\mathcal V} \in \f{Sol}(Pr)$.
We prove $\theta|_{\mathcal V}$ is principal by induction on the path length of $\mathcal V;Pr \simptostar{\theta} \mathcal V';\varnothing$.
\begin{squishitem}
\item
\emph{Length $0$.}\quad
So $Pr = \varnothing$ and $\theta = \id|_{\mathcal V}$.
By Definition~\ref{defn.instantiation.ordering}, $\id|_{\mathcal V} \leq \theta'|_{\mathcal V}$. 
\item
\emph{Length $k+1$.}\quad
We consider the rules in Definition~\ref{defn.unif}.
\begin{squishitem}
\item
The non-instantiating case.\quad
Suppose
\begin{displaymath}
\mathcal V;Pr \simpto{} \mathcal V;Pr' \simptostar{\theta} \mathcal V';\varnothing
\end{displaymath}
where $\mathcal V;Pr \simpto{} \mathcal V;Pr'$ is a non-instantiating simplification rewrite.
By inductive hypothesis $\theta|_{\mathcal V}$ is a principal solution of $Pr'$.
By Lemma~\ref{lemm.preservation.of.solutions} $\theta|_{\mathcal V}$ is a principal solution of $Pr$. 
The result follows.
\item
The case \rulefont{I1}.\quad
Suppose $\f{fa}(s) \subseteq \pi\act S$ and $X^S \not\in \f{fV}(s)$.
Write $\chi = [X^S\sm\pi^\mone\act s]$. 
Suppose ${Pr\ =\ \pi {\act} X^S  \ueq  s,\; Pr''}$ so that
\begin{displaymath}
\mathcal V;\pi \act X^S \ueq s,\ Pr'' \simpto{\chi} \mathcal V;Pr''\chi \simptostar{\theta''} \mathcal V';\varnothing.
\end{displaymath}
Further, suppose that $\theta'|_{\mathcal V} \in \f{Sol}(Pr)$.

By assumption $(\pi\act X)\theta'\aeq s\theta'$, so 
by Lemma~\ref{lemm.sub.perm} $\pi\act \theta'(X)\aeq s\theta'$.
By Lemmas~\ref{lemm.equality.permutation} and~\ref{lemm.permutation.comp} $\theta'(X)\aeq \pi^\mone\act (s\theta')$, and by Lemma~\ref{lemm.sub.perm} $\theta'(X)\aeq (\pi^\mone\act s)\theta'$.
It follows by Theorem~\ref{thrm.single.out.X} and Lemma~\ref{lemm.sigma.seq.Pr} that $\chi \fcomp (\theta'|_{\mathcal V} - X^S) \in \f{Sol}(Pr)$.
Using Lemma~\ref{lemm.theta.transfer.v}, $(\theta'|_{\mathcal V} - X^S) = (\theta' - X^S)|_{\mathcal V}$.
By Lemma~\ref{lemm.Sol.circ}, $(\theta' - X^S)|_{\mathcal V} \in \f{Sol}(Pr''\chi)$. 

By Theorem~\ref{thrm.sol.is.solution}, $\theta''|_{\mathcal V} \in \f{Sol}(Pr''\chi)$.
By Lemma~\ref{lemm.inst.rule.fv.pres}, $\f{fV}(Pr''\chi) \subseteq \mathcal V$.

By inductive hypothesis $\theta''|_{\mathcal V} \leq (\theta'{-}X^S)|_{\mathcal V}$.
By Lemma~\ref{lemm.add.substitution.on.left}, $\chi \fcomp (\theta''|_{\mathcal V}) \leq \chi \fcomp (\theta'{-}X^S)|_{\mathcal V}$.
Now by assumption ${\f{fV}(s)\subseteq\mathcal V}$ and ${X^S\in\mathcal V}$.
Using Lemma~\ref{lemm.rest.fcomp.move} it follows that $\chi \fcomp (\theta''|_{\mathcal V}) = (\chi \fcomp \theta'')|_{\mathcal V}$.
By Lemma~\ref{lemm.theta.transfer.v}, $(\theta'{-}X^S)|_{\mathcal V} = \theta'|_{\mathcal V} {-} X^S$.
By Theorem~\ref{thrm.single.out.X} and Lemma~\ref{lemm.leq.equiv},\ $(\chi \fcomp \theta'')|_{\mathcal V} \leq \theta'|_{\mathcal V}$ as required.
\item
The case \rulefont{I2} is similar to the case of \rulefont{I1}.
\item
The case \rulefont{I3}.\quad
Suppose $Pr\tinys{\sqsubseteq}$ is consistent and non-trivial.
Write $\rho = \rho\smallss{\mathcal V}{Pr\tinys{\sqsubseteq}}$, so that
\begin{displaymath}
\mathcal V;Pr\simpto{\rho}\mathcal V'';Pr\rho \simptostar{\theta''} \mathcal V';\varnothing,
\end{displaymath}
and suppose that $\theta'|_{\mathcal V} \in\f{Sol}(Pr)$.

By Theorem~\ref{thrm.sol.is.solution}, $\theta''|_{\mathcal V''} \in \f{Sol}(Pr\rho)$.
It is a fact that $\mathcal V''=\mathcal V\cup{\mathcal V'}\smallss{\mathcal V}{Pr\tinys{\sqsubseteq}}$, so $\f{fV}(Pr\rho)\subseteq\mathcal V''$.
By Lemma~\ref{lemm.Sol.Pr.Sol.Inc}, $\theta'|_{\mathcal V} \in \f{Sol}(Pr\tinys\sqsubseteq)$.
By Theorem~\ref{thrm.sigma.rho.sigma} and Lemma~\ref{lemm.sigma.seq.Pr}, $\rho \fcomp (\theta'|_{\mathcal V} {-}\rho)\in\f{Sol}(Pr)$.
By Lemma~\ref{lemm.Sol.circ},\ $\theta'|_{\mathcal V} {-} \rho\in\f{Sol}(Pr\rho)$.

By inductive hypothesis $\theta''|_{\mathcal V} \leq \theta'|_{\mathcal V} {-}\rho$.
By Lemma~\ref{lemm.add.substitution.on.left}, $\rho \fcomp \theta''|_{\mathcal V} \leq \rho \fcomp (\theta'|_{\mathcal V} {-} \rho)$.
It is a fact that $\rho\fcomp(\theta''|_{\mathcal V}) = (\rho \fcomp \theta'')|_{\mathcal V}$.
By Theorem~\ref{thrm.sigma.rho.sigma} and Lemma~\ref{lemm.leq.equiv}, $(\rho \fcomp\theta'')|_{\mathcal V} \leq \theta'|_{\mathcal V}$ as required.
\end{squishitem}
\end{squishitem}
\end{proof}

\begin{lemm}
\label{lemm.missing.link}
\begin{enumerate}
\item
Suppose $\f{fa}(s){\subseteq} \pi\act S$ and $X^S\not\in\f{fV}(s)$.
Write $\chi{=}[X^S\sm\pi^\mone\act s]$. 

If $\mathcal V;Pr\simpto{\chi}\mathcal V;Pr'$ with \rulefont{I1} or \rulefont{I2} then $\theta\in\f{Sol}(Pr)$ implies $\theta{-}X^S\in\f{Sol}(Pr')$.
\item
If $\mathcal V;Pr\simpto{\rho}\mathcal V';Pr'$ with \rulefont{I3} then $\theta\in\f{Sol}(Pr)$ implies $\theta{-}\rho\in\f{Sol}(Pr')$.
\end{enumerate}
\end{lemm}
\begin{proof}
\begin{enumerate}
\item
We consider the case of \rulefont{I1}; the case of \rulefont{I2} is similar.
Suppose
$Pr=\pi \act X^S \ueq s,\ Pr''$
so that
$\mathcal V;\pi \act X^S \ueq s,\ Pr'' \simpto{\chi} \mathcal V;Pr''\chi$.
Now suppose $\theta\in\f{Sol}(Pr)$.
By Lemma~\ref{lemm.sigma.seq.Pr} and Theorem~\ref{thrm.single.out.X}, $\chi\fcomp(\theta{-}X^S))\in\f{Sol}(Pr)$.
By Lemma~\ref{lemm.Sol.circ}, $\theta{-}X^S\in\f{Sol}(Pr\chi)$.
It follows that $\theta{-}X^S\in\f{Sol}(Pr''\chi)$ as required.
\item
Suppose $Pr\tinys{\sqsubseteq}$ is consistent and non-trivial.
Write $\rho= \rho\smallss{\mathcal V}{Pr\tinys{\sqsubseteq}}$,
so that
$
\mathcal V;Pr\simpto{\rho}\mathcal V'';Pr\rho$. 
Now suppose $\theta\in\f{Sol}(Pr)$.
By Lemma~\ref{lemm.sigma.seq.Pr} and Theorem~\ref{thrm.sigma.rho.sigma}, $\rho\fcomp(\theta{-}\rho)\in\f{Sol}(Pr)$.
By Lemma~\ref{lemm.Sol.circ}, $\theta{-}\rho\in\f{Sol}(Pr\rho)$ as required.
\end{enumerate}
\end{proof}

\begin{thrm}
\label{thrm.algorithm.correctness}
Given a problem $Pr$, if the algorithm of Definition~\ref{defn.unification.algorithm} succeeds then it returns a principal solution; if it fails then there is no solution.
\end{thrm}
\begin{proof}
If the algorithm succeeds we use Theorem~\ref{thrm.sol.is.principal}.
Otherwise, the algorithm generates 
an element of the form $\tf f(r_1,\ldots,r_n)\ueq\tf f(r_1',\ldots,r_{n'}')$ where $n\neq n'$,\  
$\tf f(\ldots)\ueq \tf g(\ldots)$,\ $\tf f(\ldots)\ueq [a]s$,\ $\tf f(\ldots)\ueq a$,\ $[a]r\aeq a$,\ $[a]r\aeq b$,\ $a\ueq b$,\ a $Pr$ such that $Pr\tinys{\sqsubseteq}$ is inconsistent, or $\pi\act X^S\ueq r$ or $r\ueq\pi\act X^S$ where $X^S\in\f{fV}(r)$.
By arguments on syntax and size of syntax, no solution to the reduced problem exists. 
It follows by Lemma~\ref{lemm.missing.link} that no solution to $Pr$ exists.
\end{proof}

\section{Lambda-term syntax}
\label{sect.lambda.calculus}

We want to relate permissive nominal unification with higher-order pattern unification as promised in the Introduction.
In this section we recall the syntax and operational semantics of the $\lambda$-calculus.
It is convenient to match the variables of permissive nominal terms with those of the $\lambda$-calculus.
Therefore, when we define $\lambda$-terms' syntax in Definition~\ref{defn.terms.g}, we use the same atoms and unknowns as we used in Definition~\ref{defn.terms}; both behave like ordinary variables (with capture-avoiding substitution).
As in permissive nominal terms we unify on the unknowns, but we do not care about the permission sets so we will let $X, Y, Z, \ldots$ range over distinct unknowns (without superscripts).
\begin{defn}
\label{defn.terms.g}

Define \deffont{$\lambda$-terms} by:
\begin{frameqn}
g,h,\ldots ::= a \mid X \mid \tf f \mid \lam{a}g \mid g'g 
\end{frameqn}
Here $\tf f$ ranges over term-formers, $a$ ranges over atoms, and $X$ ranges over unknowns.

$g,h,k$ will range over $\lambda$-terms.
\end{defn}

\begin{defn}
\label{defn.perm.g}
Define a \deffont{permutation action} by:
\begin{frameqn}
\begin{gathered}
\pi {\act} a \equiv \pi(a)
\quad
\pi {\act} X \equiv X
\quad
\pi {\act} \tf f \equiv \tf f
\quad
\pi {\act} (\lam{a}g) \equiv \lam{\pi(a)}(\pi{\act} g)
\quad
\pi {\act} (g'g) \equiv (\pi{\act} g')(\pi{\act} g)
\end{gathered}
\end{frameqn}
\end{defn}

\begin{defn}
\label{defn.fa.g}
Define \deffont{free atoms} by:
\begin{frameqn}
\begin{gathered}
\f{fa}(a) {=} \{ a \} 
\quad 
\f{fa}(X) {=} \varnothing 
\quad 
\f{fa}(\tf f) {=} \varnothing 
\quad
\f{fa}(\lam{a}g) {=} \f{fa}(g){\setminus}\{a\}
\quad
\f{fa}(g'g) {=} \f{fa}(g'){\cup}\f{fa}(g)
\end{gathered}
\end{frameqn}
\end{defn}

\begin{defn}
\label{defn.alpha.lambda.terms}
Define \deffont{$\alpha$-equivalence} $\aeq$ inductively by the rules in Figure~\ref{fig.lambda.aeq}.
\begin{figure*}
$$
\begin{array}{c@{\quad}c@{\quad}c}
\begin{prooftree}
\justifies
a \aeq a
\using\rulefont{\lambda{ \aeq }a}
\end{prooftree}
&
\begin{prooftree}
g \aeq h 
\justifies
\lam{a}g \aeq \lam{a}h
\using\rulefont{\lambda{ \aeq }\lambda aa}
\end{prooftree}
&
\begin{prooftree}
(b\ a)\act g \aeq h \quad b\not\in\f{fa}(g)
\justifies
\lam{a}g \aeq \lam{b}h
\using\rulefont{\lambda{ \aeq }\lambda ab}
\end{prooftree}
\\[3ex] 
\begin{prooftree}
\justifies
\tf f \aeq \tf f
\using\rulefont{\lambda{ \aeq }\tf f}
\end{prooftree}
&
\begin{prooftree}
\justifies
X \aeq X
\using\rulefont{\lambda{ \aeq }X}
\end{prooftree}
&
\begin{prooftree}
g\aeq g'\quad h\aeq h'
\justifies
gh\aeq g'h'
\using\rulefont{\lambda{\aeq}p}
\end{prooftree}
\end{array}
$$
\caption{$\alpha$-equivalence on $\lambda$-terms.}
\label{fig.lambda.aeq}
\end{figure*}
\end{defn}

It is not hard to prove that Definition~\ref{defn.alpha.lambda.terms} does indeed specify the usual $\alpha$-equivalence relation on $\lambda$-terms.
Our definition is designed to match the definition of $\alpha$-equivalence on nominal terms (Definition~\ref{defn.aeq}).
This makes later results easier to prove (for example Theorem~\ref{thrm.injectivity}).

Lemma~\ref{lemm.ds.pi.g} to Proposition~\ref{prop.aeq.transitive.g} mirror similar results for permissive nominal terms.
The proofs of Lemmas~\ref{lemm.ds.pi.g} to~\ref{lemm.pi.ftma.g} are by induction on $g$.

\begin{lemm}
\label{lemm.ds.pi.g}
If $\pi|_{\f{fa}(g)}=\pi'|_{\f{fa}(g)}$ then $\pi\act g\aeq \pi'\act g$.
\end{lemm}

\begin{lemm}
\label{lemm.permutation.comp.g}
$\pi \act (\pi' \act g) \equiv (\pi \fcomp \pi') \act g$
\end{lemm}

\begin{lemm}
\label{lemm.pi.ftma.g}
$\f{fa}(\pi\act g)=\pi\act\f{fa}(g)$.
\end{lemm}

\begin{lemm}
\label{lemm.pi.aeq.g}
$g \aeq h$ implies $\pi\act g\aeq\pi\act h$.
\end{lemm}
\begin{proof}
By induction on the derivation of $g \aeq h$.
We consider one case:
\begin{squishitem}
\item
The case \rulefont{{\lambda}{\aeq}{\lambda}ab}.\quad
By inductive hypothesis $\pi \act ((b\ a) \act g) \aeq \pi \act h$.
By Lemma~\ref{lemm.permutation.comp.g}, $\pi \act ((b\ a) \act g) \equiv (\pi \fcomp (b\ a)) \act g$.
It is a fact that $\pi \fcomp (b\ a) = (\pi(b)\ \pi(a)) \fcomp \pi$, therefore $(\pi(b)\ \pi(a)) \act (\pi \act g) \aeq \pi \act h$, by Lemma~\ref{lemm.permutation.comp.g}.
By Lemma~\ref{lemm.pi.ftma.g}, $\pi(b) \not\in \f{fa}(\pi \act g)$.
Using \rulefont{{\lambda}{\aeq}{\lambda}ab}, we obtain $\lam{\pi(a)}\pi \act g \aeq \lam{\pi(b)}\pi \act h$.
The result follows from Definition~\ref{defn.perm.g}.
\end{squishitem}
\end{proof}

\begin{lemm}
\label{lemm.aeq.fa.g}
If $g \aeq h$ then $\f{fa}(g) = \f{fa}(h)$.
\end{lemm}
\begin{proof}
By induction on the derivation of $g \aeq h$.
We consider one case;
\begin{squishitem}
\item
The case \rulefont{{\lambda}{\aeq}{\lambda}ab}.\quad
Suppose $\lam{a}g \aeq \lam{b}h$ using \rulefont{{\lambda}{\aeq}{\lambda}ab}, where $b \not\in \f{fa}(g)$.
We aim to show $\f{fa}(\lam{a}g) = \f{fa}(\lam{b}h)$, or, $\f{fa}(g) \setminus \{ a \} = \f{fa}(h) \setminus \{ b \}$.
As $b \not\in \f{fa}(g)$ we have $\f{fa}(g) \setminus \{ a \} = (b\ a) \act \f{fa}(g) \setminus \{ b \}$.
Using Lemma~\ref{lemm.pi.ftma.g}, $(b\ a) \act \f{fa}(g) \setminus \{ b \} = \f{fa}((b\ a) \act g) \setminus \{ b \}$.
By inductive hypothesis $\f{fa}((b\ a) \act g) = \f{fa}(s)$, as required.
\end{squishitem}
\end{proof}

\begin{defn}
\label{defn.depth.g}
Define a notion of \deffont{size} on $\lambda$-terms by:
\begin{frameqn}
\begin{gathered}
\f{size}(a) = 0
\quad
\f{size}(X) = 0
\quad
\f{size}(\tf f) = 0
\quad
\f{size}(g'g) = \f{size}(g') {+} \f{size}(g)
\\
\f{size}(\lam{a}g) = 1 {+} \f{size}(g)
\end{gathered}
\end{frameqn}
\end{defn}

\begin{lemm}
\label{lemm.pi.depth.invariant.g}
$\f{size}(g) = \f{size}(\pi \act g)$
\end{lemm}
\begin{proof}
By induction on $g$.
\end{proof}

\begin{prop}
\label{prop.aeq.transitive.g}
$\aeq$ is transitive, reflexive, and symmetric.
\end{prop}
\begin{proof}
See Appendix~\ref{sect.omitted.proofs}.
We use Lemmas~\ref{lemm.ds.pi.g}, \ref{lemm.permutation.comp.g}, \ref{lemm.pi.ftma.g}, \ref{lemm.pi.aeq.g}, \ref{lemm.aeq.fa.g} and~\ref{lemm.pi.depth.invariant.g}.
\end{proof}

\begin{defn}
\label{defn.sigma}
Call a function $\sigma$ from unknowns to $\lambda$-terms a ($\lambda$-calculus) \deffont{substitution}.
$\sigma$ will range over substitutions (and later so will $\rho$; Definition~\ref{defn.renaming}).

We will write $[h/X]$ for the substitution which maps $X$ to $h$ and maps all other $Y$ to $Y$.
\end{defn}

\begin{defn}
\label{defn.subst.g}
Define the \deffont{capture-avoiding substitution} action $g\sigma$ on $\lambda$-terms by:
\begin{frameqn}
\begin{gathered}
a\sigma \equiv a
\qquad 
X\sigma \equiv \sigma(X)
\qquad
\tf{f}\sigma \equiv \tf{f}
\qquad
(g'g)\sigma \equiv (g'\sigma)(g\sigma)
\\
\begin{array}{r@{\ }l@{\ \ }l}
(\lam{a}g)\sigma \equiv& \lam{a}(g\sigma) & (a \not\in \bigcup \{\f{fa}(\sigma(X))\mid X\in\f{fV}(g)\}) 
\\[1ex]
(\lam{a}g)\sigma \equiv& \lam{b}(((b\ a) {\act} g)\sigma) &  (a \in \bigcup \{\f{fa}(\sigma(X))\mid X\in\f{fV}(g)\}) 
\end{array}
\end{gathered}
\end{frameqn}
In the final clause, `$b$ fresh' denotes a fixed but arbitrary choice of fresh $b$ (so $b \not\in \f{fa}(g)$ and $b\not\in \bigcup \{\f{fa}(\sigma(X)){\mid} X{\in}\f{fV}(g)\}$).
\end{defn}

\begin{defn}
\label{defn.sigma.circ.g}
Define \deffont{composition} $\sigma \fcomp \sigma'$ by: $(\sigma \fcomp \sigma')(X) \equiv (\sigma(X))\sigma'$.
This mirrors the definition for substitutions on permissive terms, given in Definition~\ref{defn.sigma.com}.
\end{defn} 

\begin{lemm}
\label{lemm.sigma.fcomp.g}
$g\sigma\sigma' \aeq g(\sigma \fcomp \sigma')$
\end{lemm}
\begin{proof}
By induction on $\f{size}(g)$.
\end{proof}

We also need a substitution action on atoms so that we can talk about $\alpha\beta$-convertibility --- the distinction between unknowns and atoms is rather artificial here, but in the context of relating to permissive nominal terms it is useful to maintain it:
\begin{defn}
\label{defn.gha}
Define a \deffont{capture-avoiding substitution} $g[h/a]$ by: 
\begin{frameqn}
\begin{gathered}
a[h/a] \equiv h
\qquad 
b[h/a] \equiv b
\qquad 
X[h/a] \equiv X
\qquad
\tf{f}[h/a] \equiv \tf{f}
\\
(g'g)[h/a] \equiv (g'[h/a])(g[h/a])
\qquad
(\lam{a}g)[h/a] \equiv \lam{a}g
\\
(\lam{a}g)[h/a] \equiv \lam{b}(((b\ a) \act g)[h/a]) \quad (b\text{ fresh})
\end{gathered}
\end{frameqn}
In the final clause, `$b$ fresh' denotes a fixed but arbitrary choice of $b$ such that $b\not\in \f{fa}(h)\cup\f{fa}(g)$.
\end{defn}

\begin{defn}
\label{defn.abeq}
Let \deffont{$\alpha\beta$-equivalence} $\abeq$ be the least transitive, reflexive, symmetric relation such that
$(\lam{a}g)h \abeq g[h/a]$  and closed under the rules of Definition~\ref{defn.alpha.lambda.terms}.
\end{defn}

We define unification problems as usual and write `$g\ueq h$' for an equality considered as part of a unification problem.
$\sigma$ solves a problem when $g\sigma\abeq h\sigma$ for every $g\ueq h$ in the problem, as usual.

We conclude with definions of \emph{pattern} and \emph{pattern substitution} \cite{miller:uniump,miller:logpll}. Recall that we work in an untyped $\lambda$-term syntax.

\begin{defn}
\label{defn.phi.patterns}
Let $\phi$ map each unknown $X$ to a natural number which we call its \deffont{arity}.
Define \deffont{$\phi$-patterns}, a subset of $\lambda$-terms, by: 
\begin{frameqn}
q,r,\ldots ::= a \mid Xa_1\ldots a_{\phi(X)} \mid \tf f q_1\ldots q_n \mid \lam{a}q
\end{frameqn}
Call $q$ a \deffont{pattern} when it is a $\phi$-pattern for some $\phi$.
$q,r,\ldots$ will range over patterns.

Call $\sigma$ a \deffont{$\phi$-pattern substitution} when every $\sigma(X)$ is a $\phi$-pattern. 
Call $\sigma$ a \deffont{pattern substitution} when $\sigma$ is a $\phi$-pattern substitution for some $\phi$.
\end{defn}
So $g$ is a pattern when there exists some $\phi$ such that every $X$ in $g$ occurs as $Xa_1\ldots a_{\phi(X)}$.
This is not quite a typing constraint, but it achieves part of what a typing system would achieve; that within $g$, for each $X$, $X$ is consistently applied to a list of atoms of the same length. 
Note that the translation of Definition~\ref{defn.nominal.trans.g} below, produces terms of this form.

\section{Translating nominal terms to lambda-term syntax}
\label{sect.nominal.terms.to.lambda}

\subsection{The translation, and its soundness} 

We define the translation from permissive nominal terms to $\lambda$-terms, and show that it is sound in the sense that $\alpha$-convertible permissive nominal terms map to $\alpha$-convertible $\lambda$-terms (Theorem~\ref{thrm.nomeq.lameq}).
The translation involves a vector of atoms $D$; we discuss where this comes from in Subsection~\ref{subsect.capt}.

\begin{defn}
\label{defn.vector}
Call a finite list of distinct atoms a \deffont{vector}.
$C,D$ range over vectors.
Write $[a_1,\ldots,a_n]$ for the vector containing $a_1$, \ldots, $a_n$ in that order.
\end{defn}

\begin{defn}
\label{defn.sets.notation}
Suppose $A\subseteq\mathbb A$. 
Write $C \cap A$ for the vector of atoms in $C$ that occur in $A$, in order; thus $[a_1,a_2,a_3]\cap \{a_1,a_3,a_5\} = [a_1,a_3]$.
Write $C \subseteq A$ when every atom in $C$ is in $A$.
Write $A \subseteq C$ when every atom in $A$ is in $C$.
\end{defn} 

\begin{defn}
\label{defn.nominal.trans.g}
Translate a nominal term $r$ to a $\lambda$-term $\inter{r}^D$ by:
\begin{frameqn}
\begin{gathered}
\inter{a}^D\equiv a
\quad
\inter{\pi\act X^S}^D \equiv X^S\pi(d_1)\ldots\pi(d_n) \quad ([d_1,\ldots,d_n] = D \cap S)
\\
\inter{[a]r}^D\equiv \lam{a}\inter{r}^D
\quad
\inter{\tf f(r_1,\ldots,r_n)}^D\equiv \tf f\inter{r_1}^D\ldots \inter{r_n}^D
\end{gathered}
\end{frameqn}
\end{defn}

\begin{lemm}
\label{lemm.commutation}
$\inter{\pi\act r}^D \equiv \pi\act \inter{r}^D$
\end{lemm}
\begin{proof}
By induction on $r$.
\end{proof}

Lemma~\ref{lemm.not.fa.fresh} is useful for the proof of Theorem~\ref{thrm.nomeq.lameq}:

\begin{lemm}
\label{lemm.not.fa.fresh}
$\f{fa}(\inter{r}^D)\subseteq\f{fa}(r)$.
\end{lemm}
\begin{proof}
By induction on $r$.
We consider only one case:
\begin{squishitem}
\item
The case $\pi \act X^S$.\quad 
We reason as follows:
\begin{calcenv}
\f{fa}(\inter{\pi \act X^S}^D) & = & \f{fa}(X^S \pi(d_1) \ldots \pi(d_n)) & \text{Definition~\ref{defn.nominal.trans.g}} \\
                               & = & \f{fa}(\pi(d_1)) \cup \ldots \cup \f{fa}(\pi(d_n)) & \text{Definition~\ref{defn.fa.g}} \\
                               & = & \pi \act (\f{fa}(d_1) \cup \ldots \cup \f{fa}(d_n)) & \text{Fact} \\
                               & \subseteq & \pi \act \f{fa}(X^S) & \text{Definition~\ref{defn.fa}} \\
                               & = & \f{fa}(\pi \act X^S) & \text{Lemma~\ref{lemm.pi.ftma}}
\end{calcenv}
The result follows.
\end{squishitem}
\end{proof}

\begin{thrm}[Soundness]
\label{thrm.nomeq.lameq}
If $r \aeq s$ then $\inter{r}^D \aeq \inter{s}^D$.
\end{thrm}
\begin{proof}
By induction on the size of $r$ (Definition~\ref{defn.depth}).
We reason by cases on the last rule in the derivation of $r\aeq s$:
\begin{squishitem}
\item
The cases \rulefont{{\aeq}a}, \rulefont{{\aeq}\tf f} and \rulefont{{\aeq}[]aa} are straightforward.
\item
The case \rulefont{{\aeq}X}.\quad
There are two cases:
\begin{squishitem} 
\item
The case $D \cap S = []$.\quad 
Then $\inter{\pi\act X^S}^D = \inter{\pi'\act X^S}^D = X^S$.
We use \rulefont{\lambda{\aeq}X}.
The result follows.
\item
The case $D \cap S = [d_1, \ldots, d_n]$ and $n \geq 1$.\quad
By assumption $\pi|_{\delta(X)} = \pi'|_{\delta(X)}$.
Then $\pi(d_i)=\pi'(d_i)$ for $1 \leq i \leq n$ and $\inter{\pi \act X^S}^D \equiv \inter{\pi' \act X^S}^D \equiv X^S\pi(d_1) \ldots \pi(d_n)$.
We use \rulefont{\lambda{\aeq}p}, \rulefont{\lambda{\aeq}X}, and \rulefont{\lambda{\aeq}a}.
The result follows.
\end{squishitem}
\item
The case \rulefont{{\aeq}[]ab}.\quad
By assumption $(b\ a) \act r \aeq s$ and $b \not\in \f{fa}(r)$.
We choose fresh $c$ so $c \not\in \f{fa}(r) \cup \f{fa}(s)$.
By Lemma~\ref{lemm.equality.permutation}, $(c\ a) \act r \aeq (c\ b) \act s$.
By inductive hypothesis $\inter{(c\ a)\act r}^D \aeq \inter{(c\ b)\act s}^D$.
Using \rulefont{\lambda{\aeq}\lambda aa}, $\lam{c}\inter{(c\ a)\act r}^D \aeq \lam{c}\inter{(c\ b)\act s}^D$.
By Lemma~\ref{lemm.commutation}, $\lam{c}((c\ a)\act\inter{r}^D) \aeq \lam{c}((c\ b)\act \inter{s}^D)$.
By Lemma~\ref{lemm.not.fa.fresh}, $c\not\in\f{fa}(\inter{r}^D)\cup\f{fa}(\inter{s}^D)$.
By \rulefont{\lambda{\aeq}\lambda ab}, $\lam{c}((c\ a)\act\inter{r}^D) \aeq \lam{a}\inter{r}^D$ and $\lam{c}((c\ b)\act\inter{s}^D) \aeq \lam{b}\inter{s}^D$. 
Using Proposition~\ref{prop.aeq.transitive.g}, $\lam{a}\inter{r}^D \aeq \lam{b}\inter{s}^D$.
By Definition~\ref{defn.nominal.trans.g}, $\inter{[a]r}^D\aeq\inter{[b]s}^D$, as required.
\end{squishitem}
\end{proof}

\subsection{Capturable atoms; injectivity and minimality}
\label{subsect.capt}

We investigate the converse to Theorem~\ref{thrm.nomeq.lameq}; if the translations of two terms are $\alpha$-convertible, then so are the two terms.
This is injectivity (Theorem~\ref{thrm.injectivity}). 
The translation $\inter{r}^D$ (Definition~\ref{defn.nominal.trans.g}) is parameterised by a vector $D$.  
Levy and Villaret introduced a translation \cite[Definition 2]{levy:nomufh} (for nominal, not permissive nominal terms); they used \emph{all} the atoms in $r$.  
This is a safe choice, but we will also show that the smaller set of \emph{capturable} atoms in $r$ (Definition~\ref{defn.capturable}) is consistent with injectivity --- and that the capturable atoms are the minimal such set (Theorem~\ref{thrm.at.least}).

\begin{defn}
\label{defn.capturable}
Define the \deffont{capturable atoms} of a term (with respect to a set of atoms) $\f{capt}_A(r)$ inductively by:
\begin{frameqn}
\begin{array}{r@{\ }l@{\qquad}r@{\ }l}
\f{capt}_A(a) =& \varnothing
&
\f{capt}_A(\pi\act X^S) =& (\f{nontriv}(\pi) \cup A) \cap S
\\
\f{capt}_A([a]r)=&\f{capt}_{A \cup \{a\}}(r)
&
\f{capt}_A(\tf f(r_1,\ldots,r_n))=& \bigcup_{1 \leq i \leq n}\f{capt}_A(r_i) \qquad\qquad
\end{array}
\end{frameqn}
Write $\f{capt}_\varnothing(r)$ as $\f{capt}(r)$.
\end{defn}

For instance, if $S = (\mathbb A^<  \cup \{a\}) \setminus \{b\}$, then $\f{capt}([a][b]X^S) = \{a\}$ and $\f{capt}((b\ a) \act X^S) = \{a\}$.

\begin{lemm}
\label{lemm.capturable.fresh.atom}
If $a \not\in \f{fa}(r)$ then $\f{capt}_A(r) = \f{capt}_{A \cup \{a\}}(r)$.
\end{lemm}
\begin{proof}
By induction on $r$.
\begin{squishitem}
\item
The cases $b$, $\tf{f}(r_1, \ldots, r_n)$ and $[a]r$ are straightforward.
\item
The case $[b]r$.\quad
If $a \not\in \f{fa}([b]r)$ then $a \not\in \f{fa}(r)$ by Definition~\ref{defn.fa}.
We then have:
\begin{calcenv}
\f{capt}_A([b]r) & = & \f{capt}_{A \cup \{ b \}}(r) & \text{Definition~\ref{defn.capturable}} \\
                 & = & \f{capt}_{A \cup \{ b \} \cup \{ a \}}(r) & \text{Inductive hypothesis} \\
                 & = & \f{capt}_{A \cup \{ a \}}([b]r) & \text{Definition~\ref{defn.capturable}}
\end{calcenv}
The result follows.
\item
The case $\pi\act X^S$.\quad
If $a \not\in \f{fa}(\pi \act X^S)$ then $a \not\in \pi \act S$.
By Definition~\ref{defn.capturable}, $\f{capt}_A(\pi \act X^S) = (\f{nontriv}(\pi) \cup A) \cap S$.
Further, $\f{capt}_{A \cup \{a\}}(\pi \act X^S) = (\f{nontriv}(\pi) \cup A \cup \{a \}) \cap S$.
If $\pi(a) = a$ then $a \not\in S$. 
If $\pi(a) \not= a$ then $a \in \f{nontriv}(\pi)$, as required.
\end{squishitem}
\end{proof}

\begin{lemm}
\label{lemm.capturable.dom.pi}
If $\f{nontriv}(\pi) \subseteq A$ then $\f{capt}_A(\pi \act r) = \f{capt}_A(r)$.
\end{lemm}
\begin{proof}
See Appendix~\ref{sect.omitted.proofs}.
We use Lemma~\ref{lemm.permutation.comp}.
\end{proof}

\begin{corr}
\label{corr.capturable.basic.alpha}
If $b \not\in \f{fa}(r)$ then $\f{capt}_A([a]r) = \f{capt}_A([b](b\ a) \act r)$.
\end{corr}
\begin{proof}
See Appendix~\ref{sect.omitted.proofs}.
We use Lemmas~\ref{lemm.capturable.fresh.atom}, \ref{lemm.capturable.dom.pi}, and~\ref{lemm.pi.ftma}.
\end{proof}

Capturable atoms is a canonical notion, in the sense that it is preserved by $\alpha$-equivalence:
\begin{lemm}
\label{lemm.capturable.alpha}
If $r \aeq s$ then $\f{capt}_A(r)=\f{capt}_A(s)$.
\end{lemm}
\begin{proof}
By induction on the derivation of $r \aeq s$.
We consider one case:
\begin{squishitem}
\item
The case of \rulefont{{\aeq}[b]}.\quad
Suppose $b\not\in\f{fa}(r)$,\ $(b\ a)\act r\aeq s$ and $s\equiv [b](b\ a)\act r$.
The result follows by Corollary~\ref{corr.capturable.basic.alpha}.
\end{squishitem}
\end{proof}

Provided that $D$ is `large enough', $\alpha$-equivalence of translated terms implies $\alpha$-equivalence of the original terms: 
\begin{thrm}[Injectivity] 
\label{thrm.injectivity}
Let $D$ be a vector.
Let $r$ and $s$ be nominal terms and let $A,B\subseteq \mathbb A$ be finite.
Suppose 
$\f{capt}_A(r)\cup\f{capt}_B(s) \subseteq D$.
Then
\begin{displaymath}
\inter{r}^D \aeq \inter{s}^D
\quad\text{implies}\quad 
r \aeq s.
\end{displaymath}

As a corollary, if $\f{capt}(r) \cup \f{capt}(s) \subseteq D$ and $\inter{r}^D \aeq \inter{s}^D$ then $r \aeq s$ and $\f{capt}_A(r)=\f{capt}_A(s)$ for all $A$.
\end{thrm}
\begin{proof}
For the first part, we work by induction on the size of $r$ (Definition~\ref{defn.depth}), reasoning by cases on the last rule in the derivation of $\inter{r}^D\aeq \inter{s}^D$:
\begin{squishitem}
\item
The cases \rulefont{\lambda{\aeq} a} and \rulefont{\lambda{\aeq} \lambda aa}.\quad
Easy.
\item
The case \rulefont{\lambda{\aeq} \lambda ab}.\quad
Suppose $(b\ a) \act \inter{r}^D \aeq \inter{s}^D$, $b \not\in \f{fa}(\inter{r}^D)$ and $\f{capt}_A([a]r) \cup \f{capt}_B([b]s) \subseteq D$.

We choose fresh $c$ (so $c\not\in\f{fa}(r)\cup\f{fa}(s)$ and $c\not\in\f{fa}(\inter{r}^D)\cup\f{fa}(\inter{s}^D)$).
By Lemma~\ref{lemm.ds.pi.g}, $(c\ a) \act \inter{r}^D \aeq (c\ b) \act \inter{s}^D$.
By Lemma~\ref{lemm.commutation}:
\begin{displaymath}
\inter{(c\ a)\act r}^D\aeq \inter{(c\ b)\act s}^D .
\end{displaymath}
By Corollary~\ref{corr.capturable.basic.alpha} and Definition~\ref{defn.capturable}:
\begin{displaymath}
\f{capt}_{A\cup\{c\}}((c\ a) \act r) \cup \f{capt}_{B\cup\{c\}}((c\ b) \act s) \subseteq D.
\end{displaymath}
By inductive hypothesis $(c\ a) \act r \aeq (c\ b) \act s$, and using \rulefont{{\aeq}[]ab}, and Proposition~\ref{prop.aeq.transitive.g}, $[a]r \aeq [b]s$ as required.
\item
The case \rulefont{\lambda{\aeq} app}.
\quad
From Definition~\ref{defn.nominal.trans.g}, there are two possibilities:
\begin{squishitem}
\item
The case $\tf f(r_1,\ldots, r_n)$ and $\tf f(s_1,\ldots,s_n)$ and $\tf f\inter{r_1}^D\ldots\inter{r_n}^D\aeq \tf f\inter{s_1}^D\ldots\inter{s_n}^D$.

Then $\inter{r_i}^D\aeq \inter{s_i}^D$ for $1\leq i\leq n$.
By inductive hypothesis $r_i \aeq s_i$ for $1 \leq i \leq n$.
The result follows from \rulefont{{\aeq}\tf f}.
\item
The case $\pi\act X^S$ and $\pi'\act X^S$ and $X^S \pi(d_1) \ldots \pi(d_n)\aeq X^S \pi'(d_1) \ldots \pi'(d_n)$ where $[d_1, \ldots, d_n] = D \cap S$.

Then $\pi(d_i) \aeq \pi'(d_i)$ for $1 \leq i \leq n$, and so $\pi|_{D \cap S} = \pi'|_{D \cap S}$.
By assumption, $\f{capt}(\pi \act X^S) \subseteq D$, and by definition, $\pi|_{S} = \pi'|_{S}$.
We conclude with \rulefont{{\aeq}X}.
\end{squishitem}
\item
The case \rulefont{\lambda{\aeq}X}. 
\quad
From the form of the translation it must be that $r = \pi\act X^S$ and $s= \pi' \act X^S$ and $\f{nontriv}(\pi) \cap S = \varnothing = \f{nontriv}(\pi') \cap S$. 
The result follows from \rulefont{{\aeq}X}.
\end{squishitem}
The corollary follows from the first part, and by Lemma~\ref{lemm.capturable.alpha}.
\end{proof}

\begin{lemm}
\label{lemm.somewhere.X}
$a \in \f{capt}_A(r)$ implies $X^S \in \f{fV}(r)$ exists such that $a \in S$.
\end{lemm}
\begin{proof}
See Appendix~\ref{sect.omitted.proofs}.
\end{proof}
 
In Theorem~\ref{thrm.injectivity} we used a notion of a `large enough' vector $D$, based on capturable atoms.
Theorem~\ref{thrm.at.least} shows how this bound is precise:
\begin{thrm}[Minimality]
\label{thrm.at.least}
If $\f{capt}(r) \not\subseteq D$ then there exists some $s$ such that $r \not\aeq s$ and $\inter{r}^D \aeq \inter{s}^D$.
\end{thrm}
\begin{proof}
Suppose $a \in \f{capt}(r)$ and $a \not\in D$.
By Lemma~\ref{lemm.somewhere.X} $X^S \in \f{fV}(r)$ exists such that $a \in S$.
We choose $c \in S\setminus D$ and take $s \equiv r[X^S \ssm (c\ a) \act X^S]$. 
It is a fact that $X^S \not\aeq (c\ a) \act X^S$ whilst $\inter{X^S}^D \aeq \inter{(c\ a) \act X^S}^D$.
An easy calculation shows $r \not\aeq r[X^S \ssm (c\ a) \act X^S]$ and $\inter{r}^D \aeq \inter{r[X^S \ssm (c\ a) \act X^S]}^D$. 
\end{proof}

\section{Translating substitutions; relating solutions of nominal and pattern unification problems}
\label{sect.translating.solutions}

\subsection{Translating substitutions}
\label{subsect.translating.substitutions}

We extend the translation to substitutions.
Our main result is Theorem~\ref{thrm.sub.composition}: we can read this as a compositionality result for permissive nominal substitutions acting on terms with respect to the translation.
Given the compositionality result it is easy to prove that the translation also preserves the natural instantiation ordering of solutions to unification problems (Corollary~\ref{corr.instant}).

Translating substitutions introduces a problem: $\theta$ may solve $Pr$ but in substituting it may introduce new capturable atoms (consider $\theta = [X^S \ssm [c]Z^S]$ solving $\{X^S \ueq X^S \}$, where $c \in S$). 
This motivates introducing a second vector, to account for the capturable atoms `after' the substitution.
Accordingly, we will introduce another vector $E$ that contains at least the capturable atoms of $\theta$. 

\begin{defn}
\label{defn.thetaC}
Define $\inter{\theta}_D^E$ by:
\begin{displaymath}
\text{$\inter{\theta}_D^E(X^S) = \lam{d_1}\ldots\lam{d_n}\inter{\theta(X^S)}^E$ where $[d_1, \ldots, d_n] = D \cap S$.}
\end{displaymath}
\end{defn}

Lemma~\ref{lemm.abeq.pi} is useful in the proof of Theorem~\ref{thrm.sub.composition}:
\begin{lemm}
\label{lemm.abeq.pi}
If $\f{nontriv}(\pi)\cap\f{fa}(g) \subseteq \{ d_1, \ldots ,d_n \}$ then
$$
(\lam{d_1}\ldots\lam{d_n}g)\pi(d_1)\ldots\pi(d_n) \abeq \pi \act g.
$$
\end{lemm}
\begin{proof}
By induction on $g$.
We consider one case: 
\begin{squishitem}
\item
The case $\lam{a}g$.\quad
Choose $a'$ fresh (so $a'$ does not appear in $\{d_1,\ldots,d_n\}$,\ $\f{nontriv}(\pi)$,\ or $g$).
\begin{calcenv}
h & \abeq & (\lam{a}g)[\pi(d_1) / d_1] \ldots [\pi(d_n) / d_n] & \text{Definition~\ref{defn.abeq}} \\
  & \aeq & (\lam{a'}(a'\ a)\act g)[\pi(d_1) / d_1] \ldots [\pi(d_n) / d_n] & \text{Definition~\ref{defn.alpha.lambda.terms}} \\
  & \equiv & \lam{a'}((a'\ a)\act g)[\pi(d_1) / d_1] \ldots [\pi(d_n) / d_n]) & \text{Definition~\ref{defn.gha}, $a'$ fresh} \\
  & \abeq & \lam{a'}(\pi \act ((a'\ a)\act g)) & \text{Inductive hypothesis} \\
    & \equiv & \pi \act \lam{a'}((a'\ a)\act g) & \text{Definition~\ref{defn.perm.g}} \\
  & \aeq & \pi \act \lam{a}g & \text{Definition~\ref{defn.alpha.lambda.terms}}
\end{calcenv}
\end{squishitem}
\end{proof}

Theorem~\ref{thrm.sub.composition} is a compositionality result for permissive nominal substitutions acting on terms with respect to the translation:
\begin{thrm}
\label{thrm.sub.composition}
If $\f{capt}(r) \subseteq D$ then $\inter{r \theta}^E\abeq \inter{r}^D\inter{\theta}_D^E$.
\end{thrm}
\begin{proof}
By induction on $r$.
\begin{squishitem}
\item 
The cases $a$ and $\tf{f}(r_1, \ldots, r_n)$ are routine.
\item 
The case $\pi \act X^S$.\quad 
Let $d_1, \ldots, d_n$ be $D \cap S$
and $\inter{\theta}_D^E(X^S) = \lam{d_1}\ldots\lam{d_n}\inter{\theta(X^S)}^E$ by Definition~\ref{defn.thetaC}.
Then:
\begin{calcenv}
\inter{(\pi\act X^S)\theta}^E & \equiv & \inter{\pi \act \theta(X^S)}^E & \text{Definition~\ref{defn.subst.action}} \\
                              & \equiv & \pi \act \inter{\theta(X^S)}^E & \text{Lemma~\ref{lemm.commutation}} \\
                              & \abeq & (\lam{d_1}\ldots\lam{d_n}\inter{\theta(X^S)}^E)\pi(d_1)\ldots\pi(d_n) & \text{Lemma~\ref{lemm.abeq.pi}} \\
                              & \equiv & (X^S\pi(d_1) \ldots \pi(d_n)) \inter{\theta}_D^E & \text{Definition~\ref{defn.thetaC}} \\
                              & \equiv & \inter{\pi\act X^S}^D \inter{\theta}_D^E & \text{Definition~\ref{defn.nominal.trans.g}}
\end{calcenv}
The use of Lemma~\ref{lemm.abeq.pi} above is valid, because: By assumption $\f{capt}(\pi \act X^S) \subseteq D$. 
By Definition~\ref{defn.capturable} $\f{capt}(\pi\act X^S)=\f{nontriv}(\pi)\cap S$, so $\f{nontriv}(\pi) \cap S \subseteq D$.
By assumption in Definition~\ref{defn.subst}, $\f{fa}(\theta(X^S))\subseteq S$.
It follows from Definition~\ref{defn.nominal.trans.g} that $\f{fa}(\inter{\theta(X^S)}^E)\subseteq S$, and so that $\f{nontriv}(\pi)\cap\f{fa}(\inter{\theta(X^S)}^E)\subseteq D$. 

The result follows.
\item
The case $[a]r$.\quad
Choose $b$ fresh, so $b \not\in \f{fa}(\inter{\theta(X^S)}_D^E)$ for every $X^S \in \f{fV}(r)$ and $b \not\in \f{fa}(r)$.
Then:
\begin{calcenv}
\inter{([a]r)\theta}^E & \aeq & \inter{([b]((b\ a) \act r)) \theta}^E & \text{Definition~\ref{defn.aeq}, Theorem~\ref{thrm.nomeq.lameq}, Lemma~\ref{lemm.aeq.subst}} \\
                       & \equiv & \lam{b}(\inter{((b\ a) \act r)\theta)}^E & \text{Definitions~\ref{defn.subst.action} and~\ref{defn.nominal.trans.g}} \\
                       & \abeq & \lam{b}((\inter{(b\ a) \act r}^D)\inter{\theta}_D^E) & \text{Inductive hypothesis} \\
                       & \equiv & (\lam{b}\inter{(b\ a) \act r}^D)\inter{\theta}_D^E & \text{Definition~\ref{defn.gha}, $b$ fresh} \\
                       & \equiv & \inter{[b]((b\ a) \act r)}^D \inter{\theta}_D^E & \text{Definition~\ref{defn.nominal.trans.g}} \\
                       & \aeq & \inter{[a]r}^D \inter{\theta}_D^E & \text{Definition~\ref{defn.aeq}, Theorem~\ref{thrm.nomeq.lameq}, Lemma~\ref{lemm.aeq.subst}}
\end{calcenv}
The result follows.
\end{squishitem}
\end{proof}

Recall the instantiation ordering $\theta_1\leq\theta_2$ from Definition~\ref{defn.instantiation.ordering}.
Similarly:
\begin{defn}
\label{defn.instantiation.ordering.lambda}
Write $\sigma_1\leq\sigma_2$ when there exists some $\sigma'$ such that $X\sigma_2\abeq X(\sigma_1\circ\sigma')$, for any $X$.
Call $\leq$ the \deffont{instantiation ordering}. 
\end{defn}

We can leverage Theorem~\ref{thrm.sub.composition} to prove a corollary, describing a sense in which the instantiation ordering $\theta_1\leq\theta_2$ of Definition~\ref{defn.instantiation.ordering} translates to the instantiation ordering of Definition~\ref{defn.instantiation.ordering.lambda}:
\begin{corr}
\label{corr.instant}
Suppose $\bigcup_{X^S}\f{capt}(\theta_2(X^S)) \subseteq E$.

If $\theta_1\leq\theta_2$ then $\model{\theta_1}_D^E\leq\model{\theta_2}_D^E$.
\end{corr}
\begin{proof}

Suppose $\theta_1\leq\theta_2$.
By definition (Definition~\ref{defn.instantiation.ordering}) there exists some $\theta'$ such that
$X^S\theta_1 \aeq X^S(\theta_2\circ\theta')$ always.
We reason as follows, for any unknown $X^S$:
$$
\begin{array}{l@{\ }c@{\ }l@{\qquad}l}
\inter{X^S}^D \model{\theta_2}_D^E & \abeq & \inter{X^S\theta_2}^E & \text{Theorem~\ref{thrm.sub.composition}} \\
                                 & \aeq & \inter{X^S(\theta_1\circ\theta')}^E & \text{Theorem~\ref{thrm.nomeq.lameq}} \\
                                 & \equiv & \inter{(X^S\theta_1)\theta')}^E & \text{Lemma~\ref{lemm.sigma.sigma'.circ}} \\
                                 & \abeq & \inter{X^S\theta_1}^E\inter{\theta'}_E^E & \text{Theorem~\ref{thrm.sub.composition},}\ \f{capt}(\theta_1(X^S)) \subseteq E \\
                                 & \abeq & (\inter{X^S}^D\inter{\theta_1}_D^E)\inter{\theta'}_E^E & \text{Theorem~\ref{thrm.sub.composition}} \\
                                 & \equiv & \inter{X^S}^D(\inter{\theta_1}_D^E\fcomp\inter{\theta'}_E^E) & \text{Lemma~\ref{lemm.sigma.fcomp.g}}
\end{array}
$$
The result follows.
\end{proof}

In Corollary~\ref{corr.instant}, the precondition $\bigcup_{X^S}\f{capt}(\theta_2(X^S)) \subseteq E$ is necessary to prevent $\theta_2$ from introducing infinitely many capturable atoms.
The `complexity' of $\theta_1$ is unconstrained.
In practice it is likely that we will care about a particular finite set of unknowns $\mathcal V$ (for example, $\f{fV}(Pr)$ for some $Pr$), and the precondition can be correspondingly refined to consider just $X^S\in\mathcal V$.

\subsection{Translating permissive nominal unification to pattern unification; soundness, weak completeness}

The main result of this subsection is Theorem~\ref{thrm.sound}.

It says that if $D$ and $E$ are `large enough', then $\theta$ solves $Pr$ if and only if $\inter{\theta}_D^E$ solves $\inter{Pr}^D$.
We call this `soundness and weak completeness', to distinguish from a stronger completeness result we prove in Subsection~\ref{subsect.strong.completeness}.

\begin{defn}
\label{defn.unification.problem}
An \deffont{equation} is a pair $r \ueq s$.  
A \deffont{unification problem} $Pr$ is a finite set of equations.
A \deffont{solution} to $Pr$ is a 
$\theta$ such that
$r \theta {\aeq} s \theta$ for all $r{\ueq} s\in Pr$.
\end{defn}

\begin{defn}
\label{defn.pr.rc}
If $D=[d_1, \ldots, d_n]$ and $Pr=\{r_1 \ueq s_1,\ldots \}$ then define $\inter{Pr}^D$ by:
\begin{frameqn}
\inter{Pr}^D = \bigl\{ \lam{d_1}\ldots\lam{d_n}\inter{r}^D \ueq \lam{d_1}\ldots\lam{d_n}\inter{s}^D \mid r \ueq s \in Pr \bigr\}
\end{frameqn}
\end{defn}
For example, if $Pr=\{X^S \ueq \tf f (Y^S,a,Z^S)\}$ where $S = \mathbb A^<  \cup \{a,b\}$, then $\inter{Pr}^{[a]}=\{\lam{a}(X^S a) \ueq \lam{a}(\tf{f}\ (Y^S a)\ a\ (Z^S a))\}$.


Definition~\ref{defn.uncapt} is a technical definition, and the results following it are technical lemmas required for Lemma~\ref{lemm.capturable.pi.r}, which is a key result for Lemma~\ref{lemm.capturable.theta.r}, which is itself needed for Theorem~\ref{thrm.sound}.
\begin{defn}
\label{defn.uncapt}
Define $\f{uncapt}(r)$ by:
\begin{frameqn}
\begin{array}{r@{\ }l}
\f{uncapt}(a)=&\varnothing
\\
\f{uncapt}(\pi\act X^S)=&S\setminus\f{nontriv}(\pi)
\\
\f{uncapt}([a]r)=&\f{uncapt}(r)\setminus\{a\}
\\
\f{uncapt}(\tf f(r_1,\ldots,r_n))=&\f{uncapt}(r_1)\cup\cdots\cup\f{uncapt}(r_n)
\end{array}
\end{frameqn}
\end{defn}
$\f{uncapt}(r)$ is quite interesting; technically, it came about as the `right definition' to make Lemmas~\ref{lemm.uncapt.to.capt},\ \ref{lemm.capturable.pi.r}, and~\ref{lemm.uncapt.fa} work.\footnote{Thanks to an anonymous referee for spotting the error in a previous proof.} 
Intuitively, it may be useful to think of $\f{uncapt}(r)$ as collecting those atoms that are not necessarily in $\f{capt}(r)$, but which could feature in $\f{capt}([a]r)$ for some $a$. 

For example, if $a\in S$ then $a\in\f{uncapt}(X^S)$ but 
$$
a\not\in\f{uncapt}(a), \quad a\not\in\f{uncapt}([b][a]X^S),\quad\text{and}\quad a\not\in\f{uncapt}((b\ a)\act X^S). 
$$

\begin{lemm}
\label{lemm.capturable.A.B}
If $A \subseteq B$ then $\f{capt}_A(r) \subseteq \f{capt}_B(r)$.

As a corollary, $\f{capt}_A(r) \subseteq \f{capt}_A([a]r)$.
\end{lemm}
\begin{proof}
By induction on $r$.
As $\f{capt}_A([a]r) = \f{capt}_{A \cup \{ a \}}(r)$, the corollary follows.
\end{proof}

\begin{lemm}
\label{lemm.uncapt.to.capt}
Suppose $a\in A$.
Then if $a\in\f{uncapt}(r)$ then $a\in\f{capt}_A(r)$.

As a corollary, $a\in\f{uncapt}(r)$ implies $a\in\f{capt}([a]r)$.
\end{lemm}
\begin{proof}
By induction on $r$.
\begin{squishitem}
\item
The cases $a$ and $\tf{f}(r_1, \ldots, r_n)$ are straightforward.
\item
The case $[a]r$.\quad
$a\in\f{uncapt}([a]r)$ is impossible.
\item
The case $[b]r$.\quad 
Suppose $a\in\f{uncapt}([b]r)$. 
By Definition~\ref{defn.uncapt} $a\in\f{uncapt}(r)$.
By inductive hypothesis $a\in\f{capt}_A(r)$.
By Lemma~\ref{lemm.capturable.A.B} $a\in\f{capt}_A([b]r)$ as required.
\item
The case $\pi\act X^S$.\quad
Suppose $a\in\f{uncapt}(\pi\act X^S)$, so $a\in S\setminus\f{nontriv}(\pi)$.
By assumption $a\in A$ so $a\in (\f{nontriv}(\pi)\cup A)\cap S=\f{capt}_A(\pi\act X^S)$.
\end{squishitem}
\end{proof}

\begin{lemm}
\label{lemm.capturable.pi.r}
$\f{capt}_A(\pi \act r) \subseteq ((\f{nontriv}(\pi)\cup A) \cap \f{uncapt}(r)) \cup \f{capt}(r)$.
\end{lemm}
\begin{proof}
See Appendix~\ref{sect.omitted.proofs}.
We use Lemmas~\ref{lemm.capturable.A.B} and~\ref{lemm.uncapt.to.capt}.
\end{proof}

\begin{lemm}
\label{lemm.uncapt.fa}
$\f{uncapt}(r)\subseteq\f{fa}(r)$.
\end{lemm}
\begin{proof}
By a routine induction on $r$ using Definitions~\ref{defn.fa} and~\ref{defn.uncapt}. 
\end{proof}

\begin{lemm}
\label{lemm.capturable.theta.r}

$\f{capt}_A(r\theta) \subseteq \f{capt}_A(r) \cup \bigcup_{X^S\in\f{fV}(r)}\f{capt}(\theta(X^S))$.

\end{lemm}
\begin{proof}
By induction on $r$:
\begin{squishitem}
\item
The cases $a$ and $\tf{f}(r_1, \ldots, r_n)$.\quad
Straightforward.
\item
The case $[a]r$.\quad
We reason as follows:
\begin{calcenv}
\hspace{-2em}
\f{capt}_A([a](r\theta)) & = & \f{capt}_{A\cup\{a\}}(r\theta) & \text{Definition~\ref{defn.capturable}} \\
                            & \subseteq & \f{capt}_{A \cup \{ a \}}(r) \cup \bigcup_{X^S\in\f{fV}(r)}\f{capt}(\theta(X^S)) & \text{Ind. hyp.} \\
                            & = & \f{capt}_A([a]r) \cup \bigcup_{X^S\in\f{fV}(r)}\f{capt}(\theta(X^S)) & \text{Definition~\ref{defn.capturable}}
\end{calcenv}
The result follows.
\item
The case $\pi \act X^S$.\quad
As $(\pi \act X^S)\theta \equiv \pi \act \theta(X^S)$, we reason as follows:
$$
\hspace{-2em}
\begin{array}{r@{\ }ll}
\f{capt}_A(\pi\act\theta(X^S))  \subseteq & ((\f{nontriv}(\pi)\cup A) \cap \f{uncapt}(\theta(X^S)))
\\
 & \qquad  \cup \f{capt}(\theta(X^S)) & \text{Lemma~\ref{lemm.capturable.pi.r}} 
\\
                        \subseteq & ((\f{nontriv}(\pi) \cup A) \cap S) \cup \f{capt}(\theta(X^S))  & \text{Lemma~\ref{lemm.uncapt.fa}}  \\
                        = & \f{capt}_A(\pi\act X^S) \cup \f{capt}(\theta(X^S)) & \text{Definition~\ref{defn.capturable}}
\end{array}
$$
The result follows.
\qedhere
\end{squishitem}
\end{proof}

\begin{rmrk}
$\f{capt}(r\theta) \subseteq \bigcup_{\f{fV}(r)}\f{capt}(\theta(X^S))$ is not true in general.
For example if $a \in S$ and $b \in S$ then $\f{capt}([a]X^S)= \{ a \}$ and $\f{capt}([X^S \ssm [b]X^S]) = \{ b \}$, and $\f{capt}(([a]X^S))\theta = \{ a, b\} \not\subseteq \{ b \}$.
\end{rmrk}

\begin{lemm}
\label{lemm.solves.sound}
Suppose $\f{capt}(Pr) \subseteq D$ and $\f{capt}(Pr\theta) \subseteq E$. 
Then $\theta$ solves $Pr$ if and only if $\inter{\theta}_D^E$ solves $\inter{Pr}^D$.
\end{lemm}
\begin{proof}

We reason as follows, where $D=[d_1,\ldots,d_n]$:
  {
  $$
  \begin{array}{r@{}l@{\qquad}l}
  r \theta \aeq s \theta
  \liff\ &\model{r \theta}^E \aeq \model{s \theta}^E
  &\text{Theorems~\ref{thrm.nomeq.lameq} and~\ref{thrm.injectivity}}
 \\
 \liff\ &\model{r}^D \model{\theta}_D^E
 \abeq \model{s}^D \model{\theta}_D^E
  &\text{Theorem~\ref{thrm.sub.composition}}
 \\
 \liff\lam{d_1} \ldots \lam{d_n}&\model{r}^D \model{\theta}_D^E
 \abeq
 \lam{d_1} \ldots \lam{d_n}\model{s}^D \model{\theta}_D^E
&\text{fact of $\lambda$-terms}
 \\
 \liff
 (\lam{d_1} \ldots \lam{d_n}&\model{r}^D) \model{\theta}_D^E \abeq
 (\lam{d_1} \ldots \lam{d_n}\model{s}^D) \model{\theta}_D^E
 &\text{no atom of $D$ free in $\model{\theta}_D^E$}
 \end{array}
 $$
 }
\end{proof}

If $D$ and $E$ are `large enough', then $\theta$ solves $Pr$ if and only if the translation $\inter{\theta}_D^E$ solves the translation $\inter{Pr}^D$:
\begin{thrm}[Soundness and weak completeness]
\label{thrm.sound}
Suppose 
$$\f{capt}(Pr) \subseteq D
\qquad 
\f{capt}(\theta(X^S)) \subseteq E\text{ for all }X^S\in\f{fV}(Pr),\quad\text{and}\quad D \subseteq E.
$$
Then $\theta$ solves $Pr$ if and only if $\inter{\theta}_D^E$ solves $\inter{Pr}^D$.
\end{thrm}
\begin{proof}
Suppose $\f{capt}(Pr) \subseteq D$,\ \ $\f{capt}(\theta(X^S)) \subseteq E$ for all $X^S\in\f{fV}(Pr)$,\ \ and $D \subseteq E$.
Then:
$$
\begin{array}{r@{\ }l@{\quad}l}
\f{capt}(Pr\theta)\subseteq& \f{capt}(Pr)\cup\bigcup_{X^S\in\f{fV}(Pr)}\f{capt}(\theta(X^S))
&\text{Lemma~\ref{lemm.capturable.theta.r}}
\\
\subseteq& D\cup\bigcup_{X^S\in\f{fV}(Pr)}\f{capt}(\theta(X^S))
&\f{capt}(Pr)\subseteq D
\\
\subseteq& D\cup E
&\f{capt}(\theta(X^S)) \subseteq E \\
\subseteq & E
&D\subseteq E
\end{array}
$$
By Lemma~\ref{lemm.solves.sound}, $\theta$ solves $Pr$ if and only if $\inter{\theta}_D^E$ solves $\inter{Pr}^D$. 
\end{proof}

For example, $Pr = \{ X^S \ueq \tf{f}(Y^S, a, Z^S) \}$ where $S = \mathbb A^<  \cup \{a,b\}$ translates to $\inter{Pr}^{[a]} = \{ \lam{a}(X^S\,a) = \lam{a}(\tf{f}\,(Y^S\,a)\,a\,(Z^S\,a))\}$.

The solution $[X^S \ssm \tf f(W^S, a, b), Y^S \ssm W^S, Z^S \ssm b]$ with $S = \mathbb A^<  \cup \{a,b\}$ translates to $\inter{\theta}_{[a]}^{[a,b]} = [X^S \ssm \lam{a} (\tf{f}\,(W^S\,a\,b)\,a\,b), Y^S \ssm \lam{a}(W^S\, a\, b), Z^S \ssm \lam{a}b]$.

\subsection{Strong Completeness}
\label{subsect.strong.completeness}

The main result of this subsection is Theorem~\ref{thrm.strong}.
This strengthens the completeness result of Theorem~\ref{thrm.sound}, in a certain sense, by
expressing that a class of $\sigma$ solving $\inter{Pr}^D$ all originate from $\theta$ solving $Pr$, in a suitable formal sense. 

\begin{defn}
\label{defn.renaming}
Call a bijection on unknowns a \deffont{renaming}.
$\rho$ will range over renamings.
Each $X$ is also a $\lambda$-term (Definition~\ref{defn.terms.g}), so each $\rho$ is also a substitution (Definition~\ref{defn.sigma}).
\end{defn}

\begin{lemm}
\label{lemm.fa.subst.g}
$\f{fa}(g) = \f{fa}(g\rho)$
\end{lemm}
\begin{proof}
By induction on $g$.
\end{proof}

\begin{lemm}
\label{lemm.aeq.rho.g}
$g \aeq h$ if and only if $g\rho \aeq h\rho$.
\end{lemm}
\begin{proof}
The left to right implication is by induction on the derivation of $g \aeq h$; right to left is by induction on the derivation of $g\rho \aeq h\rho$.
We consider one case:
\begin{squishitem}
\item
The case \rulefont{{\lambda}{\aeq}{\lambda}ab}.\quad
For the left to right implication, suppose $(b\ a)\act g\aeq h$ and $b\not\in\f{fa}(g)$.
By inductive hypothesis $((b\ a) \act g)\rho \aeq h\rho$.
By Lemma~\ref{lemm.fa.subst.g}, $b \not\in \f{fa}(g\rho)$.
It is a fact that $((b\ a) \act g)\rho\equiv (b\ a)\act(g\rho)$.
It follows that $(b\ a)\act(g\rho) \aeq h\rho$.
Using \rulefont{{\lambda}{\aeq}{\lambda}ab}, $\lam{a}(g\rho) \aeq \lam{b}(h\rho)$.
The result follows.

For the right to left implication, suppose $(b\ a) \act (g\rho) \aeq h\rho$ and $b \not\in \f{fa}(g\rho)$.
It is a fact that $(b\ a) \act (g\rho)\equiv ((b\ a)\act g)\rho$.
It follows by inductive hypothesis that $(b\ a)\act g \aeq h$.
By Lemma~\ref{lemm.fa.subst.g}, $b \not\in \f{fa}(g)$.
Using \rulefont{{\lambda}{\aeq}{\lambda}ab}, $\lam{a}g \aeq \lam{b}h$ as required.
\end{squishitem}
\end{proof}

\begin{defn}
\label{defn.pi.act.sigma}
Define the substitution $\pi\act\sigma$ by:\ $(\pi\act\sigma)(X)\equiv\pi\act\sigma(X)$.
\end{defn}

Note that $\pi \act \sigma$ is a substitution.
$g(\pi\act \sigma)$ is not a shorthand for $\pi \act (g\sigma)$.
$g(\pi \act \sigma) \aeq \pi \act (g\sigma)$ does not hold in general; for example: $a\equiv a((b\ a)\act\id)\not\equiv (b\ a)\act(a\,\id)$.
However:
\begin{lemm}
\label{lemm.dom.fa}
If $\f{nontriv}(\pi) \cap \f{fa}(g) = \varnothing$ then $g(\pi \act \sigma) \aeq \pi \act (g\sigma)$.
\end{lemm}
\begin{proof}
By a routine induction on $g$. 
\end{proof}

\begin{lemm}
\label{lemm.don't.worry.be.happy}
$\sigma$ solves $\inter{Pr}^D$ if and only if $\sigma \fcomp \rho$ does.

Suppose $\f{nontriv}(\pi) \cap (\f{fa}(r) \cup \f{fa}(s)) = \varnothing$ for every $r \ueq s\in Pr$.
Then $\sigma$ solves $\inter{Pr}^D$ if and only if $\pi \act \sigma$ does. 
\end{lemm}
\begin{proof}
For the first part, we have two cases:
\begin{squishitem}
\item
The case $\sigma$ solves $\inter{Pr}^D$ implies $\sigma \fcomp \rho$ solves $\inter{Pr}^D$.\quad
Suppose $g \ueq h \in \inter{Pr}^D$ and $\sigma$ solves $\inter{Pr}^D$.
Then $g\sigma \aeq h\sigma$.
By Lemma~\ref{lemm.aeq.rho.g}, $g\sigma\rho \aeq h\sigma\rho$.
By Lemma~\ref{lemm.sigma.fcomp.g}, $g(\sigma \fcomp \rho) \aeq h(\sigma \fcomp \rho)$, as required.
\item
The case $\sigma \fcomp \rho$ solves $\inter{Pr}^D$ implies $\sigma$ solves $\inter{Pr}^D$.\quad
Suppose $g \ueq h \in \inter{Pr}^D$ and $\sigma \fcomp \rho$ solves $\inter{Pr}^D$.
Then $g(\sigma \fcomp \rho) \aeq h(\sigma \fcomp \rho)$.
By Lemma~\ref{lemm.sigma.fcomp.g}, $g\sigma\rho \aeq h\sigma\rho$.
By Lemma~\ref{lemm.aeq.rho.g}, $g\sigma \aeq h\sigma$, as required.
\end{squishitem}
For the second part, suppose $\f{nontriv}(\pi) \cap (\f{fa}(r) \cup \f{fa}(s)) = \varnothing$ for every $r \ueq s\in Pr$ and $D = [d_1, \ldots, d_n]$.
Then $\inter{Pr}^D = \{ \lam{d_1}\ldots\lam{d_n}\inter{r}^D \ueq \lam{d_1}\ldots\lam{d_n}\inter{s}^D \mid r \ueq s \in Pr \}$.
By Lemma~\ref{lemm.commutation}, $\f{nontriv}(\pi) \cap (\f{fa}(\inter{r}^D) \cup \f{fa}(\inter{s}^D)) = \varnothing$.
The result follows from Lemma~\ref{lemm.dom.fa}, and using Proposition~\ref{prop.aeq.transitive.g} and Lemma~\ref{lemm.pi.aeq.g}.
\end{proof}

\begin{rmrk}
Lemma~\ref{lemm.don't.worry.be.happy} expresses an intuition that `names of atoms and unknowns on the right in a solution, do not matter', which also underlies the $\pi$ and $\rho$ in Theorem~\ref{thrm.strong}.
$\rho$ is the price we pay for using the same unknowns in Definitions~\ref{defn.terms.g} and~\ref{defn.terms}:
This design decision makes Definition~\ref{defn.nominal.trans.g} compact, but it causes technical problems in Lemma~\ref{qwerty}, because $\sigma(X)$ can introduce new unknowns over whose permission sets (back in the nominal world) we have no control. 
$\rho$ lets us rename those new unknowns as convenient.
As for $\pi$, we discuss it below.

Another design decision is to work with an untyped $\lambda$-terms.
This simplifies our presentation and makes our results slightly more powerful (because they apply to more substitutions), but we cannot be \emph{too} liberal: 
Suppose $\sigma$ solves $\inter{Pr}^D$.
Examining Definition~\ref{defn.nominal.trans.g}, if $X$ occurs in $\inter{Pr}^D$ then it is applied to a number of atoms equal to the length of $D \cap S$. 
So, we will only be interested in $\sigma$ that respect this fragment of typability ($\mathcal V$ will be $\f{fV}(Pr)$): 
\end{rmrk}

\begin{defn}
\label{defn.D-consistent}
Let $\mathcal V$ be a finite set of unknowns.
Call $\sigma$ \deffont{$D$-consistent on $\mathcal V$} when for every $X \in \mathcal V$,\ $\sigma(X) \aeq \lam{a_1}\ldots\lam{a_n}q$ where $n$ is the length of $D \cap S$.
(So $\sigma(X)$ starts with `at least' length-$D \cap S$-many $\lambda$-abstractions.)

Call $\sigma$ \deffont{strictly} $D$-consistent when also, for every $X \in\mathcal V$,\ $\f{fa}(\sigma(X)) \cap D = []$.
\end{defn}

\begin{rmrk}
\label{rmrk.a}
Strictness is motivated by the following examples:
Take $D=[a]$.  

Take ${Pr=\{X^S \ueq \tf f([a]Y^S,Y^S)\}}$ with ${S = \mathbb A^< }$.
Then the problem 
\begin{multline*}
\inter{Pr}^D=\{\lam{a}(X^S\,a) \ueq \lam{a}(\tf f\,(\lam{a}(Y^S a))\,(Y^S a))\} 
\\
\text{has the solution}
\quad
\sigma = [ X^S \ssm \lam{c}(\tf{f}\,(\lam{c}a)\,a),\ Y^S \ssm \lam{c}a ].
\end{multline*}
Now $(\sigma\fcomp\rho)(Y^S) \aeq \inter{\theta}_D^E(Y^S)$ is impossible for any $\rho$, since $\lam{c}a \aeq \lam{a}\inter{\theta(Y^S)}^E$ is impossible.

Take $Pr=\{X^S \ueq \tf f([a]Y^T, Y^T)\}$ with $S = \mathbb A^< $ and $T = \mathbb A^<  \setminus \{a\}$.
Then the problem
\begin{multline*}
\inter{Pr}^D = \{\lam{a}(X^S\,a) \ueq \lam{a}(\tf f\,(\lam{a}Y^T)\,Y^T)\}
\\
\text{has the solution}\quad
\sigma=[X^S \ssm \lam{c}(\tf f\,(\lam{c}a)\,a), Y^T \ssm a].
\end{multline*} 
Now $(\sigma\fcomp\rho)(Y^T) \aeq \inter{\theta}_D^E(Y^T)$ is impossible, since $a\in\f{fa}(a)$ whereas

$a\not\in\f{fa}(\inter{\theta(Y^T)}^E)$ by Lemma~\ref{lemm.not.fa.fresh}.

The $a$ in $\sigma(Y^T)$ for the two $\sigma$ considered above, has nothing to do with the $a$ in $D$. 
We can regard this as an unfortunate `name-clash' which Lemma~\ref{lemm.don't.worry.be.happy} allows us to eliminate with a permutation $\pi$.
\end{rmrk}

More on this in Theorem~\ref{thrm.strong}.
We continue with the proofs:

\begin{defn}
Define the \deffont{arguments of unknowns} in a pattern $q$ by: 
\begin{frameqn}
\begin{gathered}
\f{args}(a) = \varnothing 
\qquad
\f{args}(X)=\varnothing
\qquad
\f{args}(X a_1\ldots a_n) = \{a_1, \ldots, a_n\}
\\
\f{args}(\tf f q_1 \ldots q_n) = \bigcup_{1\leq i\leq n}\f{args}(q_i) 
\qquad
\f{args}(\lam{a}q) = \f{args}(q)
\end{gathered}
\end{frameqn}
\end{defn}
$q \aeq r$ does not imply $\f{args}(q) = \f{args}(r)$.
This is by design. 

\begin{defn} 
\label{defn.qinv}
Suppose $q$ is a $\phi$-pattern and $\f{args}(q) \subseteq E$.  

Define a nominal term $q\qinv$ by: 
\begin{frameqn}
\begin{array}{r@{\ }l@{\qquad}r@{\ }l}
a\qinv\equiv& a 
&
(X b_1\ldots b_{\phi(X)})\qinv\equiv& \pi {\act} X^S
\\
(\lam{a}q)\qinv\equiv& [a]q\qinv
&
(\tf f q_1\ldots q_n)\qinv\equiv& \tf f(q_1\qinv,\ldots,q_n\qinv)
\end{array}
\end{frameqn}
Here, for each $E$ and $X$,\ $\pi$ is a fixed but arbitrary choice of permutation of the atoms in $E$, mapping the $i^\text{th}$ element of $E \cap S$ 
(Definition~\ref{defn.sets.notation}) to $b_i$ for $1 \leq i \leq \phi(X)$. 
\end{defn}

\begin{lemm}
\label{star} 
$\f{args}(q) \subseteq E$ implies $\inter{q\qinv}^E \equiv q$. 
\end{lemm}
\begin{proof} 
By induction on $q$.
\end{proof}

\begin{lemm} 
\label{qwerty} 
Suppose $\mathcal V$ is a finite set of unknowns and $\sigma$ is a $\phi$-pattern substitution, strictly $D$-consistent on $\mathcal V$.

Then there exist $\rho$, $\theta$, and $E$, such that $D \subseteq E$, $\bigcup_{X \in \mathcal V}\f{capt}(\theta(X)) \subseteq E$, and $(\sigma \fcomp \rho)(X) \aeq \inter{\theta}_D^E(X)$ for every $X \in\mathcal V$. 
\end{lemm}
\begin{proof}

Take any $E = [e_1, ..., e_p]$ which includes all atoms in $D$ and in $\{\sigma(X) \mid X \in\mathcal V \}$.
Define $\mathcal V'=\bigcup_{X\in\mathcal V}\f{fV}(\sigma(X))$ (`the unknowns in $\sigma(X)$ for $X\in\mathcal V$').
For each $Y\in\mathcal V'$ choose a fresh $Y'$ such that the length of $E\cap\f{fa}(Y')$ is equal to $\phi(Y)$.
We do this injectively, so that for distinct $Y,Z\in\mathcal V'$, $Y'$ and $Z'$ are also distinct.
Let $\rho$ be any renaming such that $\rho(Y)\equiv Y'$ for all $Y\in\mathcal V'$.

By assumption $\sigma(X) \aeq \lam{a_1}\ldots\lam{a_n}q$ for a $\phi$-pattern $q$, where $[a_1,\dots,a_n]=D \cap S$.
Take $\theta(X) \equiv (q\rho)\qinv$.

We can verify that $\bigcup_{X\in\mathcal V}\f{capt}(\theta(X)) \subseteq E$.
We then reason as follows:
\begin{calcenv}
\inter{\theta}_D^E(X) & \equiv & \lam{a_1} \ldots \lam{a_n} \inter{(q\rho)\qinv}^E & \text{Definition~\ref{defn.thetaC}} \\
                      & \equiv & \lam{a_1} \ldots \lam{a_n} (q\rho) & \text{Lemma~\ref{star}} \\
                      & \equiv & (\lam{a_1} \ldots \lam{a_n} q)\rho & \text{Fact of $\lambda$-terms} \\
                      & \aeq & (\sigma \fcomp \rho)(X) & \text{By construction}
\end{calcenv}
\end{proof}

\begin{thrm}[Strong completeness]
\label{thrm.strong}
Suppose $\f{capt}(Pr) \subseteq D$.

For $\sigma$ strictly $D$-consistent on $\f{fV}(Pr)$ solving $\inter{Pr}^D$ there are $\rho$, $\theta$, and $E$, such that 
$$
\text{$(\sigma \fcomp \rho)(X) \aeq \inter{\theta}_D^E(X)$\ \ for all\ $X \in \f{fV}(Pr)$\ and\ \ \ $\theta$ solves $Pr$.}
$$
For $\sigma$ $D$-consistent on $\f{fV}(Pr)$ solving $\inter{Pr}^D$ there are $\pi$, $\rho$, $\theta$, and $E$, such that 
$$
\text{$\pi{\act} (\sigma \fcomp \rho)(X) \aeq \inter{\theta}_D^E(X)$\ \ for all\ $X\in\f{fV}(Pr)$\ and\ \ \ $\theta$ solves $Pr$.}
$$
\end{thrm}
\begin{proof}
By Lemma~\ref{qwerty}, there are $\rho$, $\theta$, and $E$, such that $(\sigma \fcomp \rho)(X) \aeq \inter{\theta}_D^E(X)$ for all $X \in \f{fV}(Pr)$, $D{\subseteq} E$ and $\bigcup_{X \in \f{fV}(Pr)}\f{capt}(\theta(X)) \subseteq E$.
\ $\f{capt}(Pr) \subseteq D$ and $D \subseteq E$, so $\f{capt}(Pr) \subseteq E$. 
By Theorem~\ref{thrm.sound}, $\theta$ solves $Pr$.

For the second part, write $D = [d_1,\ldots,d_n]$, choose $D' = [d_1',\ldots,d_n']$ fresh (so $d_i'$ is not in $D$, $Pr$, or $\sigma(X)$ for any $X\in\f{fV}(Pr)$), and take $\pi = (d_1'\ d_1) \ldots (d_n'\ d_n)$.
$\pi \act \sigma$ is strictly $D$-consistent and the result follows from the first part and Lemma~\ref{lemm.don't.worry.be.happy}. 
\end{proof}

\section{Conclusions}
\label{sect.conclusion}

In this paper, we have presented a syntax which slightly generalises nominal terms and obtains significantly enhanced properties.
We gain `always fresh' and `always rename' properties (Corollaries~\ref{corr.always.fresh} and~\ref{corr.always.rename}) which are present in first- and higher-order syntax, absent in nominal terms, and regained in permissive nominal terms.

We do not claim a telling difference in expressivity in practice between nominal and permissive nominal terms.
It may indeed be that permissive nominal terms can express some things that nominal terms cannot\footnote{Permission sets may be larger than $\mathbb A^< $ as well as smaller, whereas intuitively in nominal terms freshness contexts can only make permission sets smaller.} but expressivity is not our main motivation in this paper.
The issues which motivate us are with the \emph{properties} of these syntaxes.
As we have seen in this paper, a significant new body of mathematics follows from these changes, which at first seem so innocuous.
It does not stop there; the interested reader can find more in \cite{gabbay:pernas}.

Permissive nominal terms do not obsolete nominal terms. 
To discuss `an arbitrary term', a nominal terms unknown $\dot X$ may be more directly useful than a permissive nominal terms unknown $X^{\mathbb A^< }$ (which means `an arbitrary term, mentioning atoms in $\mathbb A^< $').

We have leveraged the difference between nominal terms and permissive nominal terms to obtain a new unification algorithm which is more efficient in the sense that it is based just on substitutions, and in that sense is also more like the notion of solution familiar from first- and higher-order syntax.
Freshness problems are solved `all in one go' by a distinct algorithm.
We have interpreted nominal unification as a subsystem of permissive nominal unification (Section~\ref{sec.relation.to.nominal.terms}).

One nice way to view this interpretation is that $\mathbb A^< $ plays the role of `the atoms we had so far' and $\mathbb A^> $ that of `the atoms we will generate fresh in the future'.
Finally, we have exhibited permissive nominal unification as equivalent to higher-order pattern unification.

\subsection{Related work}

\subsubsection*{Infinite sets of atoms}

Permissive nominal terms are based on the idea of infinite and co-infinite sets of atoms $S$.
This is new, but it emerges from a literature rich in precedents.
As we noted in Remark~\ref{rmrk.fm.not.enough}, infinite and co-infinite sets break with the standard nominal sets semantics from \cite{gabbay:newaas-jv}, which does not admit them because they do not have finite support.
This is however not a serious mathematical problem: the idea of relaxing the `finite support' property of nominal sets to infinite generalisations is natural.

As far as we know this was first discussed in \cite{gabbay:fmhotn}, where the second author proposed to identify `small' sets of atoms not with cardinality but with \emph{well-orderability} of the atoms in the set --- so a set of atoms is small when it can be assigned a cardinal size, and large when it cannot be assigned a cardinal size (internally, within the model) --- but we do not commit ourselves to how large those cardinals can get, and in particular they could be infinite.
See \cite{gabbay:genmn} for a more extended treatment of the same ideas.

Pitts referred to the possibility of using nominal sets with infinite support in \cite{pitts:nomlfo-jv}, in order to obtain a complete semantics for nominal logic.  
This idea was taken up by Cheney in \cite{cheney:comhtn}.  
Thus, in \cite{cheney:comhtn} `small' sets of atoms are identified as elements of a \emph{support ideal} (Definition~4.1 of \cite{cheney:comhtn}), which are similar in spirit to the set of permission sets from Definition~\ref{defn.the.comb}.
In Definition~\ref{defn.the.comb} we give a concrete set of permission sets, but in a footnote to that definition we also identify some reasonable abstract conditions for the set of permission sets to satisfy such that the proofs in this paper still work. 
These conditions are extremely mild; the structure actually required of permission sets is much weaker than that provided by Definition~\ref{defn.the.comb}.

\subsubsection*{Namespaces}

Since this paper was written, the second author has prepared a manuscript treating \emph{permissive nominal algebra} \cite{gabbay:pernas} (a precursor is \cite{gabbay:semnt-ea}).
There, the set of all atoms is taken to be uncountable and permission sets are taken to be all countably infinite sets of atoms.
This setting is sufficiently general to accommodate many different notions of permission sets as subsystems.
In particular the permission sets of this paper feature as a `namespace' identified by $\mathbb A^< $; thus, the ideas in this paper slot quite nicely into a more abstract setting.

\subsubsection*{`Free atoms of' as distinct from `support of the denotation of'}

It may be useful to note some work to which this paper is not related.
One nice aspect of permissive nominal terms is that they give us a notion of `free atoms of a term' $\f{fa}(r)$.  
The judgement $\nabla\cent a\#r$ of nominal terms corresponds to the judgement $a\not\in\f{fa}(r)$ of permissive nominal terms (see Lemma~\ref{lem.interpretation.fresh.preserve.iff})
and both correspond to the informal judgement `$a$ is not free in $r$'.

Nominal sets have a native notion of semantic freshness; $a\#_{sem} r$ means `$a$ is not in the support of the denotation of $r$'.
Semantic freshness is a distinct concept.
The reader should not confuse semantic freshness with the intensional judgement `free atoms of' used in this paper.

Semantic freshness may be expressed using `free atoms of' and equality \cite{gabbay:forcie,gabbay:nomuae,gabbay:pernas}. 

\subsubsection*{Patterns}

Patterns emerged by studying Skolemisation of unification problems \cite{miller:uniump}; they proved useful in the unification of higher-order abstract syntax terms \cite{miller:logpll}.  

Cheney proposed a two-stage translation of higher-order pattern unification to nominal unification \cite{cheney:relnho}, first by exhibiting a translation of higher-order pattern unification to nominal pattern unification (where nominal patterns are a variant of nominal terms, with a \emph{concretion} operator, where unknowns have empty support), followed by a translation between nominal pattern unification and nominal unification.

Levy and Villaret related nominal unifiability with higher-order pattern unifiability in \cite{levy:nomufh}.
Our treatment translates \emph{solutions} as well, handles a more general class of higher-order patterns than considered in \cite{levy:nomufh}, and we prove our translation complete (Theorem~\ref{thrm.strong}) and optimal (Subsection~\ref{subsect.capt}).

Note that translating solutions really matters: it might have been, for example, that higher-order pattern unification and permissive nominal unification have the same notion of unifiability --- but very different notions of solution and sets of solutions.
For comparison, the $\lambda$-calculus and combinators express similar notions of computability, but have very different notions of reduction and computation.

The version of higher-order pattern unification which we examine is more general than usual, since we do not type our $\lambda$-terms.
We show how to retain enough of the properties of typing to avoid `silly' problems.
For example, we do not consider untyped higher-order pattern unification problems like $X \abeq \lam{a}(Xa)$, because this cannot be expressed as a unification of two $\phi$-patterns for any $\phi$ (Definition~\ref{defn.phi.patterns}) --- we impose a structural condition on our patterns that $X$ should be applied to a consistent, fixed number of arguments.

We hypothesise that the results in this paper would work in a typed setting; the conditions which our proofs depend on to work are just structural ones, which would also be guaranteed by types.
We have not investigated the effects of $\eta$-conversion; this is future work.

\subsubsection*{The broader literature}

Hamana's $\beta_0$ unification of $\lambda$-terms with holes adds a capturing substitution \cite{hamana:logplb}.
Level 2 variables (which are instantiated) are annotated with level 1 variable symbols that \emph{may} appear in them; permissive nominal terms move in this direction in the sense that permission sets also describe which level 1 variable symbols (we call them atoms in this paper) may appear in them, though with our permission sets there are infinitely many that may, and infinitely many that may not.
Another significant difference is that the treatment of $\alpha$-equivalence in Hamana's system is not nominal (not based on permutations) and unlike our systems, Hamana's does not have most general unifiers.

Similarly, Qu-Prolog \cite{nickolas:qupua} adds level 2 variables, but does not manage $\alpha$-conversion in nominal style, and, for better or for worse, the system is more ambitious in what it expresses, and thus loses mathematical properties (unification is semi-decidable, most general unifiers need not exist).


\subsection{Future work}

We propose permissive nominal terms as a syntax for designing logics and $\lambda$-calculi in the spirit of nominal terms. 

A first implementation of permissive nominal unification has been made \cite{mulligan:imppnt} by the third author.

We have begun to apply permissive nominal terms to construct novel logics and $\lambda$-calculi, taking advantage of their properties to simplify the theory.
It is simply very useful to reason on terms (without a freshness context), to have an inexhaustible supply of fresh names, and to be able to quotient by $\alpha$-equivalence.
We note in particular the papers \cite{gabbay:curhif-jv,gabbay:twollc,gabbay:oneaah-jv,gabbay:capasn-jv,gabbay:nomalc}, in which we have struggled with the theory of $\alpha$-equivalence given to us by nominal terms; these might benefit from the use of permissive nominal terms.

Permissive nominal terms syntax has two levels of variable, atoms $a$ and unknowns $X$.
There is no reason to stop there; we have already considered syntaxes with more than two levels of variable, for example \cite{gabbay:hienr,gabbay:lamcce}.
Again, we had difficulty with $\alpha$-conversion and fresh atoms.
It would be very interesting to revisit this material armed with the permissive ideas of this paper. 

Finally, it may be possible to extend the techniques of this paper to biject full higher-order unification with an enrichment of (permissive) nominal unification.

\hyphenation{Mathe-ma-ti-sche}

\appendix

\section{Supplementary proofs}
\label{sect.omitted.proofs}

Definition~\ref{defn.depth} is useful for Proposition~\ref{prop.aeq.transitive}. 
\begin{defn}
\label{defn.depth}
Define the \deffont{size} of a term $r$ by:
\begin{frameqn}
\begin{array}{r@{\ }l@{\qquad}r@{\ }l}
\f{size}(a) =& 0
&
\f{size}(\tf{f}(r_1, \ldots, r_n)) =& \sum_{1 \leq i \leq n}\f{size}(r_i)
\\
\f{size}([a]r) =& 1 {+} \f{size}(r)
&
\f{size}(\pi {\act} X^S) =& 0
\end{array}
\end{frameqn}
\end{defn}

\noindent
Proof of Proposition~\ref{prop.aeq.transitive}:

\begin{proof}
Reflexivity is shown by induction on terms.
Symmetry is shown by induction on derivations.
Transitivity is shown by induction on the size of a term.
We consider one case from the proof of symmetry and one from the proof of transitivity:
\begin{squishitem}
\item
\emph{$[a]r \aeq [b]s$ implies $[b]s\aeq [a]r$.}
\quad
Suppose $(b\ a) \act r \aeq s$ and $b \not\in \f{fa}(r)$.
By Lemma~\ref{lemm.pi.ftma} $a \not\in \f{fa}((a\ b) \act r)$.
From Lemmas~\ref{lemm.permutation.comp} and~\ref{lemm.equality.permutation},\ $r \aeq (a\ b) \act s$.
By Lemma~\ref{lemm.aeq.ftma.pres}, $a \not\in \f{fa}(s)$.
It is a fact that the size of terms is unaffected by permutation, so by inductive hypothesis $(a\ b) \act s \aeq r$.
Extending with \rulefont{{\aeq}[b]} we obtain $[b]s \aeq [a]r$ as required.
\item
\emph{$[a]r \aeq [b]s$ and $[b]s \aeq [c]t$ imply $[a]r \aeq [c]t$.}\quad
Suppose $(b\ a) \act r \aeq s$, $(c\ b) \act s \aeq t$, $b \not\in \f{fa}(r)$ and $c \not\in \f{fa}(s)$.
By Lemma~\ref{lemm.equality.permutation} $(c\ b) \act ((b\ a) \act r) \aeq (c\ b) \act s$.
It is a fact that the size of terms is unaffected by permutation, so by inductive hypothesis also $(c\ b) \act ((b\ a) \act r) \aeq t$, and by Lemmas~\ref{lemm.ds.pi} and~\ref{lemm.permutation.comp} $(c\ a) \act r \aeq t$.
By Lemma~\ref{lemm.aeq.ftma.pres} $c \not\in \f{fa}((b\ a) \act r)$, so by Lemma~\ref{lemm.pi.ftma} $c \not\in (b\ a) \act \f{fa}(r)$. 
Therefore, $c \not\in \f{fa}(r)$.
We use \rulefont{{\aeq}[b]} to derive $[a]r \aeq [c]t$ as required.
\qedhere
\end{squishitem}
\end{proof}

\noindent
Proof of Lemma~\ref{lemm.substitution}:

\begin{proof}
By induction on $r$.
\begin{squishitem}
\item
The cases $a$ and $\tf{f}(r_1, \ldots, r_n)$ are routine.
\item
The case $[a]r$.\quad
We reason as follows:
\begin{calcenv}
\f{fa}(([a]r)\theta) & \equiv & \f{fa}([a]r\theta) & \text{Definition~\ref{defn.subst.action}} \\
                     & =      & \f{fa}(r\theta) \setminus \{ a \} & \text{Definition~\ref{defn.fa}} \\
                     & \subseteq & \f{fa}(r) \setminus \{ a \}   & \text{Inductive hypothesis} \\
                     & =      & \f{fa}([a]r) & \text{Definition~\ref{defn.fa}}
\end{calcenv}
The result follows.
\item
The case $\pi \act X^S$.\quad
By Definition~\ref{defn.fa}, $\f{fa}(\pi\act X^S)=\pi \act S$.
By Definition~\ref{defn.subst}, $\f{fa}(\theta(X^S))\subseteq S$.
Using Lemma~\ref{lemm.pi.ftma}, it follows $\f{fa}(\pi\act \theta(X^S))\subseteq \pi\act S$. 
\end{squishitem}
\end{proof}

\noindent
Proof of Lemma~\ref{lem.interpretation.fresh.preserve.iff}:

\begin{proof}
We handle the two implications separately.
\begin{squishitem}
\item
The case $\iota(\dot{a}) \not\in \f{fa}(\interdelta{\dot{r}})$ implies $\Delta \cent \dot{a} \# \dot{r}$.\quad
We proceed by induction on $\dot r$.\quad
\begin{squishitem}
\item
The cases $\dot b$ and $\tf{f}(\dot r_1, \ldots, \dot r_n)$ are straightforward.
\item
The case $[\dot a]\dot r$.\quad
Since $\Delta \cent \dot a \# [\dot a]\dot r$ always, using \rulefont{{\#}[\dot a]}.
\item
The case $[\dot b]\dot r$.\quad
Suppose $\iota(\dot a) \not\in \f{fa}(\interdelta{[\iota(\dot b)]\dot r})$ and $\iota(\dot a) \not\in \f{fa}(\interdelta{\dot r}) \setminus \{ \iota(\dot b) \}$.
Then $\iota(\dot a) \not\in \f{fa}(\interdelta{\dot r})$, therefore $\Delta \cent \dot a \# \dot r$ by inductive hypothesis.
Using \rulefont{{\#}[\dot b]}, we have $\Delta \cent \dot a \# [\dot b]\dot r$, and the result follows.
\item
The case $\dot\pi \act \dot X$.\quad
Suppose $\iota(\dot a) \not\in \f{fa}(\interdelta{\dot\pi \act \dot X})$.
Then $\iota(\dot a) \not\in \inter{\dot\pi} \act S$, where $S = \mathbb A^<  \setminus \{ \iota(\dot a) \mid \dot a \# \dot X \in \Delta \}$.
But $\inter{\dot\pi} \act \mathbb A^<  \setminus \{ \iota(\dot a) \mid \dot a \# \dot X \in \Delta \}$ is the same as $\inter{\dot\pi} \act \mathbb A^<  \setminus \inter{\dot\pi} \act \{ \iota(\dot a) \mid \dot a \# \dot X \in \Delta \}$.
Then $\inter{\dot\pi} \act \{ \iota(\dot a) \mid \dot a \# \dot X \in \Delta \} = \{ \inter{\dot\pi} \act \iota(\dot a) \mid \dot a \# \dot X \in \Delta \}$.
Using Definition~\ref{defn.interpretation.nominal}, and the fact permutations are bijective, we have $\{ \inter{\dot\pi} \act \iota(\dot a) \mid a \# X \in \Delta \} = \{ \iota(\dot\pi^\mone \act \dot a) \mid \dot\pi^\mone \act \dot a \# X \in \Delta \}$.
We use \rulefont{{\#}\dot X} to obtain $\Delta \cent \dot a \# \dot X$, and we have the result.
\end{squishitem}
\item
The case $\Delta \cent \dot{a} \# \dot{r}$ implies $\iota(\dot{a}) \not\in \f{fa}(\interdelta{\dot{r}})$.\quad
We proceed by induction on the derivation of $\Delta \cent \dot{a} \# \dot{r}$.
\begin{squishitem}
\item
The cases \rulefont{{\#}\dot b} and \rulefont{{\#}\tf f} are routine.
\item
The case \rulefont{{\#}[\dot a]}.\quad
Suppose $\Delta \cent \dot a \# [\dot a]\dot r$ using \rulefont{{\#}[\dot a]}.
Then $\interdelta{[\dot a]\dot r} \equiv [\iota(\dot a)]\interdelta{\dot r}$.
Further, $\iota(\dot a) \not\in \f{fa}(\interdelta{\dot r}) \setminus \{ \iota(\dot a) \}$, and the result follows.
\item
The case \rulefont{{\#}[\dot b]}.\quad
Suppose $\Delta \cent \dot a \# \dot r$ and $\iota(\dot a) \not\in \f{fa}(\dot r)$ by assumption.
Then $\Delta \cent \dot a \# [\dot b]\dot r$ by \rulefont{{\#}[\dot b]}.
Further, $\f{fa}(\interdelta{[\dot b]\dot r}) = \f{fa}(\interdelta{\dot r}) \setminus \{ \iota(\dot b) \}$, and the result follows.
\item
The case \rulefont{{\#}\dot X}.\quad
Suppose $\dot\pi^\mone(\dot a) \# \dot X \in \Delta$, and $\Delta \cent \dot a \# \dot \pi \act \dot X$ by \rulefont{{\#}\dot X}.
Then $\interdelta{\dot \pi \act \dot X} = \inter{\dot\pi} \act X^S$ where $S = \mathbb A^<  \setminus \{ \iota(\dot a) \mid \dot a \# \dot X \in \Delta \}$.
Further, $\f{fa}(\inter{\dot\pi} \act X^S) = \inter{\dot\pi} \act S$.
The result follows from Definition~\ref{defn.interpretation.nominal}.
\end{squishitem}
\end{squishitem}
\end{proof}

\noindent
Proof of Theorem~\ref{thrm.interpretation.injective}:

\begin{proof}
We prove that $\interdelta{\dot{r}} \aeq \interdelta{\dot{s}}$ implies $\Delta \cent \dot{r} = \dot{s}$ by induction on the derivation of $\interdelta{\dot{r}} \aeq \interdelta{\dot{s}}$:
\begin{squishitem}
\item
The cases $\dot{a}$ and $\tf f(\dot r_1,\ldots,\dot r_n)$ are routine.
\item
The case \rulefont{{\aeq}[a]}.\quad
Suppose $\interdelta{\dot r} \aeq \interdelta{\dot s}$ and $\Delta \cent \dot r = \dot s$.
Then using \rulefont{{\aeq}[a]}, $[\iota(\dot a)]\interdelta{\dot r} \aeq [\iota(\dot a)]\interdelta{\dot s}$ and $\Delta \cent [\dot a]\dot r = [\dot a]\dot s$ also, using \rulefont{{=}[\dot a]}.
The result follows, as $[\iota(\dot a)]\interdelta{\dot r} = \interdelta{[\dot a]\dot r}$.
\item
The case \rulefont{{\aeq}[b]}.\quad
Suppose 
$(\iota(\dot{b})\ \iota(\dot{a})) \act \interdelta{\dot{r}} \aeq \interdelta{\dot{s}}$ and $\iota(\dot{b}) \not\in \f{fa}(\interdelta{\dot{r}})$.
By Lemmas~\ref{lem.nominal.pi.interpretation} and~\ref{lem.interpretation.fresh.preserve.iff} $\interdelta{(\dot{b}\ \dot{a})\act\dot{r}} \aeq \interdelta{\dot{s}}$\ and $\Delta \cent \dot{b} \# \dot{r}$.
By inductive hypothesis $\Delta \cent (\dot{b}\ \dot{a}) \act \dot{r} = \dot{s}$.
We use \rulefont{{=}[\dot b]}.\item
The case \rulefont{{\aeq}X}.\quad
Suppose $\inter{\dot\pi}|_S = \inter{\dot\pi'}|_S$ where $S = \mathbb A^<  {\setminus} \{ \iota(\dot{a}) \mid \dot{a} \# \dot{X} {\in} \Delta \}$.
$\iota$ is injective, so $a\#\dot X\in \Delta$ for all $\dot a$ such that $\dot\pi(\dot a)\neq\dot\pi'(\dot a)$. 
The result follows by \rulefont{{=}\dot X}.
\end{squishitem}
We prove that $\Delta \cent \dot{r} = \dot{s}$ implies $\interdelta{\dot r} \aeq \interdelta{\dot s}$ by induction on the derivation of $\Delta \cent \dot{r} = \dot{s}$:
\begin{squishitem}
\item
The cases \rulefont{{=}\dot a}, \rulefont{{=}\tf f } and \rulefont{{=}[\dot a]} are straightforward.
\item
The case \rulefont{{=}[\dot b]}.\quad
Suppose $\Delta \cent (\dot{b}\ \dot{a})\act\dot{r} = \dot{s}$ and $\Delta \cent \dot{b} \# \dot{r}$.
By inductive hypothesis and Lemma~\ref{lem.nominal.pi.interpretation}, $(\dot{b}\ \dot{a})\act\interdelta{\dot{r}} \aeq \interdelta{\dot{s}}$.
By Lemma~\ref{lem.interpretation.fresh.preserve.iff}, $\iota(\dot{b}) \not\in \f{fa}(\interdelta{\dot{r}})$.
The result follows by \rulefont{{\aeq}[b]}.
\item
The case \rulefont{{=}\dot X}.\quad
Recall that $\interdelta{\dot{\pi}\act\dot{X}} = \inter{\dot{\pi}} \act X^S$ and $\interdelta{\dot{\pi}'\act\dot{X}} = \inter{\dot{\pi}'} \act X^S$ where $S = \mathbb A^<  \setminus \{ \iota(\dot{a}) \mid \dot{a} \# \dot{X} \in \Delta \}$.
Suppose $\dot\pi(\dot a)\neq\dot\pi'(\dot a)$ implies $\Delta \cent \dot{a} \# \dot{X}$.
Using Lemma~\ref{lem.interpretation.fresh.preserve.iff}, $\inter{\dot\pi}(\iota(\dot a))\neq\inter{\dot \pi'}(\iota(\dot a))$ implies $\iota(\dot{a}) \not\in S$.
The result follows by \rulefont{{\aeq}X}.
\end{squishitem}
\end{proof}

\noindent
Proof of Lemma~\ref{lemm.substitution.commute.interp}:

\begin{proof}
By induction on $\dot r$.
We show two cases:
\begin{squishitem}
\item
The case $[\dot a]\dot r$.\quad
We reason as follows:
\begin{calcenv}
\interdelta{([\dot a] \dot r)\dot\theta} & \equiv & \interdelta{[\dot a]\dot r\dot\theta} & \text{Definition~\ref{defn.nominal.substitution.action}} \\
                                         & \equiv & [\iota(\dot a)]\interdelta{\dot r\dot\theta} & \text{Definition~\ref{defn.interpretation.nominal}} \\
                                         & \equiv & [\iota(\dot a)]\interdelta{\dot r}\inter{(\Delta, \dot\theta)} & \text{Inductive hypothesis} \\
                                         & \equiv & ([\iota(\dot a)]\interdelta{\dot r})\inter{(\Delta, \dot\theta)} & \text{Fact} \\
                                         & \equiv & \interdelta{[\dot a]\dot r}\inter{(\Delta, \dot\theta)} & \text{Definition~\ref{defn.interpretation.nominal}}
\end{calcenv}
The result follows.
\item
The case $\dot\pi \act \dot X$.\quad
We reason as follows:
\begin{calcenv}
\interdelta{(\dot\pi \act \dot X)\dot\theta} & \equiv & \interdelta{\dot\pi \act \dot\theta(\dot X)} & \text{Definition~\ref{defn.nominal.substitution.action}} \\
                                             & \equiv & \inter{\dot\pi} \act \interdelta{\dot\theta(\dot X)} & \text{Definition~\ref{defn.interpretation.nominal}} \\
                                             & \equiv & \inter{\dot\pi} \act \inter{\dot\theta}(\interdelta{\dot X}) &\text{Definition~\ref{defn.subst}} 
\end{calcenv}
The result follows.
\end{squishitem}
\end{proof}

\noindent
Proof of Lemma~\ref{lemm.basic}:

\begin{proof}
We prove by induction on $r$ that $a\not\in\f{fa}(r)$ implies $(b\ a)\act r\aeq r$:
\begin{squishitem}
\item
The cases $c$ and $\tf{f}(r_1, \ldots, r_n)$.\quad
Straightforward.
\item
The cases $[a]r$.\quad
We show $(b\ a) \act [a]r \aeq [a]r$ where $b \not\in \f{fa}([a]r)$, hence $b \not\in \f{fa}(r)$ and $a \not\in \f{fa}(r)$.
By Definition~\ref{defn.perm}, $(b\ a) \act [a]r \aeq [b](b\ a) \act r$.
By Definition~\ref{defn.aeq}, we must show $(a\ b) \act ((b\ a) \act r) \aeq r$ where $a \not\in \f{fa}((b\ a) \act r)$.
By Lemma~\ref{lemm.pi.ftma}, $b \not\in \f{fa}(r)$ implies $a \not\in \f{fa}((b\ a) \act r)$.
By Lemma~\ref{lemm.permutation.comp}, $(a\ b) \act ((b\ a) \act r) \aeq ((a\ b) \fcomp (b\ a)) \act r$.
As $\pi = \pi^\mone$, we have $r \aeq r$.
The result follows from Proposition~\ref{prop.aeq.transitive}.
\item
The case $[b]r$ is similar.
\item
The case $[c]r$.\quad
Suppose $b \not\in \f{fa}([c]r)$, $a \not\in \f{fa}([c]r)$ and $a,b \not\in \f{fa}(r)$.
We show $(b\ a) \act [c]r \aeq [c]r$.
By Definition~\ref{defn.perm}, $(b\ a) \act [c]r \equiv [c](b\ a) \act r$.
We use \rulefont{{\aeq}[a]} and the inductive hypothesis to obtain $(b\ a) \act r \aeq r$.
\item
The case $\pi \act X^S$.\quad
Suppose $b \not\in \f{fa}(\pi \act X^S)$, $a \not\in \f{fa}(\pi \act X^S)$ and $a, b \not\in \pi \act S$.
By Definition~\ref{defn.perm}, $(b\ a) \act (\pi \act X^S) \equiv ((b\ a) \fcomp \pi) \act X^S$.
Using \rulefont{{\aeq}X}, $((b\ a) \fcomp \pi) \act X^S \aeq \pi \act X^S$ whenever $((b\ a) \fcomp \pi)|_S = \pi|_S$.
As $a, b \not\in \pi \act S$, $((b\ a) \fcomp \pi)|_S = \pi|_S$.
The result follows.
\end{squishitem}
We prove by induction on $r$ that $(b\ a)\act r\aeq r$ implies $a\not\in\f{fa}(r)$:\quad
\begin{squishitem}
\item
The case $a$, $b$, $c$ and $\tf{f}(r_1, \ldots, r_n)$ are routine.
\item
The case $[a]r$.\quad
Suppose $(b\ a) \act [a]r \aeq [a]r$.
By Definition~\ref{defn.perm}, $(b\ a) \act [a]r \equiv [b](b\ a) \act r$.
By Definition~\ref{defn.aeq}, $[b](b\ a) \act r \aeq [a]r$ whenever $(a\ b) \act ((b\ a) \act r) \aeq r$ with $a \not\in \f{fa}((b\ a) \act r)$.
By Lemma~\ref{lemm.permutation.comp}, and the fact that swappings are self-inverse, $(a\ b) \act ((b\ a) \act r) \equiv r$.
By assumption, $b \not\in \f{fa}(r)$.
By Lemma~\ref{lemm.pi.ftma}, $a \not\in \f{fa}((b\ a) \act r)$.
The result follows.
\item
The case $[b]r$ is similar.
\item
The case $[c]r$.\quad
By inductive hypothesis $(b\ a) \act r \aeq r$ implies $a \not\in \f{fa}(r)$.
The result follows from $[c](b\ a) \act r \equiv (b\ a) \act [c]r$.
\item
The case $\pi \act X^S$.\quad
Suppose $(b\ a) \act \pi \act X^S \aeq \pi \act X^S$.
By Definition~\ref{defn.perm}, $(b\ a) \act \pi \act X^S \equiv ((b\ a) \fcomp \pi) \act X^S$.
Using \rulefont{{\aeq}X}, $((b\ a) \fcomp \pi) \act X^S \aeq \pi \act X^S$ whenever $(b\ a) \fcomp \pi|_S = \pi|_S$.
It is a fact that $(b\ a) \fcomp \pi|_S = \pi|_S$ only when $b, a \not\in \pi \act S$.
The result follows.
\end{squishitem}
\end{proof}

\noindent
The following definition is used in the proof of Proposition~\ref{prop.supp.reduct.strong.normalisation}:

\begin{defn}
\label{defn.supp.inc.size}
Define the \deffont{size} of a support inclusion problem $\f{size}(\f{Inc})$ to be a tuple $(T, A, P, S)$, where:
\begin{squishlist}
\item
$T$ is the number of term-formers appearing within terms in $\f{Inc}$,
\item
$A$ is the number of abstractions appearing within terms in $\f{Inc}$,
\item
$P$ is the number of permutations, distinct from the identity permutation, appearing within terms in $\f{Inc}$, and
\item
$S$ is the number of support inclusions within $\f{Inc}$.
\end{squishlist}
We order tuples lexicographically.
\end{defn}

\noindent
Proof of Proposition~\ref{prop.supp.reduct.strong.normalisation}:

\begin{proof}
By case analysis, checking all simplification rules reduce the measure defined in Definition~\ref{defn.supp.inc.size}.
\begin{squishitem}
\item
The case $a \sqsubseteq T, \f{Inc'}$.\quad
Suppose $a \in T$, $\f{size}(a \sqsubseteq T, \f{Inc'}) = (T, A, P, S)$, and $a \sqsubseteq T, \f{Inc'} \simpto{} \f{Inc'}$ by \rulefont{{\sqsubseteq}a}.
Then $\f{size}(\f{Inc'}) = (T, A, P, S-1)$.
Otherwise, suppose $a \not\in T$ so \rulefont{{\sqsubseteq}a} is not applicable.
No other rule is applicable by assumption, and the result follows.
\item
The case $\tf{f}(r_1, \ldots, r_n) \sqsubseteq T, \f{Inc'}$.\quad
Suppose $\f{size}(\tf{f}(r_1, \ldots, r_n) \sqsubseteq T, \f{Inc'}) = (T,A,P,S)$ and $\tf{f}(r_1, \ldots, r_n) \sqsubseteq T, \f{Inc'} \simpto{} \f{r_1} \sqsubseteq T, \ldots, \f{r_n} \sqsubseteq T, \f{Inc'}$ by \rulefont{{\sqsubseteq}\tf f}.
Then $\f{size}(\f{r_1} \sqsubseteq T, \ldots, \f{r_n} \sqsubseteq T, \f{Inc'}) = (T-1, A, P, S+n-1)$.
The result follows from the ordering.
\item
The case $[a]r \sqsubseteq T, \f{Inc'}$.\quad
Suppose $\f{size}([a]r \sqsubseteq T, \f{Inc'}) = (T,A,P,S)$ and $[a]r \sqsubseteq T, \f{Inc'} \simpto{} r \sqsubseteq T \cup \{ a \}, \f{Inc'}$ by \rulefont{{\sqsubseteq}[]}.
Then $\f{size}(r \sqsubseteq T \cup \{ a \}, \f{Inc'}) = (T,A-1,P,S)$ and the result follows.
\item
The case $\pi \act X^S \sqsubseteq T, \f{Inc'}$.\quad
Suppose $\f{size}(\pi \act X^S \sqsubseteq T, \f{Inc'}) = (T,A,P,S)$.
Then, if $S \subseteq \pi^\mone \act T$, we have $\pi \act X^S \sqsubseteq T, \f{Inc'} \simpto{} \f{Inc'}$ by \rulefont{{\sqsubseteq}X'}, with measure $(T,A,P,S-1)$.
Otherwise, if we have $S \not\subseteq \pi^\mone \act T$ and $\pi \not= \id$, we have $\pi \act X^S \sqsubseteq T, \f{Inc'} \simpto{} X^S \sqsubseteq \pi^\mone \act T, \f{Inc'}$ with measure $(T,A,P-1,S)$.
By assumption, no other rules are applicable.
\end{squishitem}
\end{proof}

\noindent
The following definition is used in the proof of Proposition~\ref{prop.unification.algorithm.strong.normalisation}:

\begin{defn}
\label{defn.size.unification.problem}
Define the \deffont{size} of a unification problem $\f{size}(Pr)$ to be a tuple $(E, T, A)$, where:
\begin{squishlist}
\item
$E$ is the number of equalities appearing in the unification problem,
\item
$T$ is the number of term-formers appearing within terms in the equalities of the unification problem,
\item
$A$ is the number of abstractions appearing within terms in the equalities of the unification problem.
\end{squishlist}
We order tuples lexicographically.
\end{defn}

\noindent
Proof of Proposition~\ref{prop.unification.algorithm.strong.normalisation}:

\begin{proof}
By case analysis, checking all simplification rules reduce the measure defined in Definition~\ref{defn.size.unification.problem}.
We consider three cases:
\begin{squishlist}
\item
The case \rulefont{{\ueq}[b]}.\quad
Suppose $b \not\in \f{fa}(r)$ and $\mathcal V; [a]r \ueq [b]s, Pr \simpto{} \mathcal V; (b\ a) \act r \ueq s, Pr$ by \rulefont{{\ueq}[b]}.
Suppose further that $\f{size}([a]r \ueq [b]s, Pr) = (E, T, A)$.
Then $\f{size}((b\ a) \act r \ueq s, Pr) = (E, T, A-1)$, and the result follows.
\item
The case \rulefont{I1}.\quad
Suppose $X^S \not\in \f{fV}(s)$ and $\f{fa}(S) \subseteq \pi \act S$, and $\mathcal V; \pi \act X^S \ueq s, Pr \simpto{[X^S \ssm \pi^\mone \act s]} \mathcal V; Pr[X^S \ssm \pi^\mone \act s]$ by \rulefont{I1}.
Suppose further that $\f{size}(\pi \act X^S \ueq s, Pr) = (E, T, A)$.
Then $\f{size}(Pr[X^S \ssm \pi^\mone \act s]) = (E-1, T, A)$, and the result follows.
\item
The case \rulefont{I3}.\quad
It is a fact that rewriting with \rulefont{I3} terminates, because of the condition that $Pr_{\sqsubseteq}$ is non-trivial.
\end{squishlist}
\end{proof}

\noindent
Proof of Lemma~\ref{lemm.preservation.of.solutions}:

\begin{proof}
The empty set cannot be simplified, so suppose $Pr = r \ueq s, Pr'$ where the simplification rule acts on $r\ueq s$.
We reason by cases: 
\begin{squishitem}
\item
The cases \rulefont{{\ueq}a}, \rulefont{{\ueq}\tf f} and \rulefont{{\ueq}X} are straightforward.
\item
The case \rulefont{{\ueq}[a]}.\quad
Suppose $Pr = [a]r \ueq [a]s, Pr'$ and $[a]r \ueq [a]s, Pr' \simpto{} r \ueq s, Pr'$ by \rulefont{{\ueq}[a]}.
Then:
\begin{squishitem}
\item
Suppose $([a]r)\theta \aeq ([a]s)\theta$.\quad
By Definition~\ref{defn.subst.action}, $[a](r\theta) \aeq [a](s\theta)$.
By the rules in Definition~\ref{defn.aeq}, $r\theta \aeq s\theta$.
The result follows.
\item
Suppose $r\theta \aeq s\theta$.\quad
By the rules in Definition~\ref{defn.subst.action}, $[a](r\theta) \aeq [a](s\theta)$.
By Definition~\ref{defn.subst.action}, $([a]r)\theta \aeq ([a]s)\theta$, as required.
\end{squishitem}
\item
The case \rulefont{{\ueq}[b]}.\quad
Suppose $Pr=[a]r\ueq [b]s,Pr'$, $b \not\in \f{fa}(r)$ and $Pr\simpto{}(b\ a)\act r\ueq s,Pr'$ with \rulefont{{\ueq}[b]}.
Then:
\begin{squishitem}
\item
Suppose $([a]r)\theta\aeq ([b]s)\theta$.
By Definition~\ref{defn.subst.action}, $[a](r\theta)\aeq [b](s\theta)$.
By the rules in Definition~\ref{defn.aeq},\ $(b\ a)\act (r\theta)\aeq s\theta$.
By Lemma~\ref{lemm.sub.perm} and Proposition~\ref{prop.aeq.transitive},\ $((b\ a)\act r)\theta\aeq s\theta$.
The result follows.
\item
Suppose $((b\ a)\act r)\theta\aeq s\theta$.
By Lemma~\ref{lemm.sub.perm} and Proposition~\ref{prop.aeq.transitive},\ $(b\ a)\act (r\theta)\aeq s\theta$.
By Lemma~\ref{lemm.substitution}, $b\not\in\f{fa}(r\theta)$.
Using \rulefont{{\aeq}[b]}, $[a](r\theta)\aeq [b](s\theta)$.
By Definition~\ref{defn.subst.action} $[a](r\theta)\aeq [b](s\theta)$, as required.
\end{squishitem}
\end{squishitem}
\end{proof}

\noindent
Proof of Lemma~\ref{lemm.simp.fv.pres}:

\begin{proof}
As the empty set cannot be simplified, it must be that $Pr = r \ueq s, Pr'$.
It therefore suffices to perform case analysis on the simplification of $r \ueq s$.
At each stage, without loss of generality, assume $Pr'$ has been simplified by non-instantiating rules as much as possible.
\begin{squishitem}
\item
The cases \rulefont{{\ueq}a}, \rulefont{{\ueq}\tf f} and \rulefont{{\aeq}X} are routine.
\item
The case \rulefont{{\ueq}[a]}.\quad
Suppose $\mathcal V; [a]r \ueq [a]s, Pr'$ and $\f{fV}([a]r \ueq [a]s, Pr') \subseteq \mathcal V$, then $\mathcal V; [a]r \ueq [a]s, Pr' \simpto{} \mathcal V; r \ueq s, Pr'$ by \rulefont{{\ueq}[a]}.
By Definitions~\ref{defn.fV} and~\ref{defn.fV.Pr}, $\f{fV}(r \ueq s, Pr') \subseteq \mathcal V$, and the result follows.
\item
The case \rulefont{{\ueq}[b]}.\quad
Suppose $\mathcal V; [a]r \ueq [b]s, Pr'$, $b \not\in \f{fa}(r)$ with $\f{fV}([a]r \ueq [b]s, Pr') \subseteq \mathcal V$, then $\mathcal V; [a]r \ueq [b]s, Pr' \simpto{} \mathcal V; (b\ a) \act r \ueq s, Pr'$ by \rulefont{{\ueq}[a]}.
By Definitions~\ref{defn.fV} and~\ref{defn.fV.Pr} it follows that $\f{fV}((b\ a) \act r) \subseteq \mathcal V$, and the result follows.
\end{squishitem}
\end{proof}

\noindent
Proof of Proposition~\ref{prop.aeq.transitive.g}:

\begin{proof}
We prove reflexivity by induction on terms; symmetry by induction on derivations; transitivity by induction on the size of a term.
We include one case from the proof of symmetry, and one from the proof of transitivity:
\begin{squishitem}
\item
$\lam{a}g \aeq \lam{b}h$ implies $\lam{b}h \aeq \lam{a}g$.\quad
Suppose $(b\ a) \act g \aeq h$ and $b \not\in \f{fa}(g)$.
By Lemma~\ref{lemm.pi.ftma.g}, $a \not\in \f{fa}((b\ a) \act h)$.
By Lemmas~\ref{lemm.pi.aeq.g} and~\ref{lemm.permutation.comp.g}, $g \aeq (b\ a) \act h$.
By Lemma~\ref{lemm.aeq.fa.g}, $a \not\in \f{fa}(h)$.
By inductive hypothesis $(b\ a) \act h \aeq g$.
Extending with \rulefont{{\lambda}{\aeq}{\lambda}[b]} we obtain $\lam{b}h \aeq \lam{a}g$, as required.
\item
$\lam{a}g \aeq \lam{b}h$ and $\lam{b}h \aeq \lam{c}k$ implies $\lam{a}g \aeq \lam{c}k$.\quad
Suppose $(b\ a) \act g \aeq h$, $(c\ b) \act h \aeq k$, $b \not\in \f{fa}(g)$ and $c \not\in \f{fa}(h)$.
By Lemma~\ref{lemm.permutation.comp.g}, $(c\ b) \act ((b\ a) \act g) \aeq (c\ b) \act h$.
By Lemmas~\ref{lemm.pi.depth.invariant.g},~\ref{lemm.permutation.comp.g} and~\ref{lemm.ds.pi.g} we have $(c\ a) \act g \aeq k$.
By Lemma~\ref{lemm.aeq.fa.g}, $c \not\in \f{fa}((b\ a) \act g)$. 
By Lemma~\ref{lemm.pi.aeq.g}, $c \not\in \f{fa}(g)$.
We use \rulefont{{\lambda}{\aeq}{\lambda}ab} to obtain $\lam{a}g \aeq \lam{c}k$, and the result follows.
\end{squishitem}
\end{proof}

\noindent
Proof of Lemma~\ref{lemm.capturable.dom.pi}:

\begin{proof}
By induction on $r$.
\begin{squishitem}
\item
The cases $a$ and $\tf{f}(r_1, \ldots, r_n)$ are routine.
\item
The case $[a]r$.\quad
We reason as follows:
\begin{calcenv}
\f{capt}_A(\pi \act [a]r) & = & \f{capt}_A([\pi(a)](\pi \act r)) & \text{Definition~\ref{defn.perm}} \\
                          & = & \f{capt}_{A \cup \{ \pi(a) \}}(\pi \act r) & \text{Definition~\ref{defn.capturable}}                     
\end{calcenv}
There are now two cases:
\begin{squishitem}
\item
The case $\pi(a) = a$.\quad
Then:
\begin{calcenv}
\f{capt}_{A \cup \{ \pi(a) \}}(\pi \act r) & = & \f{capt}_{A \cup \{ a \}}(\pi \act r) & \text{Assumption} \\
                                           & = & \f{capt}_{A \cup \{ a \}}(r) & \text{Inductive hypothesis} \\
                                           & = & \f{capt}_A([a]r) & \text{Definition~\ref{defn.capturable}}
\end{calcenv}
The result follows.
\item
The case $\pi(a) \not= a$.\quad
Then:
\begin{calcenv}
\f{capt}_{A \cup \{ \pi(a) \}}(\pi \act r) & = & \f{capt}_{A}(\pi \act r) & \text{Assumption, $\pi(a) \neq a$} \\
                                           & = & \f{capt}_{A}(r) & \text{Inductive hypothesis} \\
                                           & = & \f{capt}_{A}([a]r) & \text{Definition~\ref{defn.capturable}}
\end{calcenv}
The result follows.
\end{squishitem}
\item
The case $\pi' \act X^S$.\quad
We reason as follows:
\begin{calcenv}
\f{capt}(\pi \act (\pi' \act X^S)) & = & \f{capt}_A((\pi \fcomp \pi') \act X^S) & \text{Lemma~\ref{lemm.permutation.comp}} \\
                                   & = & (\f{nontriv}(\pi \fcomp \pi') \cup A) \cap S & \text{Definition~\ref{defn.capturable}} \\
                                   & = & (\f{nontriv}(\pi) \cup \f{nontriv}(\pi') \cup A) \cap S & \text{Fact} \\
                                   & = & (\f{nontriv}(\pi') \cup A) \cap S & \text{Assumption} \\
                                   & = & \f{capt}_A(\pi' \act X^S) & \text{Definition~\ref{defn.capturable}}
\end{calcenv}
The result follows.
\end{squishitem}
\end{proof}

\noindent
Proof of Corollary~\ref{corr.capturable.basic.alpha}:

\begin{proof}
We reason as follows:
\begin{calcenv}
\f{capt}_A([b]r) & = & \f{capt}_{A \cup \{ b \}}(r) & \text{Definition~\ref{defn.capturable}} \\
                 & = & \f{capt}_{A \cup \{ a, b \}}(r) & \text{Lemma~\ref{lemm.capturable.fresh.atom}, $a \not\in \f{fa}(r)$} \\
                 & = & \f{capt}_{A\cup\{ a, b \}}((b\ a)\act r) & \text{Lemma~\ref{lemm.capturable.dom.pi}} \\
                 & = & \f{capt}_{A\cup\{ a \}}((b\ a)\act r) & \text{Lemmas~\ref{lemm.capturable.fresh.atom} and~\ref{lemm.pi.ftma}} \\
                 & = & \f{capt}_{A}([a](b\ a)\act r) & \text{Definition~\ref{defn.capturable}}
\end{calcenv}
The result follows.
\end{proof}

\noindent
Proof of Lemma~\ref{lemm.somewhere.X}:

\begin{proof}
By induction on $r$.
\begin{squishitem}
\item
The cases $a$ and $\tf{f}(r_1, \ldots, r_n)$ are straightforward.
\item
The case $[a]r$.\quad
Suppose $a \in \f{capt}_A([a]r)$.
Then $a \in \f{capt}_{A \cup \{ a \}}(r)$, and by inductive hypothesis $X^S \in \f{fV}(r)$ exists such that $a \in S$.
As $\f{fV}([a]r) = \f{fV}(r)$, the result follows.
\item
The case $[b]r$.\quad
Suppose $a \in \f{capt}_A([b]r)$.
Then $a \in \f{capt}_{A \cup \{ b \}}(r)$, and by inductive hypothesis $X^S \in \f{fV}(r)$ exists such that $a \in S$.
As $\f{fV}([b]r) = \f{fV}(r)$, the result follows.
\item
The $\pi \act X^S$.\quad
The result follows immediately by Definition~\ref{defn.capturable}.
\qedhere
\end{squishitem}
\end{proof}

\noindent
Proof of Lemma~\ref{lemm.capturable.pi.r}:

\begin{proof}
By induction on $r$.
\begin{squishitem}
\item
The cases $a$ and $\tf{f}(r_1, \ldots, r_n)$.\quad
Routine.
\item
The case $[a]r$.\quad
By Definitions~\ref{defn.perm} and~\ref{defn.capturable} 
$$
\f{capt}_A(\pi\act[a]r)=\f{capt}_{A\cup\{\pi(a)\}}(\pi\act r).
$$
By inductive hypothesis 
$$
\f{capt}_{A\cup\{\pi(a)\}}(\pi\act r)\subseteq ((\f{nontriv}(\pi)\cup A\cup\{\pi(a)\})\cap\f{uncapt}(r))\cup\f{capt}(r).
$$
By Definition~\ref{defn.uncapt} ${\f{uncapt}([a]r)=\f{uncapt}(r)\setminus\{a\}}$,
so it suffices to show that 
\begin{multline*}
\bigl((\f{nontriv}(\pi)\cup A\cup\{\pi(a)\})\cap\f{uncapt}(r)\bigr)\cup\f{capt}(r)
\\
\subseteq \bigl((\f{nontriv}(\pi)\cup A)\cap(\f{uncapt}(r)\setminus\{a\})\bigr)\cup\f{capt}([a]r) .
\end{multline*}
By Lemma~\ref{lemm.capturable.A.B} ${\f{capt}(r)\subseteq\f{capt}([a]r)}$.
Therefore, we only need concern ourselves with the `extra $\pi(a)$' on the left, and the `missing $a$' on the right.

We consider cases for the `extra $\pi(a)$':
\begin{squishitem}
\item
Suppose ${\pi(a)\neq a}$.  Then ${\pi(a)\in\f{nontriv}(\pi)}$. 
\item
Suppose ${\pi(a)=a}$ and ${a\not\in\f{uncapt}(r)}$.  Then $\{\pi(a)\}\cap\f{uncapt}(r)=\varnothing$. 
\item
Suppose ${\pi(a)=a}$ and ${a\in\f{uncapt}(r)}$.
Then by Lemma~\ref{lemm.uncapt.to.capt} $a\in\f{capt}([a]r)$.
\end{squishitem}
We consider cases for the `missing $a$':
\begin{squishitem}
\item
Suppose $a\not\in\f{uncapt}(r)$.
Then $\f{uncapt}(r)\setminus\{a\}=\f{uncapt}(r)$.
\item
Suppose $a\in\f{uncapt}(r)$.
Then by Lemma~\ref{lemm.uncapt.to.capt} $a\in\f{capt}([a]r)$.
\end{squishitem}
In all cases, the result follows.
\item
The case $\pi' \act X^S$.\quad 
Then:
\begin{calcenv}
\hspace{-2em}
\f{capt}_A((\pi \fcomp \pi') \act X^S) & = & (\f{nontriv}(\pi \fcomp \pi') \cup A) \cap S & \\
                                       & \subseteq & (\f{nontriv}(\pi) \cup \f{nontriv}(\pi') \cup A) \cap S & \\
                                       & = & (((\f{nontriv}(\pi) \cup A) \setminus \f{nontriv}(\pi')) \cap S) \cup (\f{nontriv}(\pi') \cap S) & \\
                                       & = & (((\f{nontriv}(\pi) \cup A) \setminus \f{nontriv}(\pi')) \cap S) \cup \f{capt}(\pi' \act X^S) & \\
                                       & = & ((\f{nontriv}(\pi) \cup A) \cap (S\setminus\f{nontriv}(\pi'))) \cup \f{capt}(\pi' \act X^S) & \\
                                       & = & ((\f{nontriv}(\pi) \cup A) \cap \f{uncapt}(\pi'\act X^S)) \cup \f{capt}(\pi' \act X^S) &
\end{calcenv}
The result follows.
\end{squishitem}
\end{proof}

\end{document}